\documentclass[11pt]{article}
\usepackage[margin=1in]{geometry}
\usepackage{amsmath,amssymb,amsthm,mathtools}
\usepackage{enumitem,microtype,xcolor}
\usepackage{graphicx,float}
\usepackage{aliascnt}
\usepackage[colorlinks=true,linkcolor=blue,citecolor=blue,urlcolor=blue]{hyperref}
\usepackage[nameinlink,noabbrev,capitalise]{cleveref}
\newtheorem{theorem}{Theorem}[section]
\newaliascnt{lemma}{theorem}
\newtheorem{lemma}[lemma]{Lemma}
\aliascntresetthe{lemma}
\newaliascnt{proposition}{theorem}
\newtheorem{proposition}[proposition]{Proposition}
\aliascntresetthe{proposition}
\newaliascnt{corollary}{theorem}
\newtheorem{corollary}[corollary]{Corollary}
\aliascntresetthe{corollary}
\theoremstyle{definition}
\newaliascnt{definition}{theorem}
\newtheorem{definition}[definition]{Definition}
\aliascntresetthe{definition}
\newaliascnt{conjecture}{theorem}
\newtheorem{conjecture}[conjecture]{Conjecture}
\aliascntresetthe{conjecture}
\theoremstyle{remark}
\newaliascnt{remark}{theorem}

\aliascntresetthe{remark}
\Crefname{theorem}{Theorem}{Theorems}
\Crefname{lemma}{Lemma}{Lemmas}
\Crefname{proposition}{Proposition}{Propositions}
\Crefname{corollary}{Corollary}{Corollaries}
\Crefname{definition}{Definition}{Definitions}
\Crefname{conjecture}{Conjecture}{Conjectures}
\Crefname{remark}{Remark}{Remarks}
\newcommand{\R}{\mathbb R}
\newcommand{\E}{\mathbb E}
\newcommand{\PP}{\mathbb P}
\newcommand{\Cov}{\operatorname{Cov}}
\newcommand{\Var}{\operatorname{Var}}
\newcommand{\tr}{\operatorname{tr}}
\newcommand{\vol}{\operatorname{vol}}
\newcommand{\Unif}{\operatorname{Unif}}
\newcommand{\TV}{\operatorname{TV}}
\newcommand{\Id}{I_n}
\newcommand{\bn}{B_2^n}
\newcommand{\norm}[1]{\left\lVert #1\right\rVert}
\newcommand{\ip}[2]{\langle #1,#2\rangle}
\newcommand{\op}{\mathrm{op}}
\newcommand{\Ot}{\widetilde O}
\newcommand{\Omegat}{\widetilde\Omega}

\numberwithin{equation}{section}
\title{Near-Optimal Oracle Bounds for Isotropic Rounding}

\author{
Yang P. Liu\\
\texttt{yangl7@andrew.cmu.edu}\\
Carnegie Mellon University
\and
Richard Peng\\
\texttt{yangp@cs.cmu.edu}\\
Carnegie Mellon University
\and
Alicia Stepin\\
\texttt{astepin@andrew.cmu.edu}\\
Carnegie Mellon University
\and
Colin Tang\\
\texttt{cstang@andrew.cmu.edu}\\
Carnegie Mellon University
}
\begin{document}
\maketitle
\begin{abstract}
We give optimal bounds, up to polylogarithmic factors, for rounding convex
bodies to near-isotropic position. A convex body
$B(0,r)\subseteq K\subseteq B(0,R)$ in $\R^n$ can be rounded using
$\Ot(n^3)$ membership queries; the upper bound extends to logconcave
distributions. We prove a matching $\Omegat(n^3)$ lower bound, even when
$R/r=n^{O(1)}$. The key lemma states that if $B(0,1)\subseteq K$ and
$\Cov(\Unif(K))\preceq\kappa I_n$, with $\kappa\ge1$, then approximate
uniform sampling from an initial density bounded by a constant times the
uniform density uses
$\Ot(n^2\sqrt\kappa)$ expected membership queries.
We prove this by simulating reflected kinetic dynamics using the analysis
of Eberle and L\"orler~\cite{EL26:journal}.
Combining this sampler with the ideas of Jia, Laddha, Lee, and
Vempala~\cite{JLLV26:journal} for rounding well-rounded bodies yields our
$\Ot(n^3)$-query rounding algorithm.
We also give a counterexample to an ellipsoid-growth conjecture from previous papers on this topic.
\end{abstract}

\section{Introduction}
Isotropic rounding is the problem of finding an affine change of
coordinates that centers a distribution and makes its variance comparable
in every direction. This normalization is used by algorithms for
sampling and volume estimation, whose query bounds depend on the
geometry of the input. For a convex body given by a membership oracle,
the algorithm must find these coordinates using only tests of whether
queried points belong to the body.

We prove nearly cubic upper and lower bounds for the number of queries
needed for isotropic rounding. The upper bound also applies to logconcave
distributions accessed through an evaluation oracle for a convex potential.
The key idea for the upper bound is to combine improved sampling based on
the relaxation estimate of Eberle and L\"orler~\cite{EL26:journal} with the ideas
of Jia, Laddha, Lee, and Vempala for rounding well-rounded
bodies~\cite{JLLV26:journal}. We implement reflected kinetic dynamics using
membership queries and use the resulting sampler in the rounding algorithm.
We first state the results, then explain this combination and its
consequence for volume estimation.

\subsection{Main results}
We specify the affine normalization and give matching bounds up to
polylogarithmic factors for convex bodies. We then state the extension
to logconcave distributions.

For a full-dimensional probability measure $\mu$ on $\R^n$ with finite
second moments and $X\sim\mu$, write $m_\mu=\E X$ and
$A_\mu=\Cov(X)=\E[(X-m_\mu)(X-m_\mu)^\top]$.
The measure is isotropic if $m_\mu=0$ and $A_\mu=I_n$, where $I_n$ is the
identity matrix. We ask for an approximate normalization with fixed
constants. In the definition below, $A\preceq B$ means that $B-A$ is
positive semidefinite, and $\|x\|_2=(\sum_i x_i^2)^{1/2}$ is the
Euclidean norm.

\begin{definition}[Near-isotropic rounding]\label{def:rounding}
A pair $(T,z)$, with $T\in\R^{n\times n}$ invertible and $z\in\R^n$,
rounds $\mu$ if $Y=T(X-z)$, for $X\sim\mu$, satisfies
\begin{equation}\label{eq:rounding}
 \tfrac12\Id\preceq\Cov(Y)\preceq2\Id
 \text{ and }\norm{\E Y}_2\le\tfrac12.
\end{equation}
\end{definition}

For a convex body $K$, the target is its uniform probability measure
$\Unif(K)$. In addition to the membership oracle, the algorithm is given
radii $0<r\le R$ satisfying $r\bn\subseteq K\subseteq R\bn$, where
$B_2^n=\{x\in\R^n:\|x\|_2\le1\}$ is the Euclidean unit ball.

Our upper bound is $\Ot(n^3)$ queries, where
$\Ot(f)$ denotes $O(f\log^C(2+nR/r))$ for some absolute constant $C>0$.
We use $\Omegat(f)$ for $\Omega(f/\log^C(2+nR/r))$.
Thus the upper bound may be multiplied by a fixed power of a logarithm,
and the lower bound may be divided by one.
\Cref{sec:preliminaries} gives the convention for auxiliary parameters.

The first theorem applies to every convex body satisfying the supplied
inner and outer radius bounds.
\begin{theorem}[Convex-body upper bound]\label{thm:upper-body}
There is a randomized membership-oracle algorithm that, given
$r\bn\subseteq K\subseteq R\bn$, returns a pair $(T,z)$ satisfying
\eqref{eq:rounding} with probability at least $5/6$ using
$\Ot(n^3)$ queries.
\end{theorem}
The proof is in \Cref{subsec:gaussian-cooling-upper}.
The same $\Ot(n^3)$ upper bound extends to logconcave distributions given
by a convex-potential evaluation oracle, a known ball in a level set, and
a second-moment bound; \Cref{thm:upper-logconcave} states the precise
assumptions in \Cref{sec:upper-bound}.

The lower bound concerns the same affine output, including its centering
requirement. It holds even for bodies with polynomial aspect ratio.
\begin{theorem}[Membership lower bound]\label{thm:lower-main}
For all sufficiently large $n$, there are known radii $r_n,R_n$ with
$R_n/r_n=n^{O(1)}$ such that every randomized algorithm satisfying
\eqref{eq:rounding} with probability at least $2/3$ on every promised body
requires $\Omegat(n^3)$ membership queries in worst-case expectation.
\end{theorem}
\Cref{sec:lower-bound} proves the explicit
$\Omega(n^3/\log^2 n)$ bound and gives the radii in
\Cref{thm:lower-rounding}.
Since the upper-bound algorithm succeeds with probability $5/6>2/3$,
this lower bound applies to its success guarantee as well.

\subsection{Relation to sampling and volume}
We compare our bounds with previous rounding and sampling algorithms,
then state the consequence for volume estimation.

Dyer, Frieze, and Kannan established randomized polynomial-time volume
approximation from oracle access~\cite{DFK91:journal}.
Lov\'asz and Simonovits developed the random-walk analysis underlying
subsequent improvements~\cite{LS93:journal}.
The sampling and geometric theory extends to logconcave
densities~\cite{LV07:journal}.

Previous rounding algorithms distinguish general inputs from
well-rounded bodies. Here a well-rounded body has a known contained unit
ball and $\mathbb E\|X\|_2^2=O(n)$ for $X\sim\Unif(K)$, after scaling
and translation.
Jia, Laddha, Lee, and Vempala give a nearly cubic algorithm for rounding
a well-rounded convex body and a general rounding bound
$\Ot(n^{7/2})$ using polylogarithmic bounds for the
Kannan--Lov\'asz--Simonovits (KLS)
constant~\cite{JLLV26:journal,K23:journal}.
Kook and Vempala obtain $\Ot(n^{7/2})$ rounding for general logconcave
distributions~\cite{KV25}.
We believe that these approaches may be hard to improve--Jia, Laddha, Lee, and
Vempala conjectured that doubling an ellipsoid whose intersection with the body
is nearly isotropic would leave the intersection well-rounded.
We disprove this conjecture with the counterexample in
\Cref{sec:symmetric-ellipsoid-growth-counterexample}: even for a symmetric body,
the intersection can start as isotropic, where the trace of the covariance is
$\Omega(n^2)$ after the enlargement.
For sufficiently large $n$, the enlarged intersection has covariance
$A\npreceq2I_n$, so it no longer satisfies \Cref{def:rounding} in those
coordinates.
This rules out their proposed way of taking larger ellipsoid-growth steps.
Our Gaussian-cooling algorithm instead rounds each new target under a weaker
moment condition, and does not assume that the target remains well-rounded after
cooling.

The sampling cost determines the query bound in these rounding algorithms.
A starting law $\nu$ is $M$-warm with respect to a target law
$\mu$ if $\nu(A)\le M\mu(A)$ for every measurable set
$A$; a constant-warm start has $M=O(1)$.
Thus its density is at most a fixed constant times the target density.
From such a start, Kook and Vempala obtain
$\Ot(n^2\max\{1,\|A_\mu\|_\op\})$ queries under their level-set
assumptions~\cite{KV25}, where $\|A\|_\op$ is the operator norm of $A$.
Their guarantee is in R\'enyi divergence; the
In-and-Out algorithm gives the corresponding covariance dependence for
convex bodies~\cite{KVZ24}.
Kook and Vempala also show that sampling from an arbitrary logconcave distribution from a cold start requires just $\Ot(n^{5/2})$ membership queries~\cite{KV26:arxiv}.

Our sampling bound is $\Ot(n^2\sqrt\kappa)$ expected membership queries
from a constant-warm start when $B_2^n\subseteq K\subseteq RB_2^n$ and
$\Cov(\Unif(K))\preceq\kappa I_n$, with $\kappa\ge1$.
It guarantees any prescribed total variation error
$\varepsilon\in(0,1/2)$; see \Cref{thm:sampler}.
The improvement is the square-root dependence on the covariance bound.
Combining this sampler with the rounding method of Jia, Laddha, Lee, and
Vempala~\cite{JLLV26:journal} and Gaussian cooling gives the nearly cubic
bound for general inputs.
\Cref{subsec:overview_upper} explains this combination.

For volume estimation, Cousins and Vempala use
$\Ot(n^3\varepsilon^{-2})$ membership queries to obtain relative error
$\varepsilon\in(0,1)$ when the body is well-rounded~\cite{CV18:journal}.
Combining their algorithm with \Cref{thm:upper-body} gives the same query
bound for every convex body satisfying $rB_2^n\subseteq K\subseteq RB_2^n$;
\Cref{cor:volume-reduction} gives the precise guarantee and reduction.
Rademacher and Vempala proved an $\Omegat(n^2)$ membership-query lower
bound for randomized volume approximation with sufficiently small
constant relative error, even on parallelotopes with polynomial
aspect ratio~\cite{RV08:journal}.

After the standard facts in \Cref{sec:preliminaries},
\Cref{sec:overview} explains the upper- and lower-bound arguments.
\Cref{sec:upper-bound,sec:lower-bound} give their proofs.
\Cref{sec:rounding-support} supplies the supporting rounding estimates,
\Cref{sec:symmetric-ellipsoid-growth-counterexample} gives the
ellipsoid-growth counterexample, and \Cref{sec:kinetic} proves the sampler.

\section{Preliminaries}
\label{sec:preliminaries}
We record notation and the logconcave facts used throughout the paper.
Other facts are introduced where they are needed.

\paragraph{Asymptotic notation.}
For auxiliary bounds, let $p_1,\ldots,p_k>0$ be the scalar input
parameters appearing in the statement, apart from the dimension.
They may include geometric and moment bounds, warmness, the number of
requested samples, and accuracy or failure probabilities.
Here $k$ is the number of such parameters, and $j$ indexes them.
The notation $\Ot(f)$ may suppress a factor
\[
 \log^C\left(2+n+\sum_{j=1}^k(p_j+p_j^{-1})\right).
\]
Polynomial dependence on any input is always displayed explicitly.
For example, in \Cref{thm:sampler} the hidden factor may depend
logarithmically on the outer radius, warmness, and reciprocal accuracy.
If the known contained ball has radius $r$, we first measure lengths in
units of $r$ and second moments in units of $r^2$.
Thus an outer-radius bound becomes $R/r$, and a second-moment bound
$S$ becomes $S/r^2$.
This is only a change of units; any resulting polynomial dependence is
still displayed in the stated bound.
Implicit constants may depend on parameters declared fixed.

We write $\mathcal L(X)$ for the law of a random variable $X$ and
$\operatorname{supp}\mu$ for the support of a measure $\mu$.
The symbols $c,C$ denote positive universal constants whose values may
change between occurrences. A subscript permits dependence on the
indicated parameter.
The Poincar\'e constant $C_P(\mu)$ is the least $C$ such that
\(\Var_\mu f\le C\E_\mu\norm{\nabla f}_2^2\)
for smooth test functions, and hence by closure for the corresponding
Sobolev space. Polylogarithmic KLS bounds imply
\begin{equation}\label{eq:kls}
 C_P(\mu)\le C\log(2+n)\norm{A_\mu}_\op
\end{equation}
for logconcave $\mu$~\cite{K23:journal}.

Recall that $M$-warmness means $\nu\le M\mu$.
For $f=d\nu/d\mu$, its chi-squared divergence is
$\chi^2(\nu\|\mu)=\int(f-1)^2\,d\mu$; warmness implies
$\chi^2(\nu\|\mu)\le M-1$.
We use $\|\nu-\mu\|_{\TV}=\sup_A|\nu(A)-\mu(A)|$ for total
variation distance, where the supremum is over measurable sets.
For any density $f=d\nu/d\mu$, Cauchy--Schwarz gives
$\|\nu-\mu\|_{\TV}\le\tfrac12\|f-1\|_{L^2(\mu)}$.

A logconcave distribution contains its covariance ellipsoid in its support.
For the uniform measure on a convex body, this gives a known interior ball
once the mean and covariance have been approximately normalized.

\begin{lemma}\label{lem:support-ball}
Let $\mu$ be a full-dimensional logconcave probability measure with mean
$m$ and covariance $A$. Then
\[
 m+A^{1/2}\bn\subseteq\operatorname{supp}\mu.
\]
\end{lemma}
\begin{proof}
Fix a unit vector $u$ for which $\ell=\inf_{x\in\operatorname{supp}\mu}
\ip{u}{x}$ is finite. The nonnegative random variable
$W=\ip{u}{X}-\ell$ has a logconcave density. Its survival function
$S(t)=\PP(W\ge t)$ is logconcave by the marginal theorem for logconcave
functions~\cite{LV07:journal}, and $S(0)=1$.
Concavity of $\log S$ implies $S(s+t)\le S(s)S(t)$ for $s,t\ge0$.
Consequently,
\[
 \tfrac12\E W^2=\int_0^\infty tS(t)\,dt
 =\int_0^\infty\int_0^\infty S(s+t)\,ds\,dt
 \le (\E W)^2.
\]
Thus $\Var W\le(\E W)^2$, or
$\ip{u}{m}-\ell\ge\sqrt{u^\top Au}$. Every supporting halfspace of
$\operatorname{supp}\mu$ therefore contains $m+A^{1/2}\bn$.
\end{proof}
In particular, a body rounded as in \Cref{def:rounding} contains
$B(0,2^{-1/2}-1/2)$: its covariance ellipsoid has radius at least
$2^{-1/2}$ and its centroid is at distance at most $1/2$ from zero.
\Cref{cor:volume-reduction} uses this ball to apply the well-rounded
volume algorithm below.
\Cref{lem:level-covariance-ellipsoid} uses the same lemma for a
logconcave law conditioned on a level set, supplying the ellipsoid
needed in the subspace-growing and cooling steps.

\paragraph{Volume estimation for well-rounded bodies.}
\label{par:well-rounded-volume}
Cousins and Vempala~\cite{CV18:journal} give a membership-oracle
algorithm that, for a body containing the unit ball with second moment
$O(n)$, estimates its volume to relative error
$\varepsilon\in(0,1)$ using $\Ot(n^3\varepsilon^{-2})$ queries,
with any fixed constant success probability.
The cost is $\Ot(n^3)$ for constant relative error.
\Cref{cor:volume-reduction} checks these geometric assumptions after
applying our rounding map.

\paragraph{Logconcave marginals and moments.}
Integrating out coordinates preserves logconcavity~\cite{LV07:journal}.
We use this both for one-dimensional projections and for the position
marginal of the epigraph distribution in \Cref{sec:kinetic}.
For a real logconcave random variable $Z$ and $1\le p\le q$,
the standard moment comparison~\cite{LV07:journal} gives
\begin{equation}\label{eq:logconcave-moments}
 \left(\E|Z|^q\right)^{1/q}
 \le C\frac qp\left(\E|Z|^p\right)^{1/p}.
\end{equation}
In particular, its fourth moment is at most a universal constant times
the square of its second moment.
We use this in the cooling and moment-estimation arguments of
\Cref{sec:rounding-support}; the general comparison is also used in
\Cref{lem:nuclear-transfer}.

\paragraph{Gaussian tail bounds.}
Write $\varphi(t)=(2\pi)^{-1/2}e^{-t^2/2}$ for the standard normal
density and $\overline\Phi(t)=\int_t^\infty\varphi(u)\,du$ for its
upper tail.
The elementary Mills bounds are
\begin{equation}\label{eq:mills-bounds}
 \frac{t}{1+t^2}\varphi(t)
 \le\overline\Phi(t)
 \le\frac{\varphi(t)}{t},
 \qquad t>0.
\end{equation}
Indeed, integration by parts gives
\[
 \overline\Phi(t)
 =\frac{\varphi(t)}{t}
   -\int_t^\infty\frac{\varphi(u)}{u^2}\,du.
\]
The integral lies between zero and
$t^{-2}\overline\Phi(t)$, giving both inequalities.
We use these bounds in \Cref{subsec:lower-moments} to estimate the
maximum of the Gaussian facet projections.

\paragraph{Brascamp--Lieb variance inequality.}
For a probability density on $\R^n$ proportional to $e^{-W}$, where
$W$ is twice continuously differentiable and
$\nabla^2W$ is positive definite, the Brascamp--Lieb
inequality~\cite{BL76:journal} states that
\begin{equation}\label{eq:brascamp-lieb}
 \Var_\mu f
 \le\E_\mu\left[
   \nabla f^{\mathsf T}(\nabla^2W)^{-1}\nabla f\right]
\end{equation}
for smooth functions with finite variance and finite right-hand side.
We use its strongly logconcave consequence: if
$W(x)-\alpha\|x\|_2^2/2$ is convex for $\alpha>0$, then
\begin{equation}\label{eq:strong-logconcave-poincare}
 \Var_\mu f
 \le\alpha^{-1}\E_\mu\|\nabla f\|_2^2.
\end{equation}
This consequence also holds for proper closed potentials that can take
the value $+\infty$, including restriction to a convex set, by convex
approximation.
Linear functions give $\Cov(\mu)\preceq\alpha^{-1}I_n$.
We use this inequality for the conditional Gaussians in
\Cref{lem:lower-conditional-mean}, for cooling in
\Cref{lem:fine-cooling-renyi}, and for the position updates in
\Cref{subsec:kinetic-mixing}.

\section{Overview}
\label{sec:overview}
We first describe the $\Ot(n^3)$-query isotropic rounding algorithm, then we
proceed with the matching lower bound.

\subsection{The cubic rounding algorithm}
\label{subsec:overview_upper}

We outline the algorithm in \cref{thm:upper-body}.
It is known how to estimate the volume of a well-rounded body to constant
relative error using $\Ot(n^3)$ membership queries~\cite{CV18:journal}
(see \Cref{sec:preliminaries}), so one
approach to estimating the volume of an arbitrary body is first to make it
well-rounded
and then estimate its volume.
One way of making an arbitrary body well-rounded is to estimate its mean and
covariance and use them to bring it into nearly-isotropic position, as our
algorithm does.
This is also the approach taken in~\cite{JLLV26:journal}, which applies the
Cousins--Vempala volume algorithm~\cite{CV18:journal} after rounding.

We note that standard results show that sampling a point from a nearly-isotropic
body with a known contained ball of constant radius can be done using $\Ot(n^2)$
membership queries from a constant-warm start~\cite{KVZ24}.

For any body containing the unit ball, the
standard ball walk uses $\Ot(n^2\kappa)$ queries per sample from a
constant-warm start when the covariance satisfies $A\preceq\kappa I$, with
$\kappa\ge1$; see
\cite[Theorem~2.7]{JLLV26:journal} and the Poincar\'e bound \cref{eq:kls}.
We improve this to $\Ot(n^2\sqrt\kappa)$, and so in the case of the
nearly-isotropic body, we get a quadratic bound as $\kappa=O(1)$.

Our goal is to learn the covariance of the distribution by starting from an
easily sampled distribution concentrated near the center of the in-ball and
gradually weakening the Gaussian factor until the distribution approximates the
uniform distribution, see also~\cite{CV18:journal,KV25}.

For a convex body, we let the distribution at parameter $t$ be $\mu_t(x)\propto
e^{-t\|x-x_0\|_2^2/2}\mathbf 1_K(x)$, where $x_0$ is the center of the known
in-ball.
Initially $t$ is large, so the Gaussian is concentrated inside this ball and is
easy to sample by rejection.
We then repeatedly halve $t$ and estimate the mean and covariance of the new
distribution.
When $tR^2$ is sufficiently small, these estimates give a rounding of the
uniform distribution on $K$.
This family of restricted Gaussian distributions comes from the cooling approach
of Cousins and Vempala~\cite{CV18:journal}.

We differ from~\cite{JLLV26:journal} in our treatment of the outer rounding
loop: we halve $t$ at each cooling phase.
In contrast, \cite[Algorithm~1]{JLLV26:journal} considers $K\cap B(0,u)$ and
increases $u$ by a factor of $1+1/\sqrt n$, keeping the next body well-rounded
in the previous
coordinates, which ensures $\tr A+\|m\|_2^2=\Ot(n)$, and hence
$\|A\|_\op=\Ot(n)$.
This inner rounding procedure costs $\Ot(n^3)$, and the small outer steps
require $\Ot(\sqrt n)$ calls~\cite{JLLV26:journal}.

The counterexample in
\Cref{sec:symmetric-ellipsoid-growth-counterexample} shows that for a
symmetric convex body $K$ and a centered ellipsoid $E$, the intersection
$K\cap E$ can be isotropic while the trace of the covariance of $K\cap2E$ is $\Omega(n^2)$.
Thus, for sufficiently large $n$, the enlarged intersection fails the
covariance upper bound $A\preceq2I_n$ in \Cref{def:rounding} in the same
coordinates.
This also disproves the weaker $\Ot(n)$ trace bound conjectured
in~\cite{JLLV26:journal}.

In contrast, our larger cooling step means that we can't ensure the new
distribution is well-rounded, so we instead make do with a weaker guarantee.
Concretely, letting $A$ and $m$ be the covariance and mean of $\mu_{t/2}$ in the
coordinates that made $\mu_t$ nearly isotropic, we have the weaker bound
$ \tr A+\|m\|_2^2=\Ot(n^2)$, as per \cref{lem:cooling-transition}.
This bound applies to the largest variance as well, so we can start off with a
worst-case $\kappa = \Ot(n^2)$ at the start of each cooling phase.

We maintain a certified ball $B(z_i,\rho_i)$ contained in a level set of the
current potential, expressed in the current coordinates.
After rescaling this ball to unit radius,
we obtain a distribution with covaraince at most $\Ot(n^2) \cdot I$.
Then we sample with a particle that moves in straight lines inside the body
(or in the logconcave case, the epigraph of the potential),
reflects off the boundary,
and refrehses its velocity and take small Gaussian steps timed by a Poisson clock.
This process allows us to use a quantititave space-time divergence lemma
by Eberle and L\"{o}rler~\cite{EL26:journal} that bounds the relaxation time
of such reflected kinetic processes by the square root of the Poincar\'e constant,
rather than by the Poincar\'e constant itself as for the ball walk.
This constant is also $\Ot(n^2)$ by the KLS bound (\cref{eq:kls}),
thus giving the factor $\Ot(n)$ improvement.

Through the iterations within this phase, we improve the bound on $\kappa$,
making each sample cheaper, and take more samples to improve our estimate.
We do this using the updating scheme from~\cite{JLLV26:journal}: we make a
coarse estimate of the covariance and then stretch directions with small
estimated variance by a
factor of two.
The point of the coarse estimate is that we can maintain the desired query bound
by taking fewer samples while sampling is still expensive, with the goal of
identifying directions with small variance.
As sampling becomes cheaper, we take larger batches, ending with the $\Ot(n)$
samples needed for the final mean and covariance estimates.
Then, after the cooling we will also have we also have
$\tr\sqrt A+\|m\|_2=\Ot(n)$, and this stronger bound combined with the increased
number of possible samples gives the desired $\Ot(n^3)$ queries per phase;
see \cref{lem:cooling-transition,prop:logconcave-inner-rounder}.

For a Gaussian weight $e^{-t\|x-x_0\|_2^2/2}$, the length scale in each
direction is $t^{-1/2}$.
Within each cooling phase, we keep $t$ fixed and write the current coordinates
as $y=F_i x$.
If $P_i$ projects onto the directions of small estimated variance, the update is
$F_{i+1}=(I+P_i)F_i$, which stretches these directions by a factor of two.
The certified ball grows in these transformed coordinates, and its center is
updated as well.

We perform this halving process for $O(\log(2+nR/r))$ cooling phases, each
costing $\Ot(n^3)$ membership queries.
This gives the nearly cubic rounding bound.

\subsection{The rounding lower bound}
\label{subsec:overview_lower}

The crux of the lower bound is to hide a direction in the
center of mass of a random body and show that membership queries reveal little
information about this direction.

For volume estimation itself, the lower bound is (up to polylog factors)
about $n^2$ by Rademacher and Vempala~\cite{RV08:journal};
recently such bounds were extended to uniform sampling and linear
optimization by Vempala~\cite{V26:arxiv}.

We introduce a random body by taking the convex hull of a suitably chosen
polynomial number $N$ of independent points $Z_i\sim N(0,I_n)$, translating it
by an independent vector $\Xi\sim N(0,\tau^2I_n)$ with $\tau^2$ of order
$\log n/n^2$, and then taking the polar:
\[
 K_{\Xi,Z}=\{x:\langle Z_i+\Xi,x\rangle\le1
                    \text{ for every }i\}.
\]
With high probability, the body fits between two known concentric balls whose
radii differ by only a polynomial factor in $n$.
A uniform point in the body has almost the same variance in every direction.
Take $s$ to be the corresponding standard deviation, which is of order
$1/\sqrt{n\log n}$.
The center of mass $m$ is about $4s$ away from the origin, in the direction
opposite to $v=\Xi/\|\Xi\|_2$.
These properties are proved in
\cref{lem:lower-geometry,prop:lower-typical-moments}.

We show that any algorithm needs at least $\Omega(n^3/\log^2 n)$ membership
queries to return a unit vector $\widehat v$ with $\langle\widehat
v,v\rangle\ge1/2$ with constant probability.
The answers to $q$ adaptive queries contain at most $O(q\log^2 n/n^2)$
information about the hidden direction, whereas recovering it with constant
accuracy requires $\Omega(n)$ information.
This information bound holds even if the membership oracle also reveals the
first violated facet.
Since a successful rounding gives such a direction, this proves the lower bound.

\section{A Cubic Algorithm for Second-Moment Rounding}
\label{sec:upper-bound}

We prove \Cref{thm:upper-logconcave,thm:upper-body} by combining our
membership sampler, based on the relaxation estimate of Eberle and
L\"orler~\cite{EL26:journal}, with the subspace-growing method of Jia, Laddha,
Lee, and Vempala for rounding well-rounded bodies~\cite{JLLV26:journal}.
We first state the sampling and moment-estimation results, whose proofs
are in \Cref{sec:kinetic,subsec:upper-moments}.
The algorithmic argument then has two parts. Subspace growing gives a
cubic query bound when $\operatorname{tr}\sqrt A+\|m\|_2=\widetilde O(n)$,
where $m$ and $A$ are the input mean and covariance. Gaussian cooling
reduces each subsequent rounding call to this case. We prove both parts
here, using the supporting estimates from \Cref{sec:rounding-support}.
We treat logconcave distributions throughout; indicator potentials give
the convex-body result.

For logconcave distributions, we replace the membership oracle by an
evaluation oracle for a convex potential.
Specifically, let $d\mu\propto e^{-V}\,dx$, where
$V:\R^n\to\R\cup\{+\infty\}$ is proper, closed, and convex, with
$0<\int e^{-V}<\infty$.
A query at $x$ returns $V(x)$.
The geometric input consists of a known point $x_0$, radii
$0<r\le R$, and a fixed constant $h\ge1$ such that, for
$a=\inf V$,
\begin{equation}\label{eq:grounded-input}
 B(x_0,r)\subseteq\{V\le a+hn\}
 \text{ and }\E_\mu\norm{X-x_0}_2^2\le R^2,
\end{equation}
where $B(x_0,r)=x_0+rB_2^n$.
Thus the algorithm knows a ball in a specified level set and an upper
bound on the second moment about its center.
These assumptions give the same query upper bound as for convex bodies.

\begin{theorem}[Logconcave upper bound]\label{thm:upper-logconcave}
For every fixed $h\ge1$, there is a randomized algorithm that, under
Equation~\eqref{eq:grounded-input}, returns a pair $(T,z)$ satisfying
Equation~\eqref{eq:rounding} with probability at least $5/6$ using
$\Ot(n^3)$ evaluations of $V$.
\end{theorem}
The proof is in \Cref{subsec:gaussian-cooling-upper}.
The indicator potential $V=0$ on $K$ and $V=+\infty$ off
$K$ recovers \Cref{thm:upper-body}.

\subsection{Sampling and moment estimation}\label{subsec:rounding-sampling}
We state the improved sampling bound and then deduce the
moment-estimation guarantee used by subspace growing. The proof of the
sampling theorem is in \Cref{sec:kinetic}: it combines the relaxation
estimate of Eberle and L\"orler~\cite{EL26:journal}, the Poincar\'e bound
\eqref{eq:kls}, and a membership-query implementation of reflected
kinetic dynamics.
\begin{theorem}[Membership sampling]\label{thm:sampler}
Let $\bn\subseteq K\subseteq R\bn$, let $\pi=\Unif(K)$, and suppose
$\Cov(\pi)\preceq\kappa\Id$ with $\kappa\ge1$. Given
$\varepsilon\in(0,1/2)$ and an initial sample $X_0\sim\nu$ with
$\nu\le M\pi$ for a fixed constant $M$, an algorithm uses
$\Ot(n^2\sqrt\kappa)$ expected membership queries and returns $X$ with
$\|\mathcal L(X)-\pi\|_{\TV}\le\varepsilon$.
\end{theorem}

\Cref{thm:logconcave-warm-membership-sampling} gives the logconcave
extension, including a joint approximation of the sample sequence by
positions of a stationary Markov chain. Combining this guarantee with
matrix concentration for Markov chains~\cite{NSW24:journal} and
standard logconcave moment estimates gives the following lemma.
Its proof, including the effect of the sampler's truncation, is in
\Cref{subsec:upper-moments}.
\begin{lemma}[Moment estimation]\label{lem:rounding-estimation-tools}
Let $\mu(dx)\propto e^{-W(x)}dx$ be a full-dimensional logconcave
probability measure with mean $m$ and covariance $A$, where $W$ is
proper, closed, and convex. Assume
$B(z,\rho)\subseteq\{W\le\inf W+Hn\}$ and
$A\preceq\kappa\rho^2I$, where $z\in\mathbb R^n$, $\rho>0$, and
$\kappa\ge1$ are known, and $H\ge1$ is fixed.
For $a,b>0$ and $0<\delta<1/4$, set
$k=\widetilde O(1+\operatorname{tr}(A)(a^{-1}+b^{-2}))$.
Given an evaluation oracle for $W$, a constant-warm initial sample,
and upper bounds on $\operatorname{tr}A$ and $\mathbb E_\mu\|X-z\|_2^2$,
there is a randomized algorithm that uses
$\widetilde O(n^2k\sqrt\kappa)$ evaluations in expectation and outputs
a vector $\widehat m\in\mathbb R^n$ and a matrix $\widehat A\succeq0$
such that, with probability at least $1-\delta$,
\[
 \tfrac12A-aI\preceq\widehat A\preceq2A+aI
 \text{ and }\|\widehat m-m\|_2\le b.
\]
Alternatively, an algorithm using $\widetilde O(n^3\sqrt\kappa)$
expected evaluations outputs an estimate $\widehat A$ satisfying
$\tfrac78A\preceq\widehat A\preceq\tfrac98A$ and
$\|A^{-1/2}(\widehat m-m)\|_2\le1/16$ with probability at least $1-\delta$.
\end{lemma}
The lemma assumes a constant-warm initial sample. To provide one up to
a prescribed total variation error,
\Cref{prop:logconcave-fresh-initialization} uses at most
$\widetilde O(n^{5/2}+n^2\sqrt{\mathbb E_\mu\|X-z\|_2^2}/\rho)$
evaluations and returns $\widehat X$ with
$\|\mathcal L(\widehat X)-\nu\|_{\TV}\le\delta$ for some $\nu\le C\mu$.
The second moment here is
$\operatorname{tr}A+\|m-z\|_2^2$.
We apply this construction after fixing the coordinates of each
estimation batch. Its cost and the error in approximating a constant-warm
law are included in the proof below.
Whenever moment quantities are used to set algorithm parameters, the
algorithm is given those quantities or upper bounds substituted in both
the parameters and query bounds.

\subsection{Rounding by subspace growing}
\label{subsec:subspace-growing}

We prove \Cref{prop:logconcave-inner-rounder} by combining
\Cref{lem:rounding-estimation-tools} with the subspace-growing method
of Jia, Laddha, Lee, and Vempala~\cite{JLLV26:journal}, extended to logconcave
distributions by Kook and Vempala~\cite{KV25}. At each stage, we estimate
the covariance and double directions with small estimated variance.

For convex bodies, the nearly cubic guarantee of
\cite[Theorem 1.3]{JLLV26:journal} assumes
$\mathbb E\|X\|_2^2=\operatorname{tr}A+\|m\|_2^2=\widetilde O(n)$.
The following bound instead gives nearly cubic rounding under
$\operatorname{tr}\sqrt A+\|m\|_2=\widetilde O(n)$, the weaker condition
that our cooling argument will supply. Initialization is included in
the query bound. The output tolerances leave room to remove the final
Gaussian factor.
\begin{proposition}[Rounding within a level set]
\label{prop:logconcave-inner-rounder}
Let $W$ be a proper closed convex potential, given by an evaluation
oracle, such that $\mu(dx)\propto e^{-W(x)}dx$ is a full-dimensional
probability measure with mean $m$ and covariance $A$. Let $X\sim\mu$.
Suppose $B_2^n\subseteq\{W\le\inf W+Hn\}$ for a fixed $H\ge1$.
For every $0<\delta<1/4$, an algorithm uses
$\Ot(n^3+n^2(\operatorname{tr}\sqrt A+\|m\|_2))$ evaluations and returns
an affine map $F$ such that, with probability at least $1-\delta$,
\[
 \tfrac34I\preceq\Cov(F(X))\preceq\tfrac43I
 \text{ and }\|\E [F(X)]\|_2\le\tfrac18.
\]
\end{proposition}

We first describe the algorithm, then prove that its updates preserve a
known contained ball. We finish the proof of
\Cref{prop:logconcave-inner-rounder} by bounding the covariance growth,
the sampling cost, and the error of the final normalization.

\paragraph{Algorithm.}
Increase $H$, if necessary, to the constant in
\Cref{lem:level-covariance-ellipsoid}, and fix
$C=\{W\le\inf W+Hn\}$.
Set $\ell=\lceil\log_2(2+\operatorname{tr}\sqrt A)\rceil$ and
$\eta=(2000\ell)^{-1}$.
Throughout this algorithm, $A$ is the covariance of the original input.
The current covariance is $A_i=F_iAF_i^{\mathsf T}$ and its mean is
$m_i=F_im$.
Start with $F_0=I$, $z_0=0$, and $\rho_0=1$.
The algorithm maintains $B(z_i,\rho_i)\subseteq F_iC$.
\begin{enumerate}
\item While $\rho_i\le\eta\operatorname{tr}\sqrt A/32$, apply
\Cref{lem:rounding-estimation-tools} with $k=\Ot(\rho_i)$ to obtain
estimates satisfying
\begin{equation}\label{eq:subspace-growth-estimates}
 \begin{aligned}
 \tfrac12 A_i-\eta^2(\operatorname{tr}\sqrt A)^2I
 &\preceq\widehat A_i
 \preceq2A_i+\eta^2(\operatorname{tr}\sqrt A)^2I,\\
 \|\widehat m_i-m_i\|_2&\le\eta\operatorname{tr}\sqrt A.
 \end{aligned}
\end{equation}
Let $P_i$ project onto the eigenspaces of $\widehat A_i$ with
eigenvalues at most $(\operatorname{tr}\sqrt A)^2$, and put $M_i=I+P_i$.
Thus $M_i$ doubles these directions. Update
\[
 \begin{aligned}
 F_{i+1}&=M_iF_i,\\
 z_{i+1}&=M_i\left(z_i+
       \frac{8\rho_i}{\operatorname{tr}\sqrt A}(\widehat m_i-z_i)\right),\\
 \rho_{i+1}&=2(1-12\eta)\rho_i,
 \end{aligned}
\]
and increase $i$ by one.
\item Use the relative-estimation algorithm in
\Cref{lem:rounding-estimation-tools} to obtain a covariance estimate
$H_i$ and a mean estimate $\widehat m_i$. If $H_i\succ0$,
return $x\mapsto H_i^{-1/2}(F_ix-\widehat m_i)$; otherwise declare failure.
\end{enumerate}
Each estimation batch uses a fresh initial sample and the capped
attempts specified in \Cref{lem:subspace-initialization-and-budgets}.
The cost analysis below includes these calls.

The geometric step is to prove that the center and radius updates
preserve the contained ball. The old ball supplies width in the doubled
directions; the covariance ellipsoid supplies width in the remaining
directions. \Cref{lem:level-covariance-ellipsoid} supplies the required
ellipsoid inclusion in the fixed level set. The next lemma proves the
expansion from these two inclusions.
\begin{lemma}[Expansion inside a level set]
\label{lem:logconcave-level-expansion}
At a growth stage of the algorithm, suppose
\[
 B(z_i,\rho_i)\cup\bigl(m_i+\tfrac12A_i^{1/2}B_2^n\bigr)
 \subseteq F_iC.
\]
If \eqref{eq:subspace-growth-estimates} holds, then
$B(z_{i+1},\rho_{i+1})\subseteq F_{i+1}C$.
\end{lemma}
\begin{proof}
Take the convex combination of the ball and ellipsoid with weight
$8\rho_i/\operatorname{tr}\sqrt A$ on the ellipsoid, and apply $M_i$.
This weight is at most $\eta/4$, so the resulting set lies in $F_{i+1}C$.
Its center differs from $z_{i+1}$ by at most $16\eta\rho_i$.

Fix a unit vector $u$ and put $b=\|(I-P_i)u\|_2$.
Then $\|M_iu\|_2=\sqrt{4-3b^2}$.
On the range of $I-P_i$, the eigenvalues of $\widehat A_i$ exceed
$(\operatorname{tr}\sqrt A)^2$. The upper covariance estimate gives
\[
 \begin{aligned}
 2\|A_i^{1/2}M_iu\|_2^2
 &\ge (M_iu)^{\mathsf T}\widehat A_i(M_iu)
       -\eta^2(\operatorname{tr}\sqrt A)^2\|M_iu\|_2^2\\
 &\ge(\operatorname{tr}\sqrt A)^2(b^2-4\eta^2).
 \end{aligned}
\]
Thus $\|A_i^{1/2}M_iu\|_2
\ge(\operatorname{tr}\sqrt A)(b-2\eta)/\sqrt2$;
when $b<2\eta$, the right side is negative.
The support function of the convex combination after applying $M_i$,
measured from $z_{i+1}$, is therefore at least
\begin{equation}\label{eq:expansion-support-lower-bound}
 \rho_i\left[
 \left(1-\frac{8\rho_i}{\operatorname{tr}\sqrt A}\right)\sqrt{4-3b^2}
 +2\sqrt2(b-2\eta)-16\eta
 \right].
\end{equation}
For $0\le b\le1$, $\sqrt{4-3b^2}+2b\ge2$.
Together with $8\rho_i/\operatorname{tr}\sqrt A\le\eta/4$,
this shows that the support-function bound in
Equation~\eqref{eq:expansion-support-lower-bound} is at least
\[
 \begin{aligned}
 2\rho_i-(4\cdot2^{1/2}+16+\tfrac12)\eta\rho_i
 &\ge2(1-12\eta)\rho_i\\
 &=\rho_{i+1}.
 \end{aligned}
\]
The inequality in every direction proves the ball inclusion.
\end{proof}

\begin{proof}[Proof of \Cref{prop:logconcave-inner-rounder}]
First assume all moment estimates satisfy their stated guarantees. We
prove the ball invariant and bound the moments along the algorithm, then
apply \Cref{lem:rounding-estimation-tools} to compute its sampling cost.
We account for initialization and failure probabilities at the end.
By \Cref{lem:level-covariance-ellipsoid},
$m+\tfrac12A^{1/2}B_2^n\subseteq C$.
Its image under $F_i$ is
$m_i+\tfrac12A_i^{1/2}B_2^n$.
\Cref{lem:logconcave-level-expansion} consequently preserves the known
ball. Its radius grows by $2(1-12\eta)$, so there are at most
$2\ell$ growth stages.

We next bound the current covariance relative to the growing ball.
Its operator norm controls the cost per sample, and its trace controls
the number of samples in \Cref{lem:rounding-estimation-tools}.
The lower covariance estimate on the range of $P_i$ gives
$P_iA_iP_i\preceq3(\operatorname{tr}\sqrt A)^2P_i$.
Since $M_i^2=I+3P_i$, $A_{i+1}$ has the same eigenvalues as
$A_i^{1/2}(I+3P_i)A_i^{1/2}$.
Also $A_i^{1/2}P_iA_i^{1/2}$ and $P_iA_iP_i$ have the same nonzero
eigenvalues. Hence
\[
 \|A_{i+1}\|_\op
 \le\|A_i\|_\op+3\|P_iA_iP_i\|_\op
 \le\|A_i\|_\op+9(\operatorname{tr}\sqrt A)^2.
\]
Initially $\|A\|_\op\le(\operatorname{tr}\sqrt A)^2$.
Moreover, $\|F_i\|_\op\le2^i$ and
$\rho_i=[2(1-12\eta)]^i$, giving $\|F_i\|_\op\le C\rho_i$.
Writing $\|M\|_*$ for the sum of the singular values of a matrix $M$,
the nuclear norm inequality gives
$\operatorname{tr}\sqrt{A_i}
=\|F_iA^{1/2}\|_*
\le\|F_i\|_\op\operatorname{tr}\sqrt A$.
We obtain
\begin{equation}\label{eq:logconcave-inner-trace}
 \begin{aligned}
 \|A_i\|_\op&\le(1+9i)(\operatorname{tr}\sqrt A)^2,\\
 \operatorname{tr}\sqrt{A_i}&\le C\rho_i\operatorname{tr}\sqrt A,\\
 \operatorname{tr}A_i&\le
 C\rho_i(\operatorname{tr}\sqrt A)^2\sqrt\ell,
 \end{aligned}
\end{equation}
where the last bound follows by multiplying
$\sqrt{\|A_i\|_\op}$ and $\operatorname{tr}\sqrt{A_i}$.

Initialization also requires a second-moment bound about the known
center $z_i$. The center update and mean estimate imply
\[
 \frac{\|m_{i+1}-z_{i+1}\|_2}{\rho_{i+1}}
 \le
 \frac{\|m_i-z_i\|_2/\rho_i+8\eta}{1-12\eta}.
\]
Starting from $\|m_0-z_0\|_2=\|m\|_2$ and iterating at most $2\ell$
times gives $\|m_i-z_i\|_2/\rho_i\le C(\|m\|_2+1)$.
Combining this with \eqref{eq:logconcave-inner-trace} and $\rho_i\ge1$
gives the known bound
\begin{equation}\label{eq:subspace-second-moment}
 \frac{\sqrt{\mathbb E\|F_iX-z_i\|_2^2}}{\rho_i}
 \le C\left(\ell^{1/4}\operatorname{tr}\sqrt A+\|m\|_2+1\right).
\end{equation}

We now substitute these bounds into \Cref{lem:rounding-estimation-tools}.
For a growth batch, use
$a=\eta^2(\operatorname{tr}\sqrt A)^2$ and
$b=\eta\operatorname{tr}\sqrt A$.
By \eqref{eq:logconcave-inner-trace}, the required sample count is
$\Ot(1+\operatorname{tr}A_i/
(\eta^2(\operatorname{tr}\sqrt A)^2))=\Ot(\rho_i)$.
After rescaling the known ball to unit radius, the expected cost per
sample from a constant-warm start is
$\Ot(n^2\operatorname{tr}\sqrt A/\rho_i)$.
Thus the expected sampling work of a growth stage is
$\Ot(n^2\operatorname{tr}\sqrt A)$: the increase in sample count is
canceled by the decrease in cost per sample.
At termination, $\rho_i>\eta\operatorname{tr}\sqrt A/32$, so the expected
cost per sample is $\Ot(n^2)$ and the final estimation batch uses
$\Ot(n^3)$ expected queries.

By \eqref{eq:subspace-second-moment}, the initialization bound from
\Cref{prop:logconcave-fresh-initialization} is
$\Ot(n^{5/2}+n^2(\operatorname{tr}\sqrt A+\|m\|_2))$ per batch.
There are at most $2\ell+1$ batches. These initialization costs and
expected sampling costs from constant-warm laws sum to
$\Ot(n^3+n^2(\operatorname{tr}\sqrt A+\|m\|_2))$.
\Cref{lem:subspace-initialization-and-budgets} accounts for the
initialization error and converts the expected sampling costs into
query caps, using a logarithmic number of attempts. It gives the same
total query bound with failure probability at most $\delta$.

On success, the final estimates satisfy
$\tfrac78 A_i\preceq H_i\preceq\tfrac98 A_i$ and
$\|A_i^{-1/2}(\widehat m_i-m_i)\|_2\le1/16$.
Thus the returned map $F$ satisfies
$(8/9)I\preceq\Cov(F(X))\preceq(8/7)I$ and
$\|\E F(X)\|_2\le\sqrt{8/7}/16$, proving the stated bounds.
\end{proof}

\subsection{Gaussian cooling for general inputs}
\label{subsec:gaussian-cooling-upper}

We use Gaussian cooling~\cite{CV18:journal,KV25} to reduce the general input to
\Cref{prop:logconcave-inner-rounder}. After each halving of the
temperature, the previous affine coordinates satisfy
$\operatorname{tr}\sqrt A+\|m\|_2=\widetilde O(n)$ for the new target,
and a constant-radius ball remains in a level set. Thus each rounding
call costs $\widetilde O(n^3)$. We first verify these claims using the
logconcave estimates in \Cref{subsec:rounding-cooling-estimates}, then
prove the main theorems by initializing the sequence and removing the
final Gaussian factor.

Assume \eqref{eq:grounded-input} and write $a=\inf V$. For $t\ge0$, define
$\Phi_t(x)=V(x)+t\|x-x_0\|_2^2/2$, let
$Z_t=\int_{\mathbb R^n}e^{-\Phi_t(x)}\,dx$, and put
$\mu_t(dx)=Z_t^{-1}e^{-\Phi_t(x)}\,dx$. Thus $\mu_0=\mu$.
Write $a_t=\inf\Phi_t$, $m_t=\mathbb E_{\mu_t}X$, and
$A_t=\operatorname{Cov}(\mu_t)$. Set
$t_0=64(n+1)(hn/r)^2$ and $L=\log(2+nR/r)$.

We maintain slightly stronger rounding bounds so that the last Gaussian
factor can be removed without another rounding call. For $X_t\sim\mu_t$,
the current affine map $F$ will satisfy
\begin{equation}\label{eq:cooling-internal-rounding}
 \tfrac34I_n\preceq\operatorname{Cov}(F(X_t))\preceq\tfrac43I_n
 \quad\text{and}\quad
 \|\mathbb EF(X_t)\|_2\le\tfrac18.
\end{equation}

\Cref{cor:half-temperature-bridge} bounds the change in the covariance
and mean when the temperature is halved, and
\Cref{lem:level-covariance-ellipsoid} places a covariance ellipsoid in a
level set. The following lemma applies these two appendix estimates in
the current affine coordinates. Its conclusions are exactly the inputs
needed for a nearly cubic call to \Cref{prop:logconcave-inner-rounder},
after a constant rescaling.
\begin{lemma}[The next cooling phase]\label{lem:cooling-transition}
Let $0<t\le t_0$, and let $F$ be an invertible affine map satisfying
\eqref{eq:cooling-internal-rounding}. There is a universal constant $H_0$
such that
\[
 \tfrac14 B_2^n\subseteq
 \{y:\Phi_{t/2}(F^{-1}y)\le a_{t/2}+(H_0+h)n\}.
\]
For $X_{t/2}\sim\mu_{t/2}$, the next target also satisfies
\begin{align*}
 \operatorname{tr}\sqrt{\operatorname{Cov}(F(X_{t/2}))}
 &\le C_hnL\log(2n),\\
 \|\mathbb EF(X_{t/2})\|_2&\le C_hnL.
\end{align*}
\end{lemma}
\begin{proof}
Apply \Cref{lem:level-covariance-ellipsoid} in the coordinates $F$.
The level set at height $H_0n$ above $a_t$ contains half the covariance
ellipsoid about the transformed mean. By
\eqref{eq:cooling-internal-rounding}, it therefore contains the ball
of radius $\sqrt3/4-1/8>1/4$ about zero. Since
$\Phi_{t/2}\le\Phi_t$ and $a_t-a_{t/2}\le hn$ by
\Cref{cor:half-temperature-bridge}, this proves the asserted inclusion.

Let $T$ be the linear part of $F$ and put
$B=A_t^{-1/2}A_{t/2}A_t^{-1/2}$. The current covariance bound gives
$\|TA_t^{1/2}\|_\op\le\sqrt{4/3}$. Hence
\[
 \operatorname{tr}\sqrt{TA_{t/2}T^{\mathsf T}}
 \le\|TA_t^{1/2}\|_\op\operatorname{tr}\sqrt B
 \le C_hnL\log(2n),
\]
where the last inequality is \Cref{cor:half-temperature-bridge}.
The same corollary gives $\|A_t^{-1/2}(m_{t/2}-m_t)\|_2\le C_hnL$.
Together with $\|\E F(X_t)\|_2\le1/8$, this yields
$\|\E F(X_{t/2})\|_2\le1/8+\sqrt{4/3}\,C_hnL$, as required.
\end{proof}

We now prove the main upper bounds. The preceding lemma justifies each
rounding call. The first target is handled by the rejection sampler in
\Cref{lem:high-temperature-rejection}; the final target is compared with
the original law using \Cref{lem:cooling-endpoint}.
\begin{proof}[Proof of \Cref{thm:upper-logconcave,thm:upper-body}]
Fix a sufficiently small universal constant $\theta>0$. Put
$t_j=2^{-j}t_0$, and let $J$ be the first index for which
$t_JR^2\le\theta$. The known input bounds determine $J$, and
$J=O_h(L)$. Give the initialization and each of the $J$ rounding calls
failure probability at most $\delta=1/(6(J+1))$.

By \Cref{lem:high-temperature-rejection}, we can sample exactly and
independently from $\mu_{t_0}$ with a constant expected number of potential
evaluations per sample. The independent-sample case of the concentration argument in
\Cref{subsec:upper-moments} gives the relative moment estimates
\eqref{eq:logconcave-final-estimates} from $\widetilde O(n)$ samples.
They define an initial affine map $F_0$ satisfying
\eqref{eq:cooling-internal-rounding}. The prescribed
failure probability and the query cap for the rejection sampler change
this cost only by logarithmic factors.

Suppose $F_{j-1}$ satisfies \eqref{eq:cooling-internal-rounding} for
$\mu_{t_{j-1}}$. By \Cref{lem:cooling-transition}, the next law in these
coordinates has a known radius-$1/4$ ball in a level set of height
$(H_0+h)n$, together with the stated mean and covariance bounds.
To rescale this ball to the unit ball, use the potential
$W_j(y)=\Phi_{t_j}(F_{j-1}^{-1}(y/4))$.
Its normalized law is that of $4F_{j-1}(X_{t_j})$.
Multiplying the coordinates by four turns the known radius-$1/4$ ball
into the unit ball required by \Cref{prop:logconcave-inner-rounder}.
It also multiplies the mean norm and the sum of the square roots of the
covariance eigenvalues by four.
Apply that proposition with $H=H_0+h$, using
$4C_hnL\log(2n)$ as the bound on $\operatorname{tr}(A^{1/2})$ and
$4C_hnL$ as the bound on the mean norm.
Here $C_h$ is the constant from \Cref{lem:cooling-transition}, which
depends only on the fixed level-set height parameter $h$.
Compose its returned affine map with
$x\mapsto4F_{j-1}(x)$ to obtain $F_j$. On success, $F_j$ satisfies
\eqref{eq:cooling-internal-rounding} for $\mu_{t_j}$, completing the
induction.

After every successful phase, \Cref{lem:level-covariance-ellipsoid}
again places the radius-$1/4$ ball in the level set at height $H_0n$ of
the current potential. Thus the next call always uses height $H_0+h$;
the level-set heights do not accumulate over phases.

We pass the affine coordinates between temperatures and use the
initialization included in \Cref{prop:logconcave-inner-rounder} to
supply fresh initial samples for the new target. Each rounding call,
including this initialization, costs $\widetilde O(n^3)$,
since both moment bounds are $\widetilde O(n)$. There are $J=O_h(L)$ calls, so
the overall query bound is $\widetilde O(n^3)$. Each evaluation of $W_j$ uses one evaluation of
$V$. All failures covered by the cited algorithms are included in the
allocated error probabilities, so a union bound gives success probability
at least $5/6$.

It remains to transfer the final map $F_J$ to $\mu$. By
\Cref{lem:cooling-endpoint},
\begin{align*}
 (1-C\theta)A_\mu&\preceq A_{t_J}\preceq(1+C\theta)A_\mu,\\
 \|A_\mu^{-1/2}(m_{t_J}-m_\mu)\|_2&\le C\theta.
\end{align*}
For $X\sim\mu$ and sufficiently small $\theta$, the covariance comparison
and \eqref{eq:cooling-internal-rounding} give
$\tfrac12I_n\preceq\Cov(F_J(X))\preceq2I_n$.
If $T_J$ is the linear part of $F_J$, then
$\|T_JA_\mu^{1/2}\|_\op\le\sqrt2$, so
$\|\E F_J(X)\|_2\le1/8+\sqrt2C\theta\le1/2$.
Thus $(T_J,-T_J^{-1}F_J(0))$ is the required rounding pair.

For a convex body, take $V=0$ on $K$ and $V=+\infty$ outside $K$.
The input assumptions hold with $x_0=0$, and each potential evaluation
requires one membership query. This proves \Cref{thm:upper-body} as well.
\end{proof}

Finally, \Cref{lem:support-ball} turns the affine rounding guarantee into
the contained ball and second-moment bound required by the volume
algorithm of Cousins and Vempala~\cite{CV18:journal}. This gives the volume bound
stated in the introduction.
\begin{corollary}\label{cor:volume-reduction}
Suppose a membership algorithm rounds a body $K$ as in
\Cref{def:rounding} using at most $Q$ queries and failing with probability
at most $1/6$. For $\varepsilon\in(0,1)$, an algorithm uses
$Q+\Ot(n^3\varepsilon^{-2})$ queries and returns $\widehat V$ such that,
with probability at least $2/3$,
\[
 (1-\varepsilon)\vol(K)\le\widehat V\le(1+\varepsilon)\vol(K).
\]
\end{corollary}
\begin{proof}
Let $K'=T(K-z)$ be the returned body and condition on a successful
rounding. By \Cref{lem:support-ball}, it contains
$B(0,\rho)$ for $\rho=1/\sqrt2-1/2>0$. For $Y\sim\Unif(K')$,
we also have $\mathbb E\|Y\|_2^2\le2n+1/4$.
The body $\rho^{-1}K'$ therefore contains the unit ball and has second
moment $O(n)$. Apply the well-rounded volume algorithm of
Cousins and Vempala~\cite{CV18:journal} with failure probability at most $1/6$,
and let $\widehat W$ be its volume estimate.
Each membership query to the transformed body costs one query to $K$.
Return $\widehat V=\rho^n\widehat W/|\det T|$. Since
$\vol(K)=\rho^n\vol(\rho^{-1}K')/|\det T|$, relative error is preserved.
A union bound gives the claimed success probability.
\end{proof}

\section{A Cubic Lower Bound for Isotropic Rounding}
\label{sec:lower-bound}

We construct bodies whose covariance is nearly a scalar matrix but whose
centroid points in an unknown direction.
A successful rounding map must locate this centroid, revealing the
hidden direction to constant accuracy.
Recovering that direction requires $\Omega(n)$ mutual information,
whereas the answers to $q$ adaptive queries contain at most
$O(q\log^2 n/n^2)$ mutual information.
Comparing these bounds gives $q=\Omega(n^3/\log^2 n)$.
We first establish the geometry of the construction, then bound the
information in an adaptive membership transcript, and finally prove
the recovery claim in \Cref{subsec:lower-decoding}.

The construction gives a lower bound even when the supplied radius ratio is
polynomial in the dimension. Here and below all oracle policies and output
maps are assumed measurable.

The following is the explicit form of \Cref{thm:lower-main}; its proof
is at the end of \Cref{subsec:lower-decoding}.
\begin{theorem}[Lower bound]
\label{thm:lower-rounding}
For all sufficiently large $n$, set $r_n=1/(3\sqrt n)$ and
$R_n=2/\sqrt{180\log n}$.
Suppose a randomized membership-oracle algorithm, given any convex body
$r_nB_2^n\subseteq K\subseteq R_nB_2^n$, returns an invertible matrix $T$
and a vector $z$ such that, with probability at least $2/3$, the following
bounds hold for $X\sim\Unif(K)$:
\begin{gather*}
 \frac12 I_n\preceq T\operatorname{Cov}(X)T^{\mathsf T}
 \preceq 2I_n,\\
 \left\|T\bigl(\mathbb E X-z\bigr)\right\|_2\le\frac12.
\end{gather*}
Then its worst-case expected number of membership queries is
$\Omegat(n^3)$.
\end{theorem}
The same lower bound applies to logconcave rounding with a potential
evaluation oracle: for the indicator potential of a body, each evaluation
reveals exactly its membership bit.
These bodies also satisfy the input assumptions in
Equation~\eqref{eq:grounded-input} with $x_0=0$.

\subsection{Polars of Gaussian polytopes}
\label{subsec:lower-family}

Our bodies are polars of Gaussian polytopes with a small common translation.
For a set $A\subseteq\mathbb R^n$, its polar is
$A^\circ=\{x:\langle x,y\rangle\le1\text{ for every }y\in A\}$.
We keep the Gaussian experiment unconditioned throughout the information
argument; the required geometric event will have probability tending to one.
Set $N=\lceil n^{100}\rceil$, $a=\sqrt{180\log n}$,
$b=\sqrt{198\log n}$, and $\tau=4b/n$.
Draw $Z_1,\ldots,Z_N\sim N(0,I_n)$ and $\Xi\sim N(0,\tau^2I_n)$
independently, and put
\begin{equation}\label{eq:lower-family}
 \begin{aligned}
 P_Z&=\operatorname{conv}\{Z_1,\ldots,Z_N\},\\
 K_{\Xi,Z}&=(P_Z+\Xi)^\circ
 =\{x:\langle Z_i+\Xi,x\rangle\le1\text{ for all }i\}.
 \end{aligned}
\end{equation}
On the probability-one event $\Xi\ne0$, write
$v=\Xi/\|\Xi\|_2$ and $c_\Xi=\sqrt n\|\Xi\|_2/b$.
The direction $v$ is uniform on the unit sphere
$S^{n-1}=\{u\in\mathbb R^n:\|u\|_2=1\}$ and independent of $c_\Xi$.

The first estimate supplies the radius promise and the scale of the hidden
translation.
\begin{lemma}[Containment of the random bodies]
\label{lem:lower-geometry}
With probability $1-o(1)$, $c_\Xi\in[3.9,4.1]$ and
\begin{align*}
 aB_2^n&\subseteq P_Z\subseteq2\sqrt nB_2^n,\\
 r_nB_2^n&\subseteq K_{\Xi,Z}\subseteq R_nB_2^n.
\end{align*}
\end{lemma}
\begin{proof}
The first assertion is the chi-square concentration bound applied to
$c_\Xi=4\|G\|_2/\sqrt n$, where $G\sim N(0,I_n)$. The same bound and a
union bound give $\max_i\|Z_i\|_2\le2\sqrt n$ with probability $1-o(1)$.
Take a $(4\sqrt n)^{-1}$-net of the sphere, of size
$\exp(O(n\log n))$. At a fixed net point $u$,
\[
 \mathbb P\{\max_i\langle Z_i,u\rangle<a+1\}
 \le\exp\bigl(-N\mathbb P\{G_1\ge a+1\}\bigr)
 =\exp(-n^{10-o(1)}).
\]
Here the last equality follows, for example, by integrating the Gaussian
density on $[a+1,a+1+(a+1)^{-1}]$. A union bound over the net, followed by
the norm bound, gives $\max_i\langle Z_i,u\rangle\ge a$ at every point of
the sphere. This is the support-function criterion for $aB_2^n\subseteq P_Z$.
Finally, $\|\Xi\|_2\le4.1b/\sqrt n=o(a)$, so for all sufficiently large $n$,
\((a/2)B_2^n\subseteq P_Z+\Xi\subseteq3\sqrt nB_2^n.\)
Taking polars gives the stated promise.
\end{proof}

\subsection{Centroid and covariance bounds}
\label{subsec:lower-moments}

We show that the translation in \eqref{eq:lower-family} moves the centroid
by a constant number of standard deviations while leaving the centered
covariance asymptotically spherical. The calculation uses radial integrals
and concentration of the inverse support function of the Gaussian polytope.
Let $\sigma$ denote uniform probability measure on $S^{n-1}$.
Set $H_Z(u)=\max\{a,\max_i\langle Z_i,u\rangle\}$ and
$\rho_{\xi,Z}(u)=(H_Z(u)+\langle\xi,u\rangle)^{-1}$.
For $\|\xi\|_2<a$, define the scalar
$D=\int\rho_{\xi,Z}(u)^n\,d\sigma(u)$, the vector
$L=\int u\rho_{\xi,Z}(u)^{n+1}\,d\sigma(u)$, and the matrix
$Q=\int uu^{\mathsf T}\rho_{\xi,Z}(u)^{n+2}\,d\sigma(u)$.
On the containment event in \Cref{lem:lower-geometry}, polar coordinates
give the exact identities
\begin{equation}
 m_{\xi,Z}=\frac n{n+1}\frac LD
 \quad\text{and}\quad
 M_{\xi,Z}=\frac n{n+2}\frac QD,
 \label{eq:lower-radial-moments}
\end{equation}
where $m_{\xi,Z}$ and $M_{\xi,Z}$ are the centroid and uncentered second
moment of $K_{\xi,Z}$.

The next estimate identifies the one-dimensional scale in these integrals.
\begin{lemma}[Inverse moments of a Gaussian maximum]
\label{lem:lower-inverse-moments}
Let $H=\max\{a,G_1,\ldots,G_N\}$ for independent standard normal $G_i$,
and let $A_k(s)=\mathbb E(H+s)^{-k}$ for $s>-a$. Uniformly for
$k\in\{n,n+1,n+2\}$ and $c\in[3.9,4.1]$,
\begin{align}
 \frac{A_k(cbz/n)}{A_k(0)}&\longrightarrow e^{-cz}
 &&\text{uniformly on compact sets of }z,\label{eq:lower-shift-limit}\\
 \frac{A_{n+j}(0)}{A_n(0)}&=b^{-j}(1+o(1))
 &&(j=1,2).\label{eq:lower-ratio-limit}
\end{align}
Moreover, for all sufficiently large $n$ and $|z|\le\sqrt n$,
\begin{equation}
 0\le\frac{A_k(cbz/n)}{A_k(0)}\le e^{10\max\{-z,0\}}.
 \label{eq:lower-envelope}
\end{equation}
\end{lemma}
\begin{proof}
Use the Gaussian density and upper tail defined in
\Cref{sec:preliminaries}, with the Mills bounds in
Equation~\eqref{eq:mills-bounds}.
Let $\lambda=n/b^2$, choose $t_n$ from $N\overline\Phi(t_n)=\lambda$,
and put $w_n=t_n+3/t_n$. The bounds imply $t_n^2=b^2+O(\log\log n)$, $t_n>a$, and
$w_n/b\longrightarrow1$. They also give
$N\overline\Phi(w_n)\le\lambda/4$. Thus
\begin{equation}\label{eq:lower-normalizer}
 \begin{aligned}
 \mathbb P\{t_n<\max_iG_i\le w_n\}&\ge\tfrac12e^{-\lambda/2},\\
 A_k(0)&\ge\tfrac12w_n^{-k}e^{-\lambda/2}.
 \end{aligned}
\end{equation}
Indeed, the probability is
$(1-\overline\Phi(w_n))^N-(1-\overline\Phi(t_n))^N$, with
$(1-\overline\Phi(w_n))^N\ge e^{-\lambda/2}$ and
$(1-\overline\Phi(t_n))^N\le e^{-\lambda}$.

Give $H$ the tilted law $d\nu_k=H^{-k}d\mathbb P/A_k(0)$. For every fixed
$\varepsilon>0$, \eqref{eq:lower-normalizer} shows that
\[
 \nu_k\{H\ge(1+\varepsilon)b\}
 \le2\left(\frac{w_n}{(1+\varepsilon)b}\right)^k e^{\lambda/2}=o(1).
\]
If $(1-\varepsilon)b\ge a$, the lower bound in
Equation~\eqref{eq:mills-bounds} gives
$N\overline\Phi((1-\varepsilon)b)\ge n^{1+\delta_\varepsilon}$ for some
$\delta_\varepsilon>0$ and all sufficiently large $n$. Consequently
\[
 \nu_k\{H\le(1-\varepsilon)b\}
 \le2(w_n/a)^k\exp(\lambda/2-n^{1+\delta_\varepsilon})=o(1).
\]
Otherwise the lower-tail event is empty. Since $0<b/H\le b/a=\sqrt{1.1}$,
we have $\mathbb E_{\nu_k}|(b/H)^j-1|=o(1)$ for every fixed $j>0$.
The identity
\[
 \frac{A_k(cbz/n)}{A_k(0)}
 =\mathbb E_{\nu_k}\exp\left[-k\log\left(1+\frac{cbz}{nH}\right)\right]
\]
now proves \eqref{eq:lower-shift-limit} by a uniform Taylor expansion on
bounded $z$-intervals. Also
$A_{n+j}(0)/A_n(0)=\mathbb E_{\nu_n}H^{-j}$, proving
\eqref{eq:lower-ratio-limit}. For $z\ge0$, $(1+cbz/(nH))^{-k}\le1$.
For $z<0$, put $q=cb(-z)/(nH)$. Uniformly in the stated
range, $q=o(1)$, and
\[
 -k\log(1-q)\le\frac{kq}{1-q}\le10(-z)
\]
for sufficiently large $n$. This proves \eqref{eq:lower-envelope}.
\end{proof}

The spherical coordinate $\sqrt n\langle u,v\rangle$ becomes Gaussian.
The preceding lemma therefore computes the averages of all three radial
integrals.
\begin{lemma}[Expected radial moments]
\label{lem:lower-annealed}
Put $s_n=(b\sqrt n)^{-1}$. Uniformly for $c\in[3.9,4.1]$, $v\in S^{n-1}$,
and $\xi=cbv/\sqrt n$,
\begin{align*}
 \left\|\frac n{n+1}\frac{\mathbb E_ZL}{\mathbb E_ZD}+cs_nv\right\|_2
 &=o(s_n),\\
 \left\|\frac n{n+2}\frac{\mathbb E_ZQ}{\mathbb E_ZD}
 -s_n^2(I_n+c^2vv^{\mathsf T})\right\|_{\mathrm{op}}
 &=o(s_n^2).
\end{align*}
\end{lemma}
\begin{proof}
For $U\sim\sigma$, the variable $G_n=\sqrt n\langle U,v\rangle$ has density
\[
 g_n(z)=\frac{\Gamma(n/2)}{\sqrt{\pi n}\Gamma((n-1)/2)}
 (1-z^2/n)^{(n-3)/2}\mathbf1_{\{|z|<\sqrt n\}}.
\]
It converges locally uniformly to $\varphi$, and for $n\ge6$ it is bounded
by $Ce^{-z^2/4}$. Thus \eqref{eq:lower-envelope} permits dominated
convergence even after multiplication by $1+z^2$. Then we have
\(\mathbb E_ZD=A_n(0)\mathbb E\frac{A_n(cbG_n/n)}{A_n(0)},\)
with analogous expressions for $L$ and $Q$. Rotational symmetry about $v$
makes $\mathbb E_ZL$ parallel to $v$ and gives only two eigenvalues for
$\mathbb E_ZQ$. For $G\sim N(0,1)$, their limits follow from
\[
 \frac{\mathbb E[Ge^{-cG}]}{\mathbb E[e^{-cG}]}=-c
 \quad\text{and}\quad
 \frac{\mathbb E[G^2e^{-cG}]}{\mathbb E[e^{-cG}]}=1+c^2,
\]
together with the identity, for a unit vector $w\perp v$,
$\mathbb E[\langle U,w\rangle^2\mid G_n]=(1-G_n^2/n)/(n-1)$.
Equation~\eqref{eq:lower-ratio-limit} supplies the factors $b^{-1}$ and
$b^{-2}$. All domination is uniform over the compact interval of $c$.
\end{proof}

We next show that the integrals for a typical polytope are close to their
averages. The large number of independent facets makes inverse support
values in two moderately separated directions nearly independent.
For a matrix $B$, write $\|B\|_F=(\sum_{i,j}B_{ij}^2)^{1/2}$ for
its Frobenius norm.
\begin{lemma}[Concentration of the radial integrals]
\label{lem:lower-concentration}
For $k\in\{n,n+1,n+2\}$, put
$W_{k,u}=(H_Z(u)+\langle\xi,u\rangle)^{-k}$, let
$m_{k,u}=\mathbb E_ZW_{k,u}$, and set $B_k=\int m_{k,u}\,d\sigma(u)$.
Uniformly for $\xi=cbv/\sqrt n$, $c\in[3.9,4.1]$, $v\in S^{n-1}$,
and sufficiently large $n$,
\[
 \frac{\operatorname{Var}_Z D}{B_n^2}
 +\frac{\mathbb E_Z\|L-\mathbb E_ZL\|_2^2}{B_{n+1}^2}
 +\frac{\mathbb E_Z\|Q-\mathbb E_ZQ\|_F^2}{B_{n+2}^2}
 \le n^{-19}.
\]
\end{lemma}
\begin{proof}
Write $\Phi$ for the standard one-dimensional Gaussian distribution
function. For $u,w\in S^{n-1}$ with $\rho=\langle u,w\rangle$, let $\Phi_\rho(x,y)$
be the distribution function of a standard bivariate normal pair with
correlation $\rho$. Independence over facets gives, for $x,y\ge a$,
\[
 \mathbb P\{H_Z(u)\le x,H_Z(w)\le y\}=\Phi_\rho(x,y)^N.
\]
If $|\rho|\le1/2$, inclusion--exclusion and the implication
$\{G>x,G'>y\}\subseteq\{G+G'>2a\}$ show that
\[
 |\Phi_\rho(x,y)-\Phi(x)\Phi(y)|
 \le\mathbb P\{G+G'>2a\}+\overline\Phi(x)\overline\Phi(y)
 \le Cn^{-120}/\sqrt{\log n}.
\]
Since $\Phi(x)\Phi(y)\ge1/4$ for large $n$ and $N\le2n^{100}$, it follows
that, uniformly on this region,
\begin{equation}
 \left|\frac{\Phi_\rho(x,y)^N}{\Phi(x)^N\Phi(y)^N}-1\right|
 \le C' n^{-20}/\sqrt{\log n}=:\varepsilon_n.
 \label{eq:lower-two-directions}
\end{equation}
For $s=\langle\xi,u\rangle$ the layer-cake identity is
\[
 (H_Z(u)+s)^{-k}
 =\int_a^\infty\mathbf1_{\{H_Z(u)\le x\}}k(x+s)^{-k-1}\,dx.
\]
Applying \eqref{eq:lower-two-directions} inside the resulting double integral
gives, whenever $|\langle u,w\rangle|\le1/2$,
\begin{equation}\label{eq:lower-weight-covariance}
 |\operatorname{Cov}(W_{k,u},W_{\ell,w})|
 \le\varepsilon_n m_{k,u}m_{\ell,w}.
\end{equation}

It remains to control $|\langle u,w\rangle|>1/2$. Put $r=\|\xi\|_2$ and
$p_n=\mathbb P\{H\le b\}$.
The upper bound in Equation~\eqref{eq:mills-bounds} gives
$p_n\ge\exp(-Cn/\sqrt{\log n})$. We have
$W_{k,u}\le(a-r)^{-k}$ and $m_{k,u}\ge p_n(b+r)^{-k}$,
while
\[
 (\sigma\otimes\sigma)\{(u,w):|\langle u,w\rangle|>1/2\}
 \le\exp\left[-\frac n2\log(4/3)+O(\log n)\right].
\]
Consequently, for an absolute $\gamma>0$,
\[
 \iint_{|\langle u,w\rangle|>1/2}
 \bigl(\mathbb E[W_{k,u}W_{\ell,w}]+m_{k,u}m_{\ell,w}\bigr)
 \,d\sigma(u)\,d\sigma(w)
 \le e^{-\gamma n}B_kB_\ell.
\]
Indeed, the logarithm of the integral divided by $B_kB_\ell$ is at most
\[
 -\frac n2\log(4/3)+(2n+4)\log\frac{b+r}{a-r}
 +O(n/\sqrt{\log n})+O(\log n).
\]
This is negative of order $n$ because $r/a=O(n^{-1/2})$ and
$\log(b^2/a^2)=\log(1.1)<\tfrac12\log(4/3)$.

Finally, express each of the three variances as a double integral of
covariances. The kernels for $D,L,Q$ are respectively
$1$, $\langle u,w\rangle$, and $\langle u,w\rangle^2$, all bounded in
absolute value by one. Equation~\eqref{eq:lower-weight-covariance} and the
exceptional-set estimate bound the sum by
$3(\varepsilon_n+e^{-\gamma n})\le n^{-19}$ for large $n$.
\end{proof}

Centering the second moment cancels its directional term. The covariance
is therefore nearly spherical, while the centroid retains the unknown
direction that rounding must recover.
\begin{proposition}[Centroid and covariance of the random bodies]
\label{prop:lower-typical-moments}
There is a deterministic sequence $\eta_n\to0$ such that the following
holds uniformly for $c\in[3.9,4.1]$ and $v\in S^{n-1}$.
Put $\xi=cbv/\sqrt n$, $s_n=1/(b\sqrt n)$, and
$C_{\xi,Z}=\operatorname{Cov}(\operatorname{Unif}(K_{\xi,Z}))$.
With probability $1-o(1)$ over $Z$,
\begin{align*}
 \|m_{\xi,Z}+cs_nv\|_2&\le\eta_ns_n,\\
 (1-\eta_n)s_n^2I_n&\preceq C_{\xi,Z}\preceq(1+\eta_n)s_n^2I_n.
\end{align*}
\end{proposition}
\begin{proof}
\Cref{lem:lower-concentration} and Markov's inequality show, with
probability $1-O(n^{-13})$, that $|D-B_n|\le n^{-3}B_n$,
$\|L-\mathbb E L\|_2\le n^{-3}B_{n+1}$, and
$\|Q-\mathbb E Q\|_F\le n^{-3}B_{n+2}$.
The argument of \Cref{lem:lower-annealed} also gives
$B_{n+j}/B_n=b^{-j}(1+o(1))$ for $j=1,2$. Thus replacing the ratios of
expectations in that lemma by $L/D$ and $Q/D$ changes them by
$o(s_n)$ and $o(s_n^2)$, respectively. Intersect with the Gaussian-polytope containment event
$aB_2^n\subseteq P_Z\subseteq2\sqrt nB_2^n$ proved in
\Cref{lem:lower-geometry}, and use \eqref{eq:lower-radial-moments}.
This gives
\begin{align*}
 \|m_{\xi,Z}+cs_nv\|_2&=o(s_n),\\
 \|M_{\xi,Z}-s_n^2(I_n+c^2vv^{\mathsf T})\|_{\mathrm{op}}&=o(s_n^2).
\end{align*}
Since $c$ is bounded,
$\|m_{\xi,Z}m_{\xi,Z}^{\mathsf T}-c^2s_n^2vv^{\mathsf T}\|_{\mathrm{op}}=o(s_n^2)$.
Therefore $C_{\xi,Z}=M_{\xi,Z}-m_{\xi,Z}m_{\xi,Z}^{\mathsf T}$ satisfies
$\|C_{\xi,Z}-s_n^2I_n\|_{\mathrm{op}}=o(s_n^2)$.
Uniformity in the preceding estimates supplies a common sequence $\eta_n$.
\end{proof}

\subsection{Information in an adaptive transcript}
\label{subsec:lower-information}

We analyze an oracle that returns zero when the query point belongs to
the body and otherwise returns the index of its first violated facet.
Testing whether this answer is zero recovers the membership answer, so
a lower bound for this more informative oracle also applies to membership
queries.
Conditional on the common translation and
the algorithm's random seed, the facet vectors remain independent after
conditioning on a transcript. This factorization follows
Rademacher and Vempala~\cite{RV08:journal}; the argument
below uses it to control the information about a common Gaussian translation.
For a query $x$, let
\[
 J(x)=\begin{cases}
 0,&\langle Z_i+\Xi,x\rangle\le1\text{ for every }i,\\
 \min\{i:\langle Z_i+\Xi,x\rangle>1\},&\text{otherwise}.
 \end{cases}
\]
One call simulates a membership query by testing whether $J(x)=0$.

We use relative entropy
$D_{\mathrm{KL}}(\nu\|\mu)=\int\log(d\nu/d\mu)\,d\nu$, with value
$+\infty$ unless $\nu\ll\mu$. For random variables $Y,Z$, their
mutual information is the relative entropy of their joint law with
respect to the product of their marginal laws, denoted $I(Y;Z)$.

Condition on $\Xi=\xi$ and on the algorithm's private seed.
For a fixed answer history $h$, put $X_i=Z_i+\xi$. Each query restricts
each $X_i$ to a halfspace or leaves it unrestricted. Hence the conditional
law of the $X_i$ remains a product, with factors
\begin{equation}
 d\mu_{i,h,\xi}(x)\propto
 e^{-\|x-\xi\|_2^2/2}\mathbf1_{C_i(h)}(x)\,dx,
 \label{eq:lower-product-posterior}
\end{equation}
where each $C_i(h)$ is a convex polyhedron independent of $\xi$.
Only histories of positive probability need be considered.

A Gaussian restricted to a convex set has a dimension-free bound on the
change of its mean after conditioning. This is the only concentration fact
needed for the posterior argument.
\begin{lemma}[Conditional mean of a restricted Gaussian]
\label{lem:lower-conditional-mean}
Let $\mu(dx)\propto e^{-\|x-\xi\|_2^2/2}\mathbf1_C(x)\,dx$,
where $C\subseteq\mathbb R^n$ is convex and has positive Lebesgue measure. If $A$ is measurable and $p=\mu(A)>0$, then
\(\|\mathbb E_\mu[X\mid A]-\mathbb E_\mu X\|_2^2\le2\log(1/p).\)
\end{lemma}
\begin{proof}
Every linear tilt of $\mu$ is again a Gaussian restricted to $C$.
The Brascamp--Lieb consequence in
Equation~\eqref{eq:strong-logconcave-poincare}, applied to linear functions,
gives
$\nabla_t^2\log\mathbb E_\mu e^{\langle t,X\rangle}\preceq I_n$
for every $t\in\mathbb R^n$. The assertion also holds for
nonsmooth $C$: replace its indicator potential by
$j\operatorname{dist}(x,\overline C)^2$, smooth this convex penalty, and
pass to the limit by Gaussian domination. It follows by integrating the
Hessian of the log moment-generating function that, for every $t\in\mathbb R^n$,
\[
 \log\mathbb E_\mu e^{\langle t,X-\mathbb E_\mu X\rangle}
 \le\tfrac12\|t\|_2^2.
\]
For $\nu=\mu(\cdot\mid A)$, the entropy variational inequality now gives
\[
 \langle t,\mathbb E_\nu X-\mathbb E_\mu X\rangle
 \le D_{\rm KL}(\nu\|\mu)+\tfrac12\|t\|_2^2
 =\log(1/p)+\tfrac12\|t\|_2^2.
\]
Take $t=\mathbb E_\nu X-\mathbb E_\mu X$.
\end{proof}

The following bound controls the Fisher information of the entire adaptive
transcript by its number of queries.
For a transcript law $p_\xi(t)$, the Fisher information matrix is
$\mathcal I_\xi=\sum_{t:p_\xi(t)>0}p_\xi(t)
\nabla_\xi\log p_\xi(t)\nabla_\xi\log p_\xi(t)^{\mathsf T}$.
\begin{lemma}[Fisher information in the transcript]
\label{lem:lower-fisher}
Fix the private seed of an adaptive policy making at most $q$ calls to $J$.
Let $T_q$ be its answer transcript and
$\mathcal I_\xi(T_q)$ its Fisher information
matrix in the common location parameter $\xi$. Then
\(\operatorname{Tr}\mathcal I_\xi(T_q)\le2q\log(N+1).\)
\end{lemma}
\begin{proof}
Define
$V_\xi(h)=\sum_i\operatorname{Tr}\operatorname{Cov}_\xi(X_i\mid h)$.
For the next answer $J$, write $p_j=\mathbb P_\xi(J=j\mid h)$,
omitting zero-probability answers from all sums.
Under \eqref{eq:lower-product-posterior}, the answer event is a product
$\prod_iA_{i,j}$ with $p_j=\prod_i\mu_{i,h,\xi}(A_{i,j})$.
\Cref{lem:lower-conditional-mean} and total variance give
\begin{align*}
 V_\xi(h)-\mathbb E_\xi[V_\xi(h,J)\mid h]
 &=\sum_jp_j\sum_i
 \|\mathbb E_\xi[X_i\mid h,J=j]-\mathbb E_\xi[X_i\mid h]\|_2^2\\
 &\le2\sum_jp_j\log(1/p_j)\le2\log(N+1).
\end{align*}

For a transcript value $t$, its cell is a product of polyhedra independent
of $\xi$. Differentiating its Gaussian probability gives
\[
 \nabla_\xi\log\mathbb P_\xi(T_q=t)
 =\mathbb E_\xi\left[\sum_i(X_i-\xi)\,\middle|\,T_q=t\right].
\]
Differentiation is justified on each compact set of $\xi$ by an
integrable Gaussian times a linear polynomial. This is the usual
conditional-score identity~\cite{L82:journal}.
Let $S=\sum_i(X_i-\xi)$. Initially
$\operatorname{Tr}\operatorname{Cov}_\xi S=Nn$, and the product posterior
implies $\operatorname{Tr}\operatorname{Cov}_\xi(S\mid T_q)=V_\xi(T_q)$.
Therefore
\[
 \operatorname{Tr}\mathcal I_\xi(T_q)
 =\mathbb E_\xi\|\mathbb E_\xi[S\mid T_q]\|_2^2
 =Nn-\mathbb E_\xi V_\xi(T_q)
 \le2q\log(N+1),
\]
by telescoping the one-step estimate. An early-stopping policy can be padded
with the deterministic query $x=0$.
\end{proof}

The Gaussian prior turns the Fisher information bound into a bound on mutual
information. Let $\gamma_\tau=N(0,\tau^2I_n)$ and, for a nonnegative
integrable function $g$, define
$\operatorname{Ent}_{\gamma_\tau}(g)=
\int g\log(g/\int g\,d\gamma_\tau)\,d\gamma_\tau$,
with $\operatorname{Ent}_{\gamma_\tau}(0)=0$.
We use the normalization
$\operatorname{Ent}_{\gamma_\tau}(f^2)\le
2\tau^2\mathbb E_{\gamma_\tau}\|\nabla f\|_2^2$ of Gross' Gaussian
log-Sobolev inequality~\cite{G75:journal}.
\begin{proposition}[Mutual information in the transcript]
\label{prop:lower-information}
If $U$ is the independent private seed and $T_q$ is the transcript of any
policy making at most $q$ calls to $J$, then
\[
 I(\Xi;U,T_q)\le\tau^2q\log(N+1)
 =O\left(\frac{q\log^2 n}{n^2}\right).
\]
\end{proposition}
\begin{proof}
Fix $U=u$ and put $p_t(\xi)=\mathbb P_\xi(T_q=t\mid U=u)$.
Applying the Gaussian log-Sobolev inequality
to $\sqrt{p_t}$ yields
\[
 \operatorname{Ent}_{\gamma_\tau}(p_t)
 \le\frac{\tau^2}{2}\mathbb E_{\gamma_\tau}
 \bigl[p_t(\Xi)\|\nabla\log p_t(\Xi)\|_2^2\bigr].
\]
There are finitely many transcripts. Summing over them gives
$I(\Xi;T_q\mid U=u)$ on the left and the prior-averaged Fisher trace
times $\tau^2/2$ on the right. Positive-measure cells have positive Gaussian
probability for every $\xi$; null cells can be omitted. The bound in
\Cref{lem:lower-fisher} supplies the required integrability.
Average over $U$, use its independence of $\Xi$, and substitute
$\tau^2=16b^2/n^2$ and $\log(N+1)=O(\log n)$.
\end{proof}

\subsection{Recovering the hidden direction}
\label{subsec:lower-decoding}

We finish by showing that the centering vector returned by the rounding
algorithm identifies the hidden direction $v$.
Constant-probability recovery
of a uniform direction requires $\Omega(n)$ mutual information, which gives the
query lower bound.

The covariance lower bound on a successful output controls how far its
centering vector can lie from the true centroid.
\begin{lemma}[Recovery from a rounding map]
\label{lem:lower-decoder}
For all sufficiently large $n$, on the event in
\Cref{prop:lower-typical-moments}, every successful output $(T,z)$
in \Cref{thm:lower-rounding} satisfies $z\ne0$ and
\(\left\langle -\frac z{\|z\|_2},v\right\rangle\ge\frac12.\)
\end{lemma}
\begin{proof}
Write $m=m_{\xi,Z}$, $C=C_{\xi,Z}$, and $s=s_n$. Since
$TCT^{\mathsf T}\succeq I_n/2$ and $C\preceq(1+\eta_n)s^2I_n$,
\(\|T^{-1}\|_{\rm op}\le\sqrt{2(1+\eta_n)}\,s.\)
The centering guarantee then gives
$\|z-m\|_2\le\sqrt{(1+\eta_n)/2}\,s$. Consequently
\[
 \|z+csv\|_2\le
 \left(\eta_n+\sqrt{(1+\eta_n)/2}\right)s\le s
\]
for large $n$. Since $c\ge3.9$, this implies $z\ne0$ and
\(\left\langle-\frac z{\|z\|_2},v\right\rangle \ge\frac{c-1}{c+1}>\frac12.\)
\end{proof}

For completeness, we show that the information requirement for this recovery follows
directly from the exponentially small measure of a spherical cap.
\begin{lemma}[Information required to estimate a direction]
\label{lem:lower-direction-information}
Let $v\sim\Unif(S^{n-1})$, and let $\widehat v\in S^{n-1}$ be any
random vector. If $\mathbb P\{\langle v,\widehat v\rangle\ge1/2\}\ge p$ for a
fixed $p>0$, then $I(v;\widehat v)=\Omega_p(n)$.
\end{lemma}
\begin{proof}
A spherical cap $\{v:\langle v,w\rangle\ge1/2\}$ has uniform probability
$s_n^\star\le\exp(-c_0n)$, independently of the fixed unit vector $w$.
Under the product of the marginals of $v$ and $\widehat v$, the event
$E=\{\langle v,\widehat v\rangle\ge1/2\}$ therefore has probability
$s_n^\star$. Define the binary relative entropy by
$d(p\|q)=p\log(p/q)+(1-p)\log((1-p)/(1-q))$, with the usual
continuous interpretation at $p=0,1$. Data processing under the indicator
of $E$ gives
\[
 I(v;\widehat v)\ge d(\mathbb P(E)\|s_n^\star)
 \ge p\log(1/s_n^\star)-\log2=\Omega_p(n).
\]
\end{proof}

We now cap the query count and apply the information bounds to the
unconditioned Gaussian construction.
\begin{proof}[Proof of \Cref{thm:lower-rounding}]
Let $Q$ be the algorithm's worst-case expected query count on the promised
bodies. If $Q=\infty$ there is nothing to prove. Run it on the random body
\eqref{eq:lower-family} without conditioning its defining Gaussian variables.
Simulate membership using $J$ and stop after
$q=\max\{1,\lceil6Q\rceil\}$ queries. If it has returned an invertible
map with $z\ne0$, output $\widehat v=-z/\|z\|_2$; otherwise output a fixed
unit vector.

By \Cref{lem:lower-geometry} and
\Cref{prop:lower-typical-moments}, the body satisfies the required radius
bounds and has the stated centroid and covariance with probability
$1-o(1)$. For every
fixed body on this event, the original algorithm succeeds with probability
at least $2/3$. Markov's inequality makes the probability of exceeding the
query cap at most $1/6$. \Cref{lem:lower-decoder} therefore gives
\[
 \mathbb P\{\langle v,\widehat v\rangle\ge1/2\}\ge\frac12-o(1).
\]
This argument does not condition or otherwise change the Gaussian prior.
Since $v$ is a function of $\Xi$ and $\widehat v$ is a function of
$(U,T_q)$, \Cref{lem:lower-direction-information} and data processing
give $I(\Xi;U,T_q)=\Omega(n)$. \Cref{prop:lower-information}
then implies $q=\Omega(n^3/\log^2 n)$, and hence the same bound for $Q$.
Finally, $R_n/r_n=6\sqrt{n/(180\log n)}$ is polynomial in $n$.
\end{proof}

\section*{Acknowledgments}

We thank Santosh Vempala for suggesting this research direction.

\subsection*{LLM Usage Statement}

This project started classically among a subset of authors: YL, AS, and CT.
It has been significantly affected by the two most recent ChatGPT updates.
Those updates materially affected both the upper bound in \cref{sec:upper-bound}
and the lower bound in \cref{sec:lower-bound}.

As of April 2026, the status was an improved query bound
for symmetric convex bodies of about $n^{13/4 = 3.25}$
and a simplified version of the cooling path that matches
the previous $n^{7/2 = 3.5}$ bound from~\cite{JLLV26:journal,K23:journal}.

Shortly after the release of GPT 5.6 Sol in July 2026,
we obtained by direct prompting
an $n^{13/4 = 3.25}$ query algorithm for the asymmetric case
as well as a lower bound of about $n^3$ that is in \cref{sec:lower-bound}.
This improvement was also reproduced in early September 2026
by running K3max in Kimi Code.
Combining that with our earlier result led to a matrix-regularized
Gaussian cooling algorithm for symmetric convex bodies in about $n^3$
queries, and in turn a volume estimation routine taking
$\Ot(n^3\varepsilon^{-2})$ membership queries.
These two improved algorithms,
as well as a host of counterexamples to conjectures pointed out to us
by Santosh Vempala, are in a companion document hosted in a
repository on YL's homepage at
\url{https://yangpliu.github.io/repo/cubic-oracle-isotropic-rounding/paper.pdf}.

The symmetric counterexample to the ellipsoid-growth conjecture
in \cref{sec:symmetric-ellipsoid-growth-counterexample} was obtained
manually/classically around late August 2026,
with some motivation from the counterexamples mentioned above.

Finally, the algorithm in \cref{sec:upper-bound} was obtained by prompting
GPT 6 Astra to further improve the manuscript around mid September 2026.
We are in the process of investigating how this happened,
as the route it takes skips most of the steps in earlier versions,
yet shortly after obtaining it, our experiments by either directly prompting
or running GPT work on the statement were not able to duplicate this.

The authors take responsibility for all results and their presentations.

\bibliographystyle{alpha}
\bibliography{references}

\clearpage
\appendix
\section{Supporting estimates for rounding}\label{sec:rounding-support}
This appendix proves \Cref{lem:rounding-estimation-tools} and the
supporting estimates used in the subspace-growing and cooling arguments
of \Cref{sec:upper-bound}. We first compare moments along the Gaussian
schedule and place covariance ellipsoids in level sets. We then construct
initial distributions close to constant-warm laws and prove the moment
estimates, including the effect of the sampler's truncation. The final
lemma accounts for initialization errors and query caps in
\Cref{prop:logconcave-inner-rounder}.

\subsection{Gaussian cooling estimates}\label{subsec:rounding-cooling-estimates}
We bound the error at the end of cooling, compare moments when the
temperature is halved, and identify a level set containing a fixed
fraction of the covariance ellipsoid. We use the Gaussian family and
notation from \Cref{subsec:gaussian-cooling-upper}; in particular, the
underlying law $\mu$ is full-dimensional and logconcave.

At the endpoint, a small quadratic tilt approximately preserves every second moment
about every center. This gives a relative covariance comparison even when
the original support is unbounded.
\begin{lemma}[Removing the final Gaussian factor]\label{lem:cooling-endpoint}
Suppose $\E_\mu\norm{X-x_0}_2^2\le R^2$ and $tR^2\le\eta\le1/2$.
For $b\in\R^n$, define $M(b)=\E_\mu[(X-b)(X-b)^\top]$ and
$M_t(b)=\E_{\mu_t}[(X-b)(X-b)^\top]$.
There is a universal constant $C$ such that
$(1-C\eta)M(b)\preceq M_t(b)\preceq(1+C\eta)M(b)$.
Also $(1-C\eta)A_\mu\preceq A_t\preceq(1+C\eta)A_\mu$ and
$\norm{A_\mu^{-1/2}(m_t-m_\mu)}_2\le C\eta$.
\end{lemma}
\begin{proof}
Fix $u\in\R^n$, and write $G=(u^\top(X-b))^2$ and
$Q=\norm{X-x_0}_2^2$.
The fourth-moment bound from Equation~\eqref{eq:logconcave-moments},
applied after each translation, and Cauchy--Schwarz give
\[
 \E_\mu[GQ]
 \le\sum_{j=1}^n
   (\E_\mu G^2)^{1/2}
   (\E_\mu(X_j-x_{0,j})^4)^{1/2}
 \le CR^2\E_\mu G.
\]
Since $1-tQ/2\le e^{-tQ/2}\le1$ and
$1-\eta/2\le\E_\mu e^{-tQ/2}\le1$, division by the normalizer proves
the claimed scalar comparison. This holds for every $u$. For centered
covariance, minimize the scalar second moment over the centering scalar
in each direction.
For the mean, put $Y=u^\top(X-m_\mu)$ and $w=e^{-tQ/2}$.
Since $\E Y=0$, the same fourth-moment estimate gives
\[
 |\E_{\mu_t}Y|\le\frac{t\E(|Y|Q)}{2\E w}
 \le CtR^2\sqrt{u^\top A_\mu u}.
\]
Taking the supremum over directions proves the last assertion.
\end{proof}

For probability measures $\nu,\mu$, write
$D_2(\nu\Vert\mu)=\log\int(d\nu/d\mu)^2\,d\mu$ when $\nu\ll\mu$,
and set it to $+\infty$ otherwise.

We bound the covariance eigenvalues collectively. The argument compares
centroid bodies in each principal subspace; it loses only a logarithm when
the eigenvalue bounds are summed.

For a logconcave probability measure $\rho$ on $\R^d$ and $p\ge1$, its symmetric
$L_p$ centroid body is defined through its support function
$h_C(u)=\sup_{x\in C}\langle u,x\rangle$ by
\(h_{Z_p(\rho)}(u) =\left(\int |u^\top x|^p\,d\rho(x)\right)^{1/p}.\)
These bodies commute with linear images.
Equation~\eqref{eq:logconcave-moments} gives
$Z_q(\rho)\subseteq C(q/p)Z_p(\rho)$ for $1\le p\le q$.
We also use the standard centroid-volume bounds and the slicing theorem:
for every centered full-dimensional logconcave law $\rho$ on $\R^d$,
\begin{equation}\label{eq:centroid-volume}
 c\det(\Cov\rho)^{1/(2d)}
 \le |Z_d(\rho)|^{1/d}
 \le C\det(\Cov\rho)^{1/(2d)}.
\end{equation}
Here $|C|$ denotes the Lebesgue volume of a convex set $C$.
The two-sided estimate follows from the centroid-body volume bounds
and the slicing theorem, in its equivalent logconcave
form~\cite{GSTV15:journal,KL25:journal}.

The next lemma bounds the sum of the square roots of the covariance
eigenvalues in terms of R\'enyi divergence. Its mean estimate handles recentering
between consecutive temperatures.
\begin{lemma}[Covariance comparison]\label{lem:nuclear-transfer}
Let $\mu$ and $\nu$ be full-dimensional logconcave laws on $\R^n$, with
$\mu$ isotropic and $D=D_2(\nu\Vert\mu)<\infty$. If
$\lambda_1\ge\cdots\ge\lambda_n>0$ are the covariance eigenvalues of
$\nu$, then $\sqrt{\lambda_k}\le C(n+D)/k$ for $1\le k\le n$.
The mean also satisfies $\norm{\E_\nu X}_2\le C(n+D)$. Consequently,
\(\tr\sqrt{\Cov\nu}\le C(n+D)\log(2n).\)
\end{lemma}
\begin{proof}
Put $p=\max\{n,D,2\}$ and $m=\E_\nu X$. For every linear functional,
Cauchy--Schwarz applied to $d\nu/d\mu$ gives
\[
 \left(\E_\nu|u^\top X|^p\right)^{1/p}
 \le e^{D/(2p)}
       \left(\E_\mu|u^\top X|^{2p}\right)^{1/(2p)}.
\]
Write $\nu-m$ for the law of $X-m$ when $X\sim\nu$.
By Jensen and Minkowski, $Z_k(\nu-m)\subseteq2Z_k(\nu)$. Monotonicity of moments and the
logconcave moment comparison therefore yield, for $1\le k\le n$,
\begin{equation}\label{eq:centroid-transfer}
 Z_k(\nu-m)
 \subseteq2Z_p(\nu)
 \subseteq2e^{D/(2p)}Z_{2p}(\mu)
 \subseteq C\frac pk Z_k(\mu).
\end{equation}

Let $F$ be the span of eigenvectors corresponding to the $k$ largest
covariance eigenvalues of $\nu$, and let $\operatorname{proj}_F$ denote orthogonal projection
onto $F$. The projection of $\mu$ is isotropic, so
\eqref{eq:centroid-volume} gives, with volume measured in $F$,
\begin{align*}
 |\operatorname{proj}_F Z_k(\mu)|^{1/k}&\le C,\\
 |\operatorname{proj}_F Z_k(\nu-m)|^{1/k}
 &\ge c\left(\prod_{j=1}^k\lambda_j\right)^{1/(2k)}
 \ge c\sqrt{\lambda_k}.
\end{align*}
Taking volumes in the projection of \eqref{eq:centroid-transfer} proves
$\sqrt{\lambda_k}\le Cp/k$. Summing the harmonic series proves the
trace bound.

Finally, if $m\ne0$, take $u=m/\norm{m}_2$. The same change of measure,
followed by one-dimensional logconcave moment comparison, gives
\[
 \norm{m}_2\le\left(\E_\nu|u^\top X|^p\right)^{1/p}
 \le C\left(\E_\mu|u^\top X|^{2p}\right)^{1/(2p)}
 \le Cp\left(\E_\mu|u^\top X|^2\right)^{1/2}=Cp.
\]
Since $p\le2(n+D)$, the result follows.
\end{proof}

The next estimate applies to every temperature used by the algorithm,
including its initial high-temperature distribution. It bounds the change
in moments when the temperature is halved.
\begin{corollary}[Covariance comparison at half the temperature]
\label{cor:half-temperature-bridge}
Assume \eqref{eq:grounded-input}, with $h\ge1$ and $R\ge r$.
Put $t_0=64(n+1)(hn/r)^2$ and $L=\log(2+nR/r)$.
For every $0<t\le t_0$,
\(D_2(\mu_{t/2}\Vert\mu_t)\le C_hnL.\)
Consequently,
\begin{align*}
 \operatorname{tr}\sqrt{A_t^{-1/2}A_{t/2}A_t^{-1/2}}
 &\le C_hnL\log(2n),\\
 \|A_t^{-1/2}(m_{t/2}-m_t)\|_2
 &\le C_hnL.
\end{align*}
Moreover, $0\le a_t-a_{t/2}\le hn$.
\end{corollary}
\begin{proof}
Write $a=\inf V$. The normalizer identity and $Z_{t/2}\ge Z_t$ give
\[
 D_2(\mu_{t/2}\Vert\mu_t)
 =\log\frac{Z_tZ_0}{Z_{t/2}^2}
 \le\log\frac{Z_0}{Z_t}.
\]
Markov's inequality gives $\mu(B(x_0,\sqrt2R))\ge1/2$. Since
$e^{-V}\le e^{-a}$,
\(Z_0\le2e^{-a}\operatorname{vol}(B(x_0,\sqrt2R)).\)
Set $\rho=\min\{r,t^{-1/2}\}$. On $B(x_0,\rho)$, the input assumption
and $t\rho^2\le1$ imply
$V(x)+t\|x-x_0\|_2^2/2\le a+hn+1/2$. Hence
\(Z_t\ge e^{-a-hn-1/2}\operatorname{vol}(B(x_0,\rho)),\)
and therefore
$D_2(\mu_{t/2}\Vert\mu_t)
\le hn+\tfrac12+\log2+n\log(\sqrt2R/\rho)$.
For $t\le t_0$, $R/\rho\le8h n\sqrt{n+1}\,R/r$.
This proves the stated divergence bound.
Apply the affine map $x\mapsto A_t^{-1/2}(x-m_t)$ to both laws.
The image of $\mu_t$ is isotropic, so \Cref{lem:nuclear-transfer}
gives the two moment estimates.

Finally let $y$ minimize $\Phi_{t/2}$. Since
$\Phi_{t/2}(y)\le V(x_0)\le a+hn$ and $V(y)\ge a$, we have
$t\|y-x_0\|_2^2/4\le hn$. Evaluating $\Phi_t$ at $y$ gives
$a_t\le a_{t/2}+\tfrac t4\|y-x_0\|_2^2\le a_{t/2}+hn$.
The reverse inequality follows from $\Phi_t\ge\Phi_{t/2}$.
\end{proof}

The covariance ellipsoid can also be placed in a level set of the
potential. This is the form needed to preserve a contained ball when
rounding a nonuniform logconcave distribution.
\begin{lemma}[A covariance ellipsoid in a level set]
\label{lem:level-covariance-ellipsoid}
There is a universal constant $H$ with the following property.
Let $V$ be a proper closed convex potential such that
$\mu(dx)\propto e^{-V(x)}dx$ is a full-dimensional probability measure.
Put $m=\E_\mu X$, $A=\Cov(\mu)$, and $a=\inf V$.
Then $m+\tfrac12 A^{1/2}\bn\subseteq\{V\le a+Hn\}$.
\end{lemma}
\begin{proof}
Convexity and a change of scale about a minimizer of $V$ give
$\E_\mu e^{(V-a)/2}\le2^n$; an approximate minimizer gives the
same conclusion by passage to the limit. More explicitly, convexity
implies $V(x_*+2y)-a\ge2(V(x_*+y)-a)$ when $V(x_*)=a$.
Integrating its exponential proves the assertion.
Thus the complement of $C=\{V\le a+Hn\}$ has probability
$p\le e^{-(H/2-\log2)n}$.

Apply an affine transformation so that $m=0$ and $A=I$.
The fourth-moment bound for one-dimensional logconcave marginals gives
$\E\langle u,X\rangle^4\le C_4$ for every unit vector $u$.
For $\nu=\mu(\cdot\mid C)$, Cauchy--Schwarz therefore gives
$\|m_\nu\|_2\le\sqrt p/(1-p)$ and
$\E_\nu XX^{\mathsf T}\succeq(1-\sqrt{C_4p})I/(1-p)$.
Choosing the universal constant $H$ sufficiently large ensures
$\|m_\nu\|_2\le1/8$ and $\Cov(\nu)\succeq(3/4)I$.
The conditional measure is logconcave and
$\operatorname{supp}\nu\subseteq C$. By \Cref{lem:support-ball},
\[
 \tfrac12\bn\subseteq m_\nu+\sqrt{3/4}\bn
 \subseteq\operatorname{supp}\nu\subseteq C.
\]
Undo the affine transformation to conclude.
\end{proof}

\subsection{Initialization and moment estimation}
\label{subsec:rounding-initialization-moments}
We construct a fresh initial distribution for each sampling batch, then
prove the covariance and mean estimates used by subspace growing.
The final estimates account for the joint truncation in the sampler.

\subsubsection{Warm initial distributions}
We generate a new initial distribution whenever the algorithm changes
coordinates. The generator uses Gaussian cooling with small changes in
temperature. A bounded R\'enyi divergence supplies the constant warmness
needed in the analysis, up to a prescribed total variation error.

We use \Cref{thm:logconcave-warm-membership-sampling} to sample
from consecutive distributions along the Gaussian schedule. Only a fixed
warmness bound is needed in each call.

For $q>1$, the R\'enyi divergence of order $q$ is
$D_q(\nu\|\pi)=(q-1)^{-1}\log\int(d\nu/d\pi)^q\,d\pi$,
with value $+\infty$ unless $\nu$ is absolutely continuous with respect
to $\pi$. The next estimate approximates a distribution by a warm one
when their R\'enyi divergence is bounded.
\begin{lemma}[Approximation by a warm distribution]
\label{lem:warmness-clipping}
If $D_q(\nu\|\pi)\le1$ for some $q\ge2$, then there is a probability
measure $\nu'$ satisfying
$\|\nu-\nu'\|_{\mathrm{TV}}\le e^{-(q-1)}$ and $\nu'\le2e^2\pi$.
\end{lemma}
\begin{proof}
Put $f=d\nu/d\pi$ and $A=\{f\le e^2\}$. Since
$\int f^q\,d\pi\le e^{q-1}$,
\[
 \nu(A^c)=\int_{\{f>e^2\}}f\,d\pi
 \le e^{-2(q-1)}\int f^q\,d\pi\le e^{-(q-1)}.
\]
Conditioning $\nu$ on $A$ gives the result because $\nu(A)>1/2$.
The conditional measure is used only in a coupling argument and need not
be computed by the algorithm.
\end{proof}

The Gaussian variance controls the change of measure along the cooling
schedule. The second moment of the target distribution is the
only other geometric quantity in the estimate.
\begin{lemma}[A fine cooling step]
\label{lem:fine-cooling-renyi}
Let $\mu(dx)\propto e^{-V(x)}dx$ be logconcave with
$\E_\mu\|X\|_2^2\le R^2$, and set
$\mu_t(dx)\propto e^{-t\|x\|_2^2/2}\mu(dx)$.
For $0<s<t$ and $q\ge2$,
\(D_q(\mu_t\|\mu_s)\le\frac{qR^2(t-s)^2}{2s}.\)
\end{lemma}
\begin{proof}
Write $F(u)=\log\int e^{-V(x)-u\|x\|_2^2/2}\,dx$. Differentiation gives
\[
 F''(u)=\tfrac14\operatorname{Var}_{\mu_u}(\|X\|_2^2)
 \le\frac1u\mathbb E_{\mu_u}\|X\|_2^2\le\frac{R^2}{u}.
\]
The first inequality is Equation~\eqref{eq:strong-logconcave-poincare}
with $\alpha=u$ and $f(x)=\|x\|_2^2$.
The second follows because
$\mathbb E_{\mu_u}\|X\|_2^2$ is decreasing in $u$, as differentiation shows.
Direct integration of the likelihood ratio gives, with $\Delta=t-s$,
\[
 D_q(\mu_t\|\mu_s)
 =\frac{F(t+(q-1)\Delta)+(q-1)F(t-\Delta)-qF(t)}{q-1}.
\]
On the interval involved, $F''\le R^2/s$. Taylor's formula cancels the
constant and linear terms and bounds the numerator by
$q(q-1)\Delta^2R^2/(2s)$.
\end{proof}

We first sample a narrow Gaussian tilt by rejection.
We then lower its temperature while keeping consecutive distributions
close to constant-warm initial laws. This supplies a fresh initial
distribution in any fixed coordinates.

Let $\mu(dx)\propto e^{-V(x)}dx$ be full-dimensional, with a proper
closed convex potential. Suppose
$B(x_0,r)\subseteq\{V\le\inf V+hn\}$ and
$\E_\mu\|X-x_0\|_2^2\le R^2$, where $h\ge1$ is a fixed constant
and $R\ge r>0$ are known.
Translate $x_0$ to zero and subtract $V(x_0)$ from $V$, so $V(0)=0$.
Use the Gaussian family from \Cref{subsec:gaussian-cooling-upper} in these translated coordinates.

The next rejection sampler starts without any prior sample from $\mu$.
The point $x_0$ need not minimize the potential.
\begin{lemma}[Rejection sampling at high temperature]
\label{lem:high-temperature-rejection}
Put $L=hn/r$ and $t_0=64(n+1)L^2$. There is an algorithm returning
$X\sim\mu_{t_0}$ using $O(1)$ expected evaluations of $V$.
Independent repetitions give independent samples.
\end{lemma}
\begin{proof}
Since the domain of $V$ contains a neighborhood of zero, a subgradient
$g\in\partial V(0)$ exists. For every unit vector $u$,
$V(\pm ru)\le hn$. The subgradient inequalities therefore give
$|\langle g,u\rangle|\le hn/r=L$. Consequently,
$V(x)\ge\langle g,x\rangle\ge-L\|x\|_2$ for every $x\in\R^n$.
Convex interpolation between zero and the point $rx/\|x\|_2$ also gives
$V(x)\le L\|x\|_2$ whenever $\|x\|_2\le r$.

Set $\sigma^2=(1+1/n)/t_0$ and
$\Delta=t_0-\sigma^{-2}=t_0/(n+1)$. Draw $G\sim N(0,\sigma^2I)$
and accept it with probability
\(\exp\left(-V(G)-\frac\Delta2\|G\|_2^2-\frac{L^2}{2\Delta}\right).\)
Reject when $V(G)=+\infty$. The exponent is nonpositive, because
$-V(x)-\Delta\|x\|_2^2/2\le L\|x\|_2-\Delta\|x\|_2^2/2\le L^2/(2\Delta)$.
Multiplying the proposal density by this acceptance probability gives
a constant multiple of the density of $\mu_{t_0}$.

We have $\mathbb P\{\|G\|_2^2\le2n\sigma^2\}\ge1/2$.
On this event,
$L\|G\|_2\le\sqrt{2(n+1)L^2/t_0}<1/4$ and
$\Delta\|G\|_2^2/2\le1$. The definition of $\Delta$ also gives
$L^2/(2\Delta)=1/128$. Since $\|G\|_2\le1/(4L)\le r$,
we have $V(G)\le1/4$.
Thus acceptance probability is at least $e^{-2}/2$. Each proposal
uses one potential evaluation, and ordinary rejection sampling proves
the assertion.
\end{proof}

At temperatures above $r^{-2}$, the ball $B(0,t^{-1/2})$ remains in a
level set whose height is $O(hn)$ above the minimum. Rescaling this ball
to the unit ball gives the hypotheses of
\Cref{thm:logconcave-warm-membership-sampling}.
\begin{lemma}[Bounds at high temperature]
\label{lem:high-temperature-cooling}
For every $t>0$, we have $t\E_{\mu_t}\|X\|_2^2\le2(h+1)n$ and
$\Var_{\mu_t}(\|X\|_2^2)\le8(h+1)n/t^2$.
For $t\ge r^{-2}$,
$B(0,t^{-1/2})\subseteq\{\Phi_t\le a_t+hn+1/2\}$ and
$A_t\preceq t^{-1}I_n$.
For $q\ge2$ and $s=(1-\theta)t$ with $0<\theta\le1/2$,
\(D_q(\mu_t\|\mu_s)\le C(h+1)nq\theta^2.\)
\end{lemma}
\begin{proof}
Put $a=\inf V$, so $-hn\le a\le0$. A change of scale about a minimizer of
$\Phi_t$ gives $J(s)\le s^{-n}J(1)$ for $0<s\le1$, where
$J(s)=\int e^{-s(\Phi_t(x)-\inf\Phi_t)}\,dx$.
Taking the left derivative of this logarithmic inequality at $s=1$
gives $\E_{\mu_t}\Phi_t\le\inf\Phi_t+n\le n$, since $\Phi_t(0)=0$.
Because $V\ge a$, it follows that
$t\E\|X\|_2^2/2\le n-a\le(h+1)n$.
Brascamp--Lieb, applied to the $t$-strongly convex potential
$\Phi_t$, gives
\[
 \Var_{\mu_t}(\|X\|_2^2)
 \le\frac4t\E_{\mu_t}\|X\|_2^2
 \le\frac{8(h+1)n}{t^2},
\]
and also $\Cov(\mu_t)\preceq t^{-1}I_n$.
These statements extend to proper convex potentials by convex approximation.

When $\|x\|_2\le t^{-1/2}\le r$, we have
$\Phi_t(x)\le a+hn+1/2\le\inf\Phi_t+hn+1/2$.
This proves the stated ball inclusion.

Let $F(u)=\log\int e^{-V(x)-u\|x\|_2^2/2}dx$.
Then $F''(u)=\Var_{\mu_u}(\|X\|_2^2)/4\le2(h+1)n/u^2$.
Writing $\Delta=t-s$, direct integration gives
\(D_q(\mu_t\|\mu_s) =\frac{F(t+(q-1)\Delta)+(q-1)F(s)-qF(t)}{q-1}.\)
On the interval in this expression, $u\ge s\ge t/2$.
Taylor's formula cancels the constant and linear terms, giving
\[
 D_q(\mu_t\|\mu_s)
 \le\frac{q\Delta^2}{2}\sup_{u\ge s}F''(u)
 \le C(h+1)nq\theta^2.
\]
\end{proof}

The preceding estimates give a cooling schedule with two stages.
We obtain a fixed query bound by capping independent runs and returning
the first that finishes. The proof also accounts for conditioning on
completion: the returned law is close in total variation to a
constant-warm law.
\begin{proposition}[Constructing a warm initial distribution]
\label{prop:logconcave-fresh-initialization}
Under \eqref{eq:grounded-input}, with fixed $h\ge1$ and known $R\ge r>0$,
for every $0<\varepsilon<1/4$ there is an algorithm using at most
$\Ot(n^{5/2}+n^2R/r)$ potential evaluations and returning $\widehat X$
such that, for some probability measure $\sigma\le2e^{1/2}\mu$,
$\|\mathcal L(\widehat X)-\sigma\|_{\TV}\le\varepsilon$.
All randomness can be chosen independently of the data used to fix
$V,x_0,r,R$.
\end{proposition}
\begin{proof}
Put $\theta=\varepsilon/8$ and take
$q=\lceil C_h\log((n+2)(2+R/r)/\theta)\rceil$ for a sufficiently
large constant. One uncapped run begins with the exact rejection sample
from \Cref{lem:high-temperature-rejection}.
While $t>r^{-2}$, decrease $t$ by the factor
$1-c_h/\sqrt{nq}$, except for a last step ending exactly at $r^{-2}$.
\Cref{lem:high-temperature-cooling} gives
$D_q(\mu_t\|\mu_s)\le1$ for each consecutive pair $t>s$.
The number of steps satisfies $J_1\le C_h\sqrt{nq}\log(2+hn)$.
For the next target $\mu_s$, the map $x\mapsto\sqrt s\,x$ sends its known ball
to $B_2^n$. The bounds $sA_s\preceq I_n$ and
$s\mathbb E_{\mu_s}\|X\|_2^2\le2(h+1)n$ give the required covariance
and second-moment bounds, and the level-set height remains $O(hn)$. By \Cref{thm:logconcave-warm-membership-sampling},
each call from a distribution with constant warmness costs
$\Ot(n^2)$ evaluations in expectation.

From $r^{-2}$ to $R^{-2}$ use
$t_{j+1}=\max\{R^{-2},t_j-c\sqrt{t_j}/(R\sqrt q)\}$.
\Cref{lem:fine-cooling-renyi} gives
$D_q(\mu_{t_j}\|\mu_{t_{j+1}})\le1$.
Throughout this part of the schedule,
$B(0,r)\subseteq\{\Phi_t\le a_t+hn+1/2\}$ and
$A_t\preceq t^{-1}I_n$, so each sampling call costs
$\Ot(n^2/(r\sqrt t))$ evaluations.
Except for the final step,
$c/(2R\sqrt q)\le\sqrt{t_j}-\sqrt{t_{j+1}}\le c/(R\sqrt q)$.
Thus the number of calls satisfies $J_2\le C(R/r)\sqrt q$.
Summing the reciprocals of these square roots gives
$\sum_jt_{j+1}^{-1/2}\le CR\sqrt q\log(2+R/r)$.
When $R=r$, this part of the schedule is empty.

Apply \Cref{lem:warmness-clipping} at each transition. The exact
preceding target is within $e^{-(q-1)}$ of a law that is $2e^2$-warm
for the next target. Choose each sampling error at most
$\theta/(4(J_1+J_2+1))$. The choice of $q$ makes the sum of the
clipping errors at most $\theta/4$.

We couple the endpoint and the work of the complete run simultaneously.
In the analysis, maintain a reference point with the exact target law
at each temperature. Couple this point to the clipped warm initial law
for the next target, run a comparison sampler from that law, and couple
its output to a reference point with the next exact target law.
The initial rejection sample supplies the first reference point exactly.
Use the same sampler randomness in the actual and comparison runs
whenever their initial states agree. A union bound over the clipping
and sampling couplings gives an event $G$ with
$\mathbb P(G^c)\le\theta$ on which every actual and comparison call
agrees, including its query count, and the final output equals the final
reference point. Continue the reference and comparison construction even
after an earlier coupling fails. Each reference point therefore retains
its exact target marginal, and each comparison input retains its prescribed
warm marginal; the comparison calls need not be independent.

Write $X$ for the uncapped actual output, $Y$ for the final reference
point, and $Q$ for the actual query count. Then
$Y\sim\pi_*:=\mu_{R^{-2}}$. The expected sum of comparison query
counts, including the initial rejection sampler, has a known upper bound
$Q_*=\Ot(n^{5/2}+n^2R/r)$ by the preceding schedule estimates.
On $G$, this sum equals $Q$. Consequently,
$\mathbb P\{Q>4Q_*\}\le1/4+\theta$.

Run the actual algorithm with a cap of $\lceil4Q_*\rceil$ evaluations,
stopping an attempt before any query that would exceed this cap.
Let $A$ be its completion event and write $p=\mathbb P(A)$.
The preceding bound gives $p\ge3/4-\theta\ge1/2$.
Let $\sigma=\mathcal L(Y\mid A\cap G)$.
Since $\mathbb P(A\cap G)\ge p-\theta\ge1/2$, we have
$\sigma\le2\pi_*$. On $G$, the outputs $X$ and $Y$ agree, so
$\sigma=\mathcal L(X\mid A\cap G)$ as well.
It follows that
$\|\mathcal L(X\mid A)-\sigma\|_{\TV}
\le\mathbb P(G^c\mid A)\le\theta/p\le2\theta$.

Repeat the capped attempt independently at most
$K=\lceil\log_2(2/\varepsilon)\rceil$ times, returning the output of
the first completed attempt. If none completes, return $x_0$.
Denote the returned point by $\widehat X$.
Conditional on at least one completion, the returned law is exactly
$\mathcal L(X\mid A)$, since the attempts are independent and identical.
The probability of no completion is at most $(1-p)^K\le\varepsilon/2$.
Thus $\|\mathcal L(\widehat X)-\sigma\|_{\TV}
\le2\theta+\varepsilon/2\le\varepsilon$.
The deterministic query count is at most
$K\lceil4Q_*\rceil=\Ot(n^{5/2}+n^2R/r)$.

Finally, Jensen's inequality gives
$Z_{R^{-2}}/Z_0=\int e^{-\|x\|_2^2/(2R^2)}\,d\mu(x)\ge e^{-1/2}$.
Hence $\pi_*\le e^{1/2}\mu$, so $\sigma\le2e^{1/2}\mu$ as claimed.
\end{proof}

The rejection sampler can also provide the first covariance-estimation
batch directly. At $t_0$ its independent outputs are exact; a batch of
$\Ot(n)$ samples costs $\Ot(n)$ evaluations in expectation. This avoids
any dependence issue in the initial empirical affine transformation.

\subsubsection{Covariance and mean estimates}
\label{subsec:upper-moments}
After generating a new initial distribution, we keep the current affine
map fixed while collecting its estimation batch. Matrix concentration for
Markov chains gives the same sample count, up to logarithms, as independent
sampling. We first prove concentration for a fixed chain, then bound
the effect of conditioning and combine the two estimates with the
sampler's joint-distribution guarantee.

For a Markov chain $(Z_j)$ with invariant law $\pi$, let
$Pf(z)=\mathbb E[f(Z_{j+1})\mid Z_j=z]$ be its transition operator
and $\Pi f=\int f\,d\pi$ the projection onto constant functions.
The operator norm of $P-\Pi$ is taken on $L^2(\pi)$.
The matrix Bernstein inequality of Neeman, Shi, and
Ward~\cite[Theorem 2.6 and Corollary 2.8]{NSW24:journal} applies
to stationary chains with $\|P-\Pi\|_{L^2(\pi)\to L^2(\pi)}\le1/2$.
For a centered observable $H$ taking values in symmetric $d\times d$
matrices, with $\|H\|_\op\le L$ almost surely and
$\|\mathbb E_\pi H^2\|_\op\le v$, it gives
\[
 \mathbb P\left\{\left\|\frac1k\sum_{j=1}^kH(Z_j)\right\|_\op\ge u\right\}
 \le C d^2\exp\left(-\frac{cku^2}{v+Lu}\right),
\]
where $d$ is the matrix dimension. The result does not require reversibility.
Starting from an $M$-warm law multiplies the right side by at most $M$.

The next consequence supplies the additive estimates used in the subspace
rounding algorithm. The full Markov state may include velocity; only the
position marginal is required to be logconcave.
\begin{lemma}[Covariance and mean estimates from a Markov chain]
\label{lem:additive-moment-estimate}
Let $(Z_j)_{j\ge0}$ be a Markov chain with invariant law $\pi$ and
transition operator $P$ satisfying
$\|P-\Pi\|_{L^2(\pi)\to L^2(\pi)}\le1/2$.
Assume $\mathcal L(Z_0)\le M\pi$ for fixed $M$.
Suppose that, under $\pi$, its position observable has a logconcave law with mean $m$ and
covariance $A$.
For $a,b>0$ and $0<\delta<1/4$, let $\widehat m$ and $\widehat A$
be the sample mean and sample covariance of a path of
\(k=\widetilde O\left(1+\frac{\operatorname{tr}A}{a}
  +\frac{\operatorname{tr}A}{b^2}\right)\)
retained positions. Then
$\tfrac12A-aI\preceq\widehat A\preceq2A+aI$ and
$\|\widehat m-m\|_2\le b$ with probability at least $1-\delta$.
\end{lemma}
\begin{proof}
First consider a stationary path, and let $X$ denote its position
observable under $\pi$. In the covariance analysis set
$Y=(A+aI)^{-1/2}(X-m)$, $C=\operatorname{Cov}(Y)$, and
$d_0=\max\{1,\operatorname{tr}C\}$. Then
$\|C\|_\op\le1$ and $d_0\le1+\operatorname{tr}(A)/a$.
Set $L=C_0\sqrt{d_0}\log(10k/\delta)$.
Standard logconcave norm tails give
$\mathbb P\{\|Y\|_2>L\}\le\delta/(10k)$.
Let $Y_L=Y\mathbf1_{\{\|Y\|_2\le L\}}$.
Increasing $C_0$ ensures $\|\mathbb E Y_L\|_2\le1/64$ and
$\|\mathbb E[Y_LY_L^{\mathsf T}]-C\|_\op\le1/64$.

The covariance matrix observable
$H=Y_LY_L^{\mathsf T}-\mathbb E[Y_LY_L^{\mathsf T}]$ satisfies
$\|H\|_\op\le L^2+1$ and $\|\mathbb E H^2\|_\op\le Cd_0$.
For the latter estimate, one-dimensional logconcave fourth-moment bounds
give, for every unit vector $u$ and an orthonormal basis $(e_i)$,
\[
 \mathbb E[\|Y\|_2^2\langle u,Y\rangle^2]
 =\sum_i\mathbb E[\langle e_i,Y\rangle^2\langle u,Y\rangle^2]
 \le C\operatorname{tr}(C)\,u^{\mathsf T}Cu\le Cd_0.
\]
Truncation can only decrease this positive quadratic form. Centering the
matrix observable subtracts the square of its mean from its second moment.
Matrix Bernstein therefore controls the empirical second moment in
operator norm from $\widetilde O(d_0)$ retained positions.

For the mean of $Y_L$, apply the same inequality to the symmetric dilation
\[
 \begin{pmatrix}
  0&(Y_L-\mathbb E Y_L)^{\mathsf T}\\
  Y_L-\mathbb E Y_L&0
 \end{pmatrix}.
\]
Its operator norm is $\|Y_L-\mathbb E Y_L\|_2\le2L$, and the operator
norm of its second moment is
$\operatorname{tr}\operatorname{Cov}(Y_L)\le\operatorname{tr}C\le d_0$.
Let $Y_1,\ldots,Y_k$ be the transformed retained positions.
With the same sample count, the matrix estimates and a union bound over
the truncation events give
\[
 \left\|\frac1k\sum_{j=1}^kY_j\right\|_2\le\frac18
 \text{ and }
 \left\|\frac1k\sum_{j=1}^kY_jY_j^{\mathsf T}-C\right\|_\op\le\frac18.
\]
Subtracting the outer product of the empirical mean gives
$\|(A+aI)^{-1/2}(\widehat A-A)(A+aI)^{-1/2}\|_\op\le1/4$.
Conjugation by $(A+aI)^{1/2}$ gives
$\tfrac34 A-\tfrac14aI\preceq\widehat A
\preceq\tfrac54 A+\tfrac14aI$, which is stronger than required. The regularizing transformation is only
used in the analysis; the estimates themselves are the ordinary sample
mean and sample covariance.

The Euclidean mean bound follows by the same vector-dilation argument
applied to $(X-m)/b$. Its trace covariance is $\operatorname{tr}(A)/b^2$, and
$\widetilde O(1+\operatorname{tr}(A)/b^2)$ retained positions suffice. Finally, domination of
the law of the entire Markov chain path by $M$ times its stationary path law transfers
all failure bounds from stationarity. Replace $\delta$ by $\delta/M$ in
the stationary calculation.
\end{proof}

The sampler works on an enlarged state space that includes auxiliary
coordinates and velocity. Its invariant position marginal is obtained by
conditioning on a truncation event. We bound the resulting changes in
mean and covariance before applying \Cref{lem:additive-moment-estimate}.

The event in the following statement may depend on auxiliary coordinates
as well as on the position.
\begin{lemma}[Moments under conditioning]
\label{lem:joint-conditioning-moments}
Let $X$ have a full-dimensional logconcave law with mean $m$ and covariance
$A$. On any joint probability space extending the law of $X$, let $E$ be
an event with $\mathbb P(E)=1-p$, where $0\le p\le1/4$.
Set $m_E=\mathbb E[X\mid E]$ and $A_E=\operatorname{Cov}(X\mid E)$.
Then $\|A^{-1/2}(m_E-m)\|_2\le\sqrt p/(1-p)$ and
$(1-C\sqrt p)A\preceq A_E\preceq(1+C\sqrt p)A$
for an absolute constant $C$.
\end{lemma}
\begin{proof}
Apply an affine change of coordinates so that $m=0$ and $A=I$.
For every unit vector $u$, Cauchy--Schwarz gives
\(|\mathbb E[\langle u,X\rangle\mathbf1_{E^c}]|\le\sqrt p.\)
Since the unconditional mean vanishes, division by $1-p$ and taking the
supremum over $u$ prove the mean bound. The one-dimensional logconcave
fourth-moment bound also gives
\[
 0\le\mathbb E[\langle u,X\rangle^2\mathbf1_{E^c}]
 \le\sqrt{p\,\mathbb E\langle u,X\rangle^4}\le C\sqrt p.
\]
Consequently,
$(1-C\sqrt p)I/(1-p)\preceq\mathbb E[XX^{\mathsf T}\mid E]
\preceq I/(1-p)$.
Subtracting the outer product of the conditional mean and adjusting the
absolute constant proves the covariance bounds. Neither application of Cauchy--Schwarz requires $E$
to be determined by $X$.
\end{proof}

The construction in \Cref{lem:truncated-epigraph-parameters} makes
the stationary position marginal of the sampler the law of $X$ conditioned
on an event in an enlarged probability space. Its complement has
probability at most a prescribed $p$. Thus
\Cref{lem:joint-conditioning-moments} applies to this construction;
no moment comparison is inferred from total variation distance alone.

We record the tolerances used in \Cref{prop:logconcave-inner-rounder}.
Here $A$ is the covariance before any subspace-growing step, and $A_i$
and $m_i$ are the covariance and mean in the current coordinates.
Set $\ell=\lceil\log_2(2+\operatorname{tr}\sqrt A)\rceil$, and let
$0<\eta<1/100$ be the estimation tolerance. Suppose
$\|A_i\|_\op\le C\ell(\operatorname{tr}\sqrt A)^2$.
Let $m_{i,E}$ and $A_{i,E}$ be the mean and covariance after conditioning
on the truncation event. Choosing its excluded probability so that
$p\le c\eta^4/(1+\ell)^2$ ensures
$\|m_{i,E}-m_i\|_2\le\eta\operatorname{tr}\sqrt A/4$ and
$\tfrac{15}{16}A_i\preceq A_{i,E}\preceq\tfrac{17}{16}A_i$.
Apply \Cref{lem:additive-moment-estimate}
with tighter constant errors to this marginal. Its sample mean
$\widehat m_i$ and sample covariance $\widehat A_i$ satisfy
\begin{equation}\label{eq:growth-moment-estimates}
 \begin{gathered}
 \tfrac12A_i-\eta^2(\operatorname{tr}\sqrt A)^2I
 \preceq\widehat A_i
 \preceq2A_i+\eta^2(\operatorname{tr}\sqrt A)^2I,\\
 \|\widehat m_i-m_i\|_2\le\eta\operatorname{tr}\sqrt A
 \end{gathered}
\end{equation}
from $\widetilde O(1+\operatorname{tr}(A_i)/
 (\eta^2(\operatorname{tr}\sqrt A)^2))$ retained positions.
The error in the law of the simulated sample sequence is charged separately
by total variation coupling. The truncation parameter is fixed throughout
a batch.

At the final estimation stage, use the same matrix-concentration proof
after multiplying centered positions by the inverse square root of
the truncated marginal's covariance. From
$\widetilde O(n)$ retained positions, with a sufficiently small fixed
truncation mass and fixed relative moment errors, the sample mean
$\widehat m_i$ and sample covariance $H_i$ satisfy
\begin{equation}\label{eq:logconcave-final-estimates}
 \tfrac78 A_i\preceq H_i\preceq\tfrac98 A_i
 \text{ and }\|A_i^{-1/2}(\widehat m_i-m_i)\|_2\le\tfrac1{16}.
\end{equation}
In every batch, choose the total variation tolerance for the joint law
of the simulated samples to be at most its allocated failure probability.
This change of coordinates is used only in the analysis: the estimators
are the ordinary sample mean and sample covariance.
For $X\sim\mu$, the returned map
$F(x)=H_i^{-1/2}(F_ix-\widehat m_i)$ therefore satisfies
$\tfrac34I\preceq\operatorname{Cov}(F(X))\preceq\tfrac43I$ and
$\|\mathbb EF(X)\|_2\le1/8$.
These stricter internal bounds allow the final Gaussian factor to be
removed.

We now deduce the moment-estimation statement used in the main algorithm.
The sampling theorem compares an entire batch with one stationary chain,
so the concentration bound applies without an independence assumption.
\begin{proof}[Proof of \Cref{lem:rounding-estimation-tools}]
Translate the known center to zero and rescale its ball to unit radius.
The covariance and tolerances become $A/\rho^2$, $a/\rho^2$, and
$b/\rho$, leaving $\operatorname{tr}(A)(a^{-1}+b^{-2})$ unchanged.
We map the sampled positions back before estimating their moments;
all moment comparisons below use the original coordinates.
Apply \Cref{thm:logconcave-warm-membership-sampling} with a joint total
variation tolerance at most $\delta/4$. Its stationary position marginal
is obtained by conditioning the exact target on an event in an enlarged
probability space. By choosing the excluded probability $p$ sufficiently
small, \Cref{lem:joint-conditioning-moments} makes the mean and covariance
errors as small as required. For the additive estimates, one may take
$p\le c\min\{1,a^2/\|A\|_\op^2,b^2/\|A\|_\op\}$, with a sufficiently small
universal $c$. In the notation of \Cref{lem:joint-conditioning-moments},
the conditional covariance satisfies $\operatorname{tr}A_E\le
C\operatorname{tr}A$.

Apply the proof of \Cref{lem:additive-moment-estimate}, with tighter
absolute constants, to the stationary chain. It gives the stated
additive errors from
$\widetilde O(1+\operatorname{tr}(A)(a^{-1}+b^{-2}))$
positions after accounting
for conditioning. For the relative estimates, use a sufficiently small
constant $p$ and apply the same proof after multiplying centered positions
by $A_E^{-1/2}$. The sample count is then
$\widetilde O(n)$, as in \eqref{eq:logconcave-final-estimates}.
Only the analysis uses these unknown coordinate transformations.

Take the concentration failure probability at most $\delta/2$.
Coupling the simulated batch with the stationary one adds at most
$\delta/4$ to this probability. The algorithm returns $\widehat m$ and
$\widehat A$, the empirical mean and covariance of the $k$ simulated
positions in the original coordinates. Since
$A\preceq\kappa\rho^2I$, the sampler's expected cost after rescaling
is $\widetilde O(n^2k\sqrt\kappa)$. Taking $k=\widetilde O(n)$ gives
the relative-estimation bound.
\end{proof}

The affine coordinates depend on earlier samples, so each batch uses
fresh randomness after those coordinates are fixed. We repeat a capped
batch when its sampler runs out of queries. We bound the probability
that any completed attempt returns incorrect moment estimates.
\begin{lemma}[Initialization and query budgets for subspace growing]
\label{lem:subspace-initialization-and-budgets}
Consider the algorithm in \Cref{subsec:subspace-growing}, with the
parameters defined there and the input assumptions of
\Cref{prop:logconcave-inner-rounder}.
Suppose \eqref{eq:logconcave-inner-trace} and
\eqref{eq:subspace-second-moment} hold after every successful preceding
batch. For each batch, let $Q_i$ be a deterministic upper bound on its
expected sampling work from a $2e^{1/2}$-warm initial law, using the
tolerances specified below.
Then the batches can be implemented with total failure probability at
most $\delta$ and total query count at most
\[
 \Ot\left(n^{5/2}+n^2(\operatorname{tr}\sqrt A+\|m\|_2)
              +\sum_iQ_i\right).
\]
The work bound holds even after an earlier failed estimate.
\end{lemma}
\begin{proof}
Put $N=2\ell+1$, $\beta=\delta/(2N)$,
$K=\lceil\log_2(1/\beta)\rceil$, and $\alpha=\beta/(4K)$.
Conditional on a successful preceding history, the current potential,
coordinates, and known ball are fixed.
Use the growth or final moment tolerances specified before
\eqref{eq:growth-moment-estimates} and
\eqref{eq:logconcave-final-estimates}, with statistical and joint
simulation errors each at most $\alpha$.

For each batch, make at most $K$ independent attempts in these same
coordinates. In an attempt, first use
\Cref{prop:logconcave-fresh-initialization} with error $\alpha$ and
its deterministic query budget, using fresh randomness.
Then run the moment-estimation sampler with a cap of $8Q_i$ queries.
Return the first batch that finishes within this cap; if no attempt
finishes, declare failure.

Couple an attempt's initial point to a point from the warm law with
disagreement probability at most $\alpha$, and use the same subsequent
randomness. The warm comparison run has expected sampling work at most
$Q_i$, so it exceeds the cap with probability at most $1/8$.
Thus the actual attempt fails to finish with probability at most
$1/8+\alpha\le1/4$. The warm run's moment estimates fail with probability
at most $2\alpha$. On agreement of the initial points, both runs are
identical through the cap. Therefore the unconditional probability that
the actual attempt finishes with incorrect estimates is at most
$3\alpha$.

Independence of the attempts bounds the probability that none finishes
by $4^{-K}\le\beta$. A union bound over the $K$ attempts bounds the
probability of returning incorrect estimates by
$3K\alpha\le\beta$. This union bound does not condition on completion.
The completed batch is used only for its estimates; the next batch
receives a fresh initial distribution.
The conditional failure probability of a batch is consequently at most
$2\beta$. Summing over the first failed batch gives total failure
probability at most $2N\beta=\delta$.

By \eqref{eq:subspace-second-moment}, each initialization costs
$\Ot(n^{5/2}+n^2(\operatorname{tr}\sqrt A+\|m\|_2))$ deterministically.
Here the additive constant in the second-moment bound contributes
$O(n^2)$, absorbed by $n^{5/2}$.
Every attempt also makes at most $8Q_i$ sampling queries.
There are at most $NK$ initialization calls and at most $K$ sampling
attempts per batch. Since $N$ is logarithmic in the moment bound and
$K=O(\log(N/\delta))$, their total work has the claimed bound.

All attempt counts, sample sizes, and caps use the asserted moment
bounds and the deterministic stage radius, rather than unverified
empirical bounds. They therefore also bound work after a failed
estimate. If a proposed center has infinite potential, a computed
rejection probability lies outside $[0,1]$, or a final covariance to
be inverted is not positive definite, declare failure.
\end{proof}

\begingroup
\renewcommand{\E}[2]{\mathbb{E}_{#1}\left[#2\right]}
\newcommand{\V}[2]{\Var_{#1}\left[#2\right]}

\section{A symmetric counterexample to ellipsoid growth}
\label{sec:symmetric-ellipsoid-growth-counterexample}
This appendix gives a human-generated symmetric counterexample to the
ellipsoid-growth conjecture of Jia, Laddha, Lee, and
Vempala~\cite{JLLV26:journal}.
We first state the conjectured property and the counterexample, then prove
the geometric construction using estimates for intersections of cubes and
balls.
The remaining subsections prove the density estimates used in the
construction.

For a convex body $L$, write $A_L=\Cov(\Unif(L))$.
All bodies whose covariance we compare below are symmetric about the
origin, so their means vanish and their covariances equal their uncentered
second moments.
Such a body is isotropic when $A_L=I_n$.
It is near-isotropic in the sense of \Cref{def:rounding} when
$\tfrac12 I_n\preceq A_L\preceq2I_n$; the centering condition is automatic.

We show that doubling an ellipsoid can take an isotropic intersection
outside these covariance bounds in the same coordinates.
The construction also disproves the weaker trace bound conjectured
in~\cite{JLLV26:journal}.
We state that bound explicitly for symmetric bodies with an exactly
isotropic intersection.
\begin{conjecture}
\label{conj:jllv-ellipsoid-growth}
Let $K\subset\R^n$ be a convex body symmetric about the origin, and let
$E\subset\R^n$ be an ellipsoid centered at the origin.
If $A_{K\cap E}=I_n$, then $\tr A_{K\cap2E}=\Ot(n)$.
\end{conjecture}

The next theorem shows that even a smaller enlargement can increase a
directional variance by a factor of order $(n\lambda)^2$.
After an affine change of coordinates, a constant enlargement therefore
takes an isotropic intersection to a body with covariance trace
$\Omega(n^2)$.
For sufficiently large $n$, this exceeds $2n$, so the enlarged intersection
cannot have covariance at most $2I_n$.

\begin{theorem}
\label{thm:ellipsoid-growth-counterexample}
There exist absolute constants $n_\ast,\lambda_\ast,C_\ast>0$ such that for all
$n\ge n_\ast$ and
$\frac1n\le\lambda<\lambda_\ast$, if we define
$\alpha:=\exp(\lambda)$, then there exists a symmetric convex body
$K$, a symmetric ellipsoid
$E$ centered at the origin, and a vector
$v\in\R^n$ such that
\[
K\subseteq\alpha E,
\qquad
\norm{v}_2=1,
\qquad
v^\top(A_{K\cap E})v\le C_\ast(n\lambda)^{-2}\cdot v^\top(A_{K\cap\alpha E})v.
\]
\end{theorem}

Choose $\lambda$ to be an absolute constant less than
$\min\{\lambda_\ast,\log2\}$, and take $K,E,v$ from
\Cref{thm:ellipsoid-growth-counterexample}.
For sufficiently large $n$, the theorem gives
\[
v^\top \left(A_{K\cap E}\right)v
\le
O(n^{-2}) \cdot v^\top\left(A_{K\cap\alpha E}\right)v.
\]
The intersection $K\cap E$ has nonempty interior, so $A_{K\cap E}$
is positive definite.
Setting
\[
K':=A_{K\cap E}^{-1/2}(K\cap\alpha E)
\qquad
\text{and}
\qquad
E':=A_{K\cap E}^{-1/2}E
\]
we have $K'\cap E'=A_{K\cap E}^{-1/2}(K\cap E)$, since
$\alpha>1$.
Thus $K'\cap E'$ is isotropic by the definition of
$A_{K\cap E}$ and symmetry.
Moreover, since $K\subseteq\alpha E$ and
$\alpha<2$, we have
\[
K'\subseteq\alpha E'\subseteq2E',
\]
and therefore $K'\cap2E'=K'\cap\alpha E'=K'$.
Consequently,
\[
\tr\left(A_{K'\cap2E'}\right)
=\tr\left(A_{K'\cap\alpha E'}\right)
=\tr\left(A_{K\cap E}^{-1/2}A_{K\cap\alpha E}A_{K\cap E}^{-1/2}\right)
\ge\frac{v^\top(A_{K\cap\alpha E})v}{v^\top(A_{K\cap E})v}
\ge\Omega(n^2),
\]
which is not $\widetilde O(n)$.
For sufficiently large $n$, the trace exceeds $2n$, and hence
$A_{K'\cap2E'}\npreceq2I_n$.
Thus $K'\cap2E'$ fails the covariance guarantee in \Cref{def:rounding}
in these coordinates, even though $K'\cap E'$ is isotropic.
Note that our counterexample $K'$ is a symmetric convex body.
Since $K'\cap E'$ is isotropic, \Cref{lem:support-ball} gives
$B_2^n\subseteq K'\cap E'\subseteq K'\cap2E'$.
The contained-ball condition therefore still holds, while both the
near-isotropic covariance bound and the weaker trace bound in
\Cref{conj:jllv-ellipsoid-growth} fail.

For a real random variable $X$ with density $f_X$, write
\[
 \phi_X(\xi)=\E{X\sim f_X}{e^{i\xi X}}
\]
for its characteristic function.

We write $B_\infty^n,B_2^n$ to denote the hypercube and the Euclidean
$n$-ball, respectively; i.e.
\[
B_\infty^n:=[-1,1]^n
\]
\[
B_2^n:=\{x\in\R^n:\norm{x}_2\le1\}.
\]
In what follows, it will be more convenient to work with the ``positive''
versions, i.e. intersecting with the positive orthant:
\[
B_{\infty,+}^n:=[0,1]^n
\]
\[
B_{2,+}^n:=B_2^n\cap[0,+\infty)^n.
\]

\begin{lemma}\label{lem:cube-ball-volume-ratio}
Let $0<c_1<c_2<(1/3)^{1/2}$ be real numbers.
There exist constants $d_1=d_1(c_2)$\footnote{A notation of the form
$d=d(c)$ means that
$d$ depends on
$c$, i.e.
$d$ is a function of
$c$, and thus
$d$ is not an absolute constant.} and
$d_2=d_2(c_1)$ and
$n_0=n_0(c_1,c_2)$, satisfying
$0<d_1<d_2<1$, such that for all
$n\ge n_0$, we have that
\[
d_1n\log\frac{\rho_2}{\rho_1}
\le \log\frac{\vol(B_{\infty,+}^n\cap\rho_2B_{2,+}^n)}{\vol(B_{\infty,+}^n\cap\rho_1B_{2,+}^n)}
= \log\frac{\vol(B_\infty^n\cap\rho_2B_2^n)}{\vol(B_\infty^n\cap\rho_1B_2^n)}
\le d_2n\log\frac{\rho_2}{\rho_1}
\]
for all $c_1n^{1/2}\le\rho_1\le\rho_2\le c_2n^{1/2}$.
\end{lemma}

Let $U$ be a uniform random variable in the interval
$[0,1]$, and let
$f_U(t)$ denote the probability density function of
$U$; i.e.
\[
f_U(t)=\begin{cases}
1&\text{if }t\in[0,1],\\
0&\text{otherwise.}
\end{cases}
\]
Then we have
\[
f_{U^2}(t)=\begin{cases}
\dfrac1{2t^{1/2}}&\text{if }t\in(0,1],\\
0&\text{otherwise,}
\end{cases}
\]
where $f_{U^2}(t)$ is the probability density function of the random variable
$U^2$ (the square of
$U$).

Define the random variable
\[
V:=U_1^2+U_2^2+\cdots+U_n^2,
\]
i.e. the sum of $n$ independent copies of
$U^2$.
Note that
\begin{equation}\label{eq:cube-ball-slice-volume}
\vol(B_{\infty,+}^n\cap\rho B_{2,+}^n)=\int_0^{\rho^2}f_V(t)\,dt.
\end{equation}
In what follows, we will give bounds on $f_V$.

For $s\in\R$, define the function
\[
g_s(t):=f_{U^2}(t)\cdot\exp(st).
\]
For $p\ge0$, define the moments
\[
G_{s,p}:=\int_\R g_s(t)\cdot t^p\,dt.
\]
Let $X_s$ denote the random variable whose probability density function is given
by
$g_s(t)/G_{s,0}$ (i.e. proportional to
$g_s(t)$).
In particular, we have that $X_0$ is identically distributed with
$U^2$.
Note that $X_s$ is an exponential change of measure, which we will use as in the
proof of Cram\'er's theorem (large deviations).

For convenience, define
\[
a_s
:=G_{s,1}/G_{s,0}
=\E{X_s\sim g_s/G_{s,0}}{X_s}.
\]
Note that $a_s$ is a strictly increasing function of
$s$.

\begin{lemma}\label{lem:density-log-derivative}
There exist absolute constants $C_{\mathrm{deriv}},n_{\mathrm{deriv}}>0$ such
that for all
$s\in\R$ and
$n\ge n_{\mathrm{deriv}}$, we have
\[
\left|s+\left.\frac{d}{dt}\right|_{t=na_s}\log f_V(t)\right|
\le C_{\mathrm{deriv}}\frac{\max(1,|s|)}{n}.
\]
\end{lemma}
The Fourier-analytic proof appears below in
\Cref{subsec:fourier-density-derivative}.

\begin{lemma}\label{lem:integrated-density}
For all $s_0<0$, there exist constants
$n_{\mathrm{int}}=n_{\mathrm{int}}(s_0)>0$ and
$C_{\mathrm{int}}=C_{\mathrm{int}}(s_0)>0$ such that for all
$n\ge n_{\mathrm{int}}$ and
$s\le s_0$, we have
\[
-s-C_{\mathrm{int}}\frac{|s|}{n}
\le \frac{f_V(na_s)}{\int_0^{na_s}f_V(t')\,dt'}
\le -s+C_{\mathrm{int}}\frac{|s|}{n}.
\]
\end{lemma}
The proof is a sequence of inequalities starting from
\Cref{lem:density-log-derivative}; it appears
below in \Cref{subsec:integrated-density-bound}.

\begin{proof}[Proof of \Cref{lem:cube-ball-volume-ratio}]
We rewrite the expression
\begin{align*}
\log\frac{\vol(B_{\infty,+}^n\cap\rho_2B_{2,+}^n)}{\vol(B_{\infty,+}^n\cap\rho_1B_{2,+}^n)}
&=\log\frac{\int_0^{\rho_2^2}f_V(t)\,dt}{\int_0^{\rho_1^2}f_V(t)\,dt} \qquad\text{(from Equation~\eqref{eq:cube-ball-slice-volume})}\\
&=\int_0^{\log(\rho_2^2/\rho_1^2)}
\left(\frac{d}{dk}\left(\log\int_0^{e^k\rho_1^2}f_V(t)\,dt\right)\right)\,dk\\
&=\int_0^{\log(\rho_2^2/\rho_1^2)}
\frac{f_V(e^k\rho_1^2)}{\int_0^{e^k\rho_1^2}f_V(t)\,dt}\cdot e^k\rho_1^2\,dk.
\end{align*}
Now, by assumption, for each $k\in[0,\log(\rho_2^2/\rho_1^2)]$ we have
$e^k\rho_1^2\in[c_1^2n,c_2^2n]$, and also
$0<c_1^2<c_2^2<1/3=a_0$.

Since $a_s$ is a continuous function of
$s$ (from \Cref{lem:tilted-continuity}), and
$\lim_{s\searrow-\infty}a_s=0$ (from \Cref{lem:tilted-limits}), and
$a_s$ is a strictly increasing function of
$s$, we conclude that we can find
$s_2<s_1<0$ satisfying
$a_{s_1}=c_2^2$ and
$a_{s_2}=c_1^2$.

Now assume $n\ge n_{\mathrm{int}}(s_1)$, where
$n_{\mathrm{int}}$ is as in \Cref{lem:integrated-density}.
The lemma implies
\[
|a^{-1}(e^k\rho_1^2/n)|\cdot\left(1-\frac{C_{\mathrm{int}}}{n}\right)
\le \frac{f_V(e^k\rho_1^2)}{\int_0^{e^k\rho_1^2}f_V(t)\,dt}
\le |a^{-1}(e^k\rho_1^2/n)|\cdot\left(1+\frac{C_{\mathrm{int}}}{n}\right)
\]
where $C_{\mathrm{int}}=C_{\mathrm{int}}(s_1)$ is as in the lemma, and
$a^{-1}(e^k\rho_1^2/n)$ denotes the unique real number
$s$ satisfying
$a_s=e^k\rho_1^2/n$.
Multiplying by $e^k\rho_1^2$, we have
\[
|a^{-1}(e^k\rho_1^2/n)|\cdot e^k\rho_1^2\cdot\left(1-\frac{C_{\mathrm{int}}}{n}\right)
\le \frac{f_V(e^k\rho_1^2)}{\int_0^{e^k\rho_1^2}f_V(t)\,dt}\cdot e^k\rho_1^2
\le |a^{-1}(e^k\rho_1^2/n)|\cdot e^k\rho_1^2\cdot\left(1+\frac{C_{\mathrm{int}}}{n}\right).
\]
Note that $a^{-1}(e^k\rho_1^2/n)\in[s_2,s_1]$.
If
\[
s=a^{-1}(e^k\rho_1^2/n),
\]
then $e^k\rho_1^2=na_s$.
Moreover, \Cref{lem:tilted-mean-monotone} implies that the map $s\mapsto
|s|a_s=-sa_s$ is
decreasing on
$(-\infty,0)$.
Hence
\[
n|s_1|a_{s_1}
\le |a^{-1}(e^k\rho_1^2/n)|\cdot e^k\rho_1^2
\le n|s_2|a_{s_2}
<\frac n2.
\]
It follows that
\[
n|s_1|a_{s_1}\left(1-\frac{C_{\mathrm{int}}}{n}\right)
\le \frac{f_V(e^k\rho_1^2)}{\int_0^{e^k\rho_1^2}f_V(t)\,dt}\cdot e^k\rho_1^2
\le n|s_2|a_{s_2}\left(1+\frac{C_{\mathrm{int}}}{n}\right).
\]
Integrating over $k\in[0,\log(\rho_2^2/\rho_1^2)]$ gives
\[
2n|s_1|a_{s_1}\left(1-\frac{C_{\mathrm{int}}}{n}\right)\log\frac{\rho_2}{\rho_1}
\le \log\frac{\vol(B_{\infty,+}^n\cap\rho_2B_{2,+}^n)}{\vol(B_{\infty,+}^n\cap\rho_1B_{2,+}^n)}
\le 2n|s_2|a_{s_2}\left(1+\frac{C_{\mathrm{int}}}{n}\right)\log\frac{\rho_2}{\rho_1}.
\]
Now choose $d_2$ arbitrarily in the interval
$(2|s_2|a_{s_2},1)$, which is nonempty because
$|s_2|a_{s_2}<1/2$.
If
\[
n\ge \frac{C_{\mathrm{int}}}{d_2/(2|s_2|a_{s_2})-1},
\]
then the upper bound is at most
\[
d_2n\log\frac{\rho_2}{\rho_1}.
\]
If in addition $n\ge2C_{\mathrm{int}}$, then the lower bound is at least
\[
|s_1|a_{s_1}n\log\frac{\rho_2}{\rho_1}.
\]
Thus we may take $d_1=|s_1|a_{s_1}$, after increasing
$n_0$ if necessary, and the proof is complete.
\end{proof}

The next lemma concerns the one-dimensional slice-profile bounds.

\begin{lemma}\label{lem:slice-profile}
Let $h:\R\to[0,\infty)$ be an even, integrable, logconcave function with
$0<h(0)<\infty$, and let
$T$ have density
$h/\int_{\R}h$.
\begin{enumerate}[label=\textup{(\roman*)}]
\item \label{itm:slice-profile:1} If $h(\delta)\le e^{-\kappa}h(0)$ for some
  $\delta,\kappa>0$, then
\[
\V{T\sim h/\int_{\R}h}{T}\le \frac{6}{\min\{\kappa,1\}^2}\,\delta^2.
\]
\item \label{itm:slice-profile:2} If $h$ is supported on
  $[-1,1]$ and
  $h(1)\ge h(0)/C$ for some
  $C\ge1$, then
\[
\V{T\sim h/\int_{\R}h}{T}\ge \frac1{3C}.
\]
\end{enumerate}
\end{lemma}
\begin{proof}
Since $h$ is even and logconcave, it is nonincreasing on
$[0,\infty)$ and
$\E{T\sim h/\int_{\R}h}{T}=0$.

For part~\textup{(i)}, set $c:=\min\{\kappa,1\}$ and
$b:=\delta/c$.
Writing
\[
\psi(t):=\log h(0)-\log h(t)
\]
on the interior of the support, logconcavity of $h$ says that
$\psi$ is convex and
$\psi(0)=0$.
Since $b\ge\delta$,
\[
\psi(b)
\ge \frac{b}{\delta}\psi(\delta)
\ge\frac{\kappa}{c}
\ge1.
\]
Consequently, there is some $a\in(0,b]$ such that
$h(a)=e^{-1}h(0)$.

Convexity of $\psi$ then gives
$\psi(t)\ge t/a$ for
$t\ge a$, and hence
\[
h(t)\le h(0)e^{-t/a}\qquad(t\ge a).
\]
On the other hand, monotonicity gives $h(t)\ge h(a)=e^{-1}h(0)$ for
$0\le t\le a$.
Therefore
\[
\int_{\R}h(t)\,dt\ge \frac{2a}{e}h(0),
\]
while
\begin{align*}
\int_{\R}t^2h(t)\,dt
&\le 2h(0)\left(\int_0^a t^2\,dt+\int_a^\infty t^2e^{-t/a}\,dt\right)\\
&=2h(0)a^3\left(\frac13+\frac5e\right).
\end{align*}
Thus
\[
\V{T\sim h/\int_{\R}h}{T}
=\E{T\sim h/\int_{\R}h}{T^2}
\le\left(\frac e3+5\right)a^2
<6a^2
\le \frac{6}{c^2}\delta^2.
\]

For part~\textup{(ii)}, monotonicity gives $h(t)\ge h(1)\ge h(0)/C$ on
$[0,1]$, whereas
$h(t)\le h(0)$ everywhere on its support.
Hence
\[
\int_{-1}^1h(t)\,dt\le2h(0)
\]
and
\[
\int_{-1}^1t^2h(t)\,dt
\ge2\frac{h(0)}C\int_0^1t^2\,dt
=\frac{2h(0)}{3C}.
\]
Dividing proves $\V{T\sim h/\int_{\R}h}{T}\ge1/(3C)$.
\end{proof}

\begin{proof}[Proof of \Cref{thm:ellipsoid-growth-counterexample}]
Let $0<\sigma_1<\sigma_2<\rho_1<\rho_2$ be real numbers satisfying
$\rho_2/\rho_1=\alpha$, where the precise values of
$\sigma_1,\sigma_2,\rho_1$ will be specified later.
Assume $n\ge3$.
First, define the function $g:[0,\sigma_2]\to[0,\infty)$ by
\[
g(\sigma):=
\begin{cases}
(\rho_2^2-\sigma^2)^{1/2}&\text{if }0\le\sigma\le\sigma_1,\\[1ex]
\dfrac{\sigma_2-\sigma}{\sigma_2-\sigma_1}(\rho_2^2-\sigma_1^2)^{1/2}
+\dfrac{\sigma-\sigma_1}{\sigma_2-\sigma_1}(\rho_1^2-\sigma_2^2)^{1/2}&\text{if }\sigma_1\le\sigma\le\sigma_2.
\end{cases}
\]
In other words, $g$ is a circular arc on the interval
$[0,\sigma_1]$, and a line segment on the interval
$[\sigma_1,\sigma_2]$ (see \Cref{fig:ellipsoid-growth-profile}).
We will impose condition~\eqref{eq:ellipsoid-growth-concavity} below.
That condition also guarantees that the even extension
$\sigma\mapsto g(|\sigma|)$ is concave; the short verification is given
immediately after Equation~\eqref{eq:ellipsoid-growth-concavity}.

\begin{figure}[H]
    \centering
    \includegraphics[width=0.35\textwidth]{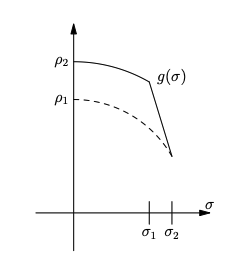}
    \caption{Schematic graph of $g$: a circular arc followed by a line segment.}
    \label{fig:ellipsoid-growth-profile}
\end{figure}

Set $v=e_n\in\R^n$, and introduce the intermediate cylinder
\[
\widehat E
:=\rho_1B_2^{n-1}\times\R
=\left\{x\in\R^n:(x_1^2+x_2^2+\cdots+x_{n-1}^2)^{1/2}\le\rho_1\right\},
\]
where we identify $\R^{n-1}$ with the hyperplane
$v^\perp:=\{x\in\R^n:\langle x,v\rangle=0\}$.
Since $\rho_2/\rho_1=\alpha$,
\[
\alpha\widehat E=\rho_2B_2^{n-1}\times\R.
\]
The cylindrical calculation will be transferred to a genuine ellipsoid at the
end of the proof.

Finally, define
\[
K:=\bigcap_{i=1}^{n-1}\left\{x\in\R^n:
|x_n|\le1\text{ and }|x_i|\le\sigma_2-(\sigma_2-\sigma_1)|x_n|
\text{ and }(\norm{x}_2^2-x_n^2-x_i^2)^{1/2}\le g(|x_i|)
\right\}.
\]
The set $K$ is symmetric and unconditional.
Under condition~\eqref{eq:ellipsoid-growth-concavity}, the function $g(|\cdot|)$
is concave, and therefore every set appearing in the intersection above is
convex; hence
$K$ is a convex body.
For each $\eta\in[-1,1]$, the slice
$K\cap\{x\in\R^n:x_n=\eta\}$ is an
$(n-1)$-dimensional unconditional convex body.

Now note that
\[
K\subseteq\alpha\widehat E
\]
(because $g(\sigma)\le(\rho_2^2-\sigma^2)^{1/2}$ for
$\sigma\in[\sigma_1,\sigma_2]$).
Thus we have
\[
K\cap\alpha\widehat E\cap\{x\in\R^n:x_n=0\}
=
\left(\bigcap_{i=1}^{n-1}\left\{x\in\R^{n-1}:|x_i|\le\sigma_2\text{ and }(\norm{x}_2^2-x_i^2)^{1/2}\le g(|x_i|)\right\}\right)\times\{0\}
\]
\[
K\cap\alpha\widehat E\cap\{x\in\R^n:x_n=1\}
=(\sigma_1B_\infty^{n-1}\cap\rho_2B_2^{n-1})\times\{1\}.
\]
For convenience, define the convex bodies $K_0,K_1\subset\R^{n-1}$ by
\[
K_0:=\bigcap_{i=1}^{n-1}\left\{x\in\R^{n-1}:|x_i|\le\sigma_2\text{ and }(\norm{x}_2^2-x_i^2)^{1/2}\le g(|x_i|)\right\}
\]
\[
K_1:=\sigma_1B_\infty^{n-1}\cap\rho_2B_2^{n-1}
\]
so that we have
\[
\vol_{n-1}(K\cap\alpha\widehat E\cap\{x\in\R^n:x_n=0\})=\vol_{n-1}(K_0)
\]
\begin{equation}\label{eq:ellipsoid-growth-outer-slice}
\vol_{n-1}(K\cap\alpha\widehat E\cap\{x\in\R^n:x_n=1\})=\vol_{n-1}(K_1).
\end{equation}

Now, let us assume that $\sigma_1,\sigma_2,\rho_1$ satisfy
\begin{equation}\label{eq:ellipsoid-growth-concavity}
\frac{\sigma_2}{(\rho_1^2-\sigma_2^2)^{1/2}}
<
\frac{(\rho_2^2-\sigma_1^2)^{1/2}-(\rho_1^2-\sigma_2^2)^{1/2}}{\sigma_2-\sigma_1}.
\end{equation}
Let
\[
m:=\frac{(\rho_1^2-\sigma_2^2)^{1/2}-(\rho_2^2-\sigma_1^2)^{1/2}}{\sigma_2-\sigma_1}
\]
be the slope of the linear segment.
Condition~\eqref{eq:ellipsoid-growth-concavity}, together with
\[
\frac{\sigma_2}{(\rho_1^2-\sigma_2^2)^{1/2}}
>
\frac{\sigma_1}{(\rho_2^2-\sigma_1^2)^{1/2}},
\]
gives
\[
m
< -\frac{\sigma_1}{(\rho_2^2-\sigma_1^2)^{1/2}}
=\left.\frac{d}{d\sigma}(\rho_2^2-\sigma^2)^{1/2}\right|_{\sigma=\sigma_1}.
\]
Thus the derivative does not jump upward at $\sigma_1$.
Since the circular arc is concave and the second piece is affine, $g$ is concave
and nonincreasing on
$[0,\sigma_2]$; because
$g'(0)=0$, its even extension
$g(|\sigma|)$ is concave on
$[-\sigma_2,\sigma_2]$.
This justifies the convexity assertion for $K$ above.

Condition~\eqref{eq:ellipsoid-growth-concavity} also gives
\[
g(\sigma)\ge(\rho_1^2-\sigma^2)^{1/2}
\]
for all $\sigma\in[0,\sigma_2]$: on
$[0,\sigma_1]$ this is immediate from
$\rho_2>\rho_1$, while on
$[\sigma_1,\sigma_2]$ it follows by comparing derivatives and using equality at
$\sigma_2$.
This implies that for each $\eta\in[-1,1]$, we have
\begin{equation}\label{eq:ellipsoid-growth-inner-slice}
K\cap\widehat E\cap\{x\in\R^n:x_n=\eta\}
=
\left((\sigma_2-(\sigma_2-\sigma_1)|\eta|)B_\infty^{n-1}\cap\rho_1B_2^{n-1}\right)\times\{\eta\}.
\end{equation}

Fix, for definiteness,
\[
c_1:=\frac1{100},
\qquad
c_2:=\frac14;
\]
then $0<c_1<c_2<(1/3)^{1/2}$ and
$0.1(c_2/c_1)=2.5$.
Choose $d_1,d_2,n_0$ as in \Cref{lem:cube-ball-volume-ratio} for these fixed
constants.
We choose $\lambda_\ast>0$ small enough that
\[
\lambda_\ast<\min\left\{\log2,\frac3{d_1},0.8\right\}.
\]
Henceforth assume $n-2\ge n_0$ and
$1/n\le\lambda<\lambda_\ast$.
We set
\[
\sigma_1=1,
\qquad
\sigma_2=1+\frac13d_1\lambda,
\qquad
\rho_1=\frac1\alpha c_2(n-2)^{1/2},
\qquad
\rho_2=c_2(n-2)^{1/2}.
\]
Our choice gives $\alpha=e^\lambda<2<0.1(c_2/c_1)$ and
$\sigma_2<2$.
In particular,
\[
\frac{\rho_1}{\sigma_2}
> \frac{c_2}{2\alpha}(n-2)^{1/2}
> 5c_1(n-2)^{1/2}.
\]
Also, since $\sigma_2=O(1)$ whereas
$\rho_1=\Theta(n^{1/2})$, we have
$\sigma_2<\rho_1$ once
$n$ is larger than an absolute constant.

Let us now check that condition~\eqref{eq:ellipsoid-growth-concavity} holds.
We have that
\[
(\rho_2^2-\sigma_1^2)^{1/2}-(\rho_1^2-\sigma_2^2)^{1/2}=\Theta(\lambda\cdot(n-2)^{1/2})
\]
(differentiate
\[
\left(c_2^2(n-2)e^{-2\lambda}-\left(1+\frac13d_1\lambda\right)^2\right)^{1/2}
\]
wrt $\lambda$)
\[
\sigma_2-\sigma_1=\Theta(\lambda)
\]
and hence we conclude that the right-hand side of
Equation~\eqref{eq:ellipsoid-growth-concavity} is
\[
\frac{(\rho_2^2-\sigma_1^2)^{1/2}-(\rho_1^2-\sigma_2^2)^{1/2}}{\sigma_2-\sigma_1}
=\Theta((n-2)^{1/2}).
\]
Since the LHS is $\Theta(1/(n-2)^{1/2})$, we see that
Equation~\eqref{eq:ellipsoid-growth-concavity} holds for sufficiently
large
$n$, as desired.

Thus we may proceed using Equation~\eqref{eq:ellipsoid-growth-inner-slice}.
We have
\[
\vol_{n-1}(K\cap\widehat E\cap\{x\in\R^n:x_n=\eta\})
=\vol_{n-1}((\sigma_2-(\sigma_2-\sigma_1)|\eta|)B_\infty^{n-1}\cap\rho_1B_2^{n-1}).
\]
Let us compare the volume of the slices at $\eta=0$ and at
$\eta=\frac1{n\lambda}$ (this latter slice is guaranteed to intersect
$K$ since
$\frac1{n\lambda}\le1$ by assumption).
We write
\begin{align*}
&\log\frac{\vol_{n-1}(K\cap\widehat E\cap\{x\in\R^n:x_n=0\})}{\vol_{n-1}(K\cap\widehat E\cap\{x\in\R^n:x_n=\frac1{n\lambda}\})}\\
&=\log\frac{\vol_{n-1}(\sigma_2B_\infty^{n-1}\cap\rho_1B_2^{n-1})}{\vol_{n-1}((\sigma_2-\frac13d_1\cdot\frac1n)B_\infty^{n-1}\cap\rho_1B_2^{n-1})}\\
&=(n-1)\log\frac{\sigma_2}{\sigma_2-\frac13d_1\cdot\frac1n}
+\log\frac{\vol_{n-1}(B_\infty^{n-1}\cap\frac{\rho_1}{\sigma_2}B_2^{n-1})}{\vol_{n-1}(B_\infty^{n-1}\cap\frac{\rho_1}{\sigma_2-\frac13d_1\cdot\frac1n}B_2^{n-1})}.
\end{align*}

By the choices above, $\sigma_2<2$ and
$\rho_1/\sigma_2>5c_1(n-2)^{1/2}$, so
$\rho_1/\sigma_2>c_1(n-1)^{1/2}$ for sufficiently large
$n$.
Also,
\[
\frac{\rho_1}{\sigma_2-\frac13d_1\cdot\frac1n}
\le\frac{\rho_1}{\sigma_1}
\le\frac{\rho_2}{\sigma_1}
=c_2(n-2)^{1/2}
<c_2(n-1)^{1/2}.
\]
We conclude that the hypotheses of \Cref{lem:cube-ball-volume-ratio} hold, and
thus we may apply
the lemma, yielding
\[
\log\frac{\vol_{n-1}(B_\infty^{n-1}\cap\frac{\rho_1}{\sigma_2}B_2^{n-1})}{\vol_{n-1}(B_\infty^{n-1}\cap\frac{\rho_1}{\sigma_2-\frac13d_1\cdot\frac1n}B_2^{n-1})}
\ge d_2(n-1)\log\frac{\sigma_2-\frac13d_1\cdot\frac1n}{\sigma_2}.
\]
Thus
\[
\log\frac{\vol_{n-1}(K\cap\widehat E\cap\{x\in\R^n:x_n=0\})}{\vol_{n-1}(K\cap\widehat E\cap\{x\in\R^n:x_n=\frac1{n\lambda}\})}
\ge(1-d_2)(n-1)\log\frac{\sigma_2}{\sigma_2-\frac13d_1\cdot\frac1n}.
\]
Since
\[
\frac{\sigma_2-\frac13d_1\cdot\frac1n}{\sigma_2}
=1-\frac1{\sigma_2}\cdot\frac13d_1\cdot\frac1n
<1-\frac12\cdot\frac13d_1\cdot\frac1n
\]
we have
\[
\log\frac{\sigma_2-\frac13d_1\cdot\frac1n}{\sigma_2}< -\frac16d_1\cdot\frac1n
\]
and hence
\[
\log\frac{\vol_{n-1}(K\cap\widehat E\cap\{x\in\R^n:x_n=0\})}{\vol_{n-1}(K\cap\widehat E\cap\{x\in\R^n:x_n=\frac1{n\lambda}\})}
\ge(1-d_2)(n-1)\cdot\frac16d_1\cdot\frac1n.
\]
For $n\ge2$, the right-hand side is at least the absolute constant
\[
\kappa_0
:=\frac{(1-d_2)d_1}{12}
>0.
\]
Define the slice-volume profile
\[
h_{\mathrm{small}}(\eta):=\vol_{n-1}\!\left(K\cap\widehat E\cap\{x\in\R^n:x_n=\eta\}\right).
\]
By the Pr\'ekopa--Leindler theorem, $h_{\mathrm{small}}$ is logconcave, and it
is even because
$K\cap\widehat E$ is symmetric.
The preceding estimate says
\[
h_{\mathrm{small}}\!\left(\frac1{n\lambda}\right)\le e^{-\kappa_0}h_{\mathrm{small}}(0).
\]
Applying \Cref{lem:slice-profile}\textup{(i)} with $\delta=1/(n\lambda)$
gives
\begin{equation}\label{eq:ellipsoid-growth-inner-variance}
v^\top(A_{K\cap\widehat E})v
=\V{x\sim\Unif(K\cap\widehat E)}{x_n}
 \le C_0\frac1{(n\lambda)^2},
\end{equation}
where $C_0:=6/\min\{\kappa_0,1\}^2$ is an absolute constant.

Let us return to Equation~\eqref{eq:ellipsoid-growth-outer-slice}.
We would like to show that the ratio
\[
\frac{\vol_{n-1}(K\cap\alpha\widehat E\cap\{x\in\R^n:x_n=0\})}{\vol_{n-1}(K\cap\alpha\widehat E\cap\{x\in\R^n:x_n=1\})}
\]
is not too large; i.e. that the volume of $K_0$ is not too much larger than the
volume of
$K_1$.
First, note that
\[
\begin{aligned}
K_0\subseteq K_1\cup\bigcup_{i=1}^{n-1}\bigl\{x\in\R^{n-1}:{}&
|x_i|\in[\sigma_1,\sigma_2],\quad
(\norm{x}_2^2-x_i^2)^{1/2}\le g(|x_i|),\\
& |x_j|\le|x_i|\text{ for each }j=1,2,\ldots,n-1\bigr\}.
\end{aligned}
\]
Thus
\[
\vol_{n-1}(K_0)
\le\vol_{n-1}(K_1)+2(n-1)\int_{\sigma_1}^{\sigma_2}\vol_{n-2}(\sigma B_\infty^{n-2}\cap g(\sigma)B_2^{n-2})\,d\sigma
\]
and since
\[
(\sigma_1B_\infty^{n-2}\cap(\rho_2^2-\sigma_1^2)^{1/2}B_2^{n-2})\times[-\sigma_1,\sigma_1]\subseteq K_1,
\]
we have
\begin{align*}
\frac{\vol_{n-1}(K_0)}{\vol_{n-1}(K_1)}
&\le1+\frac{n-1}{\sigma_1}\int_{\sigma_1}^{\sigma_2}
\frac{\vol_{n-2}(\sigma B_\infty^{n-2}\cap g(\sigma)B_2^{n-2})}{\vol_{n-2}(\sigma_1B_\infty^{n-2}\cap g(\sigma_1)B_2^{n-2})}\,d\sigma\\
&=1+\frac{n-1}{\sigma_1}\int_{\sigma_1}^{\sigma_2}
\left(\frac\sigma{\sigma_1}\right)^{n-2}
\frac{\vol_{n-2}(B_\infty^{n-2}\cap\frac{g(\sigma)}\sigma B_2^{n-2})}{\vol_{n-2}(B_\infty^{n-2}\cap\frac{g(\sigma_1)}{\sigma_1}B_2^{n-2})}\,d\sigma.
\end{align*}
Now
\[
\frac{g(\sigma_1)}{\sigma_1}
=\frac{(\rho_2^2-\sigma_1^2)^{1/2}}{\sigma_1}
<\frac{\rho_2}{\sigma_1}
=c_2(n-2)^{1/2},
\]
whereas
\[
\frac{g(\sigma)}\sigma
\ge\frac{g(\sigma_2)}{\sigma_2}
=\frac{(\rho_1^2-\sigma_2^2)^{1/2}}{\sigma_2}
>((5c_1(n-2)^{1/2})^2-1)^{1/2}
\]
by the assumption on $\alpha$, so
$g(\sigma)/\sigma>c_1(n-2)^{1/2}$ for sufficiently large
$n$.
We conclude that the hypotheses of \Cref{lem:cube-ball-volume-ratio} hold, and
thus we may apply
this lemma, yielding
\[
\log\frac{\vol_{n-2}(B_\infty^{n-2}\cap\frac{g(\sigma)}\sigma B_2^{n-2})}{\vol_{n-2}(B_\infty^{n-2}\cap\frac{g(\sigma_1)}{\sigma_1}B_2^{n-2})}
\le d_1(n-2)\log\left(\frac{g(\sigma)}{g(\sigma_1)}\cdot\frac{\sigma_1}\sigma\right).
\]
Thus
\begin{equation}\label{eq:ellipsoid-growth-volume-ratio}
\frac{\vol_{n-1}(K_0)}{\vol_{n-1}(K_1)}
\le1+\frac{n-1}{\sigma_1}\int_{\sigma_1}^{\sigma_2}
\exp\left(d_1(n-2)\log\frac{g(\sigma)}{g(\sigma_1)}+(1-d_1)(n-2)\log\frac\sigma{\sigma_1}\right)\,d\sigma.
\end{equation}
Let us define the auxiliary variable
\[
\gamma:=
\frac{(\rho_2^2-\sigma_1^2)^{1/2}-(\rho_1^2-\sigma_2^2)^{1/2}}{(\rho_2^2-\sigma_1^2)^{1/2}}
\cdot\frac{\sigma_1}{\sigma_2-\sigma_1}.
\]
Note that if we write
\[
\frac\sigma{\sigma_1}=1+\theta
\]
(thus $\theta\in[0,(\sigma_2-\sigma_1)/\sigma_1]$), we will have
\[
\frac{g(\sigma)}{g(\sigma_1)}=1-\gamma\theta.
\]
Thus
\begin{align*}
&d_1(n-2)\log\frac{g(\sigma)}{g(\sigma_1)}+(1-d_1)(n-2)\log\frac\sigma{\sigma_1}\\
&=d_1(n-2)\log(1-\gamma\theta)+(1-d_1)(n-2)\log(1+\theta)\\
&\le d_1(n-2)(-\gamma\theta)+(1-d_1)(n-2)\theta\\
&=(n-2)\cdot\theta\cdot(1-d_1-d_1\gamma).
\end{align*}
If we can show $1-d_1-d_1\gamma\le-1$, we would be able to conclude from
Equation~\eqref{eq:ellipsoid-growth-volume-ratio} that
\begin{align*}
\frac{\vol_{n-1}(K_0)}{\vol_{n-1}(K_1)}
&\le1+\frac{n-1}{\sigma_1}\int_{\sigma_1}^{\sigma_2}\exp(-(n-2)\cdot\theta)\,d\sigma\\
&=1+(n-1)\int_0^{(\sigma_2-\sigma_1)/\sigma_1}\exp(-(n-2)\cdot\theta)\,d\theta\\
&\le1+(n-1)\int_0^{+\infty}\exp(-(n-2)\cdot\theta)\,d\theta\\
&=1+\frac{n-1}{n-2}\\
&\le 3
\end{align*}
for $n\ge3$.
Define
\[
h_{\mathrm{large}}(\eta):=\vol_{n-1}\!\left(K\cap\alpha\widehat E\cap\{x\in\R^n:x_n=\eta\}\right).
\]
The function $h_{\mathrm{large}}$ is even and logconcave by symmetry and the
Pr\'ekopa--Leindler theorem, and it is supported on
$[-1,1]$.
The estimate above and Equation~\eqref{eq:ellipsoid-growth-outer-slice} give
\[
h_{\mathrm{large}}(1)\ge \frac13 h_{\mathrm{large}}(0).
\]
Thus \Cref{lem:slice-profile}\textup{(ii)} yields
\begin{equation}\label{eq:ellipsoid-growth-outer-variance}
v^\top(A_{K\cap\alpha\widehat E})v
=\V{x\sim\Unif(K\cap\alpha\widehat E)}{x_n}
 \ge \frac19.
\end{equation}
Combining Equations~\eqref{eq:ellipsoid-growth-inner-variance}
and~\eqref{eq:ellipsoid-growth-outer-variance} gives the required comparison for
the intermediate
cylinder $\widehat E$.

It suffices to show that $1-d_1-d_1\gamma\le-1$.
We know that
\[
\frac{(\rho_1^2-\sigma_2^2)^{1/2}}{(\rho_2^2-\sigma_1^2)^{1/2}}
<\frac{\rho_1}{\rho_2}
=e^{-\lambda}.
\]
Hence
\[
\frac{(\rho_2^2-\sigma_1^2)^{1/2}-(\rho_1^2-\sigma_2^2)^{1/2}}{(\rho_2^2-\sigma_1^2)^{1/2}}
>1-e^{-\lambda}
>\frac23\lambda
\]
assuming $\lambda<0.8$.
Thus
\[
\gamma
= \frac{(\rho_2^2-\sigma_1^2)^{1/2}-(\rho_1^2-\sigma_2^2)^{1/2}}{(\rho_2^2-\sigma_1^2)^{1/2}} \cdot\frac{\sigma_1}{\sigma_2-\sigma_1}
>\frac23\lambda\cdot\frac1{\frac13d_1\lambda}
=\frac2{d_1}
>\frac{2-d_1}{d_1}
\]
which implies $1-d_1-d_1\gamma\le-1$, as desired.

It remains to replace the cylinder by a genuine ellipsoid.
For $L>1/\alpha$, define
\[
\beta_L:=\left(1-\frac1{\alpha^2L^2}\right)^{-1/2}
\]
and
\[
E_L:=\left\{(x',x_n)\in\R^{n-1}\times\R:
\frac{\norm{x'}_2^2}{\beta_L^2\rho_1^2}+\frac{x_n^2}{L^2}\le1\right\}.
\]
Then $E_L$ is a symmetric ellipsoid centered at the origin.
Moreover, $K\subseteq\alpha\widehat E$ implies
$\norm{x'}_2\le\rho_2$ for every
$(x',x_n)\in K$, while the definition of
$K$ gives
$|x_n|\le1$.
Hence, for every $x\in K$,
\begin{align*}
\frac{\norm{x'}_2^2}{\alpha^2\beta_L^2\rho_1^2}+\frac{x_n^2}{\alpha^2L^2}
&\le \frac{\rho_2^2}{\beta_L^2\rho_2^2}+\frac1{\alpha^2L^2}\\
&=\beta_L^{-2}+\frac1{\alpha^2L^2}=1.
\end{align*}
Thus
\[
K\subseteq\alpha E_L.
\]

As $L\to\infty$, we have
$\beta_L\to1$, and the indicators of
$K\cap E_L$ and
$K\cap\alpha E_L$ converge almost everywhere to the indicators of
$K\cap\widehat E$ and
$K\cap\alpha\widehat E$, respectively.
Since $K$ is compact, dominated convergence applies both to the volumes and to
the integrals of
$x_n^2$.
All of these sets are symmetric, so their means are at the origin.
Consequently,
\[
v^\top A_{K\cap E_L}v\longrightarrow v^\top A_{K\cap\widehat E}v,
\qquad
v^\top A_{K\cap\alpha E_L}v\longrightarrow v^\top A_{K\cap\alpha\widehat E}v.
\]
Choose $L$ sufficiently large that, using
Equations~\eqref{eq:ellipsoid-growth-inner-variance}
and~\eqref{eq:ellipsoid-growth-outer-variance},
\[
v^\top A_{K\cap E_L}v\le \frac{2C_0}{(n\lambda)^2},
\qquad
v^\top A_{K\cap\alpha E_L}v\ge\frac1{18}.
\]
Then
\[
v^\top A_{K\cap E_L}v
\le 36C_0(n\lambda)^{-2}\,v^\top A_{K\cap\alpha E_L}v.
\]
Taking $E:=E_L$,
$C_\ast:=36C_0$, and
$n_\ast:=n_0+2$ proves the theorem, including the additional property
$K\subseteq\alpha E$.
\end{proof}

\subsection{Fourier inversion theorem for the density derivative}
\label{subsec:fourier-density-derivative}

\begin{proof}[Proof of \Cref{lem:density-log-derivative}]
Define the random variable
\[
Y_s:=(X_s)_1+(X_s)_2+\cdots+(X_s)_n,
\]
i.e. the sum of $n$ independent copies of
$X_s$.
Since $X_s$ is an exponential change of measure from
$U^2$, we have
\[
f_{Y_s}(t)=f_V(t)\cdot\exp(st)/G_{s,0}^n
\]
and hence
\[
\frac{d}{dt}\log f_{Y_s}(t)=s+\frac{d}{dt}\log f_V(t)
\]
for all $t$.
Thus, it suffices to show
\[
\left|\left.\frac{d}{dt}\right|_{t=na_s}\log f_{Y_s}(t)\right|
\le C_{\mathrm{deriv}}\frac{\max(1,|s|)}{n}
\]
for sufficiently large $n$, or equivalently,
\begin{equation}\label{eq:fourier-density-derivative-target}
\left|\frac{f'_{Y_s}(na_s)}{f_{Y_s}(na_s)}\right|
\le C_{\mathrm{deriv}}\frac{\max(1,|s|)}{n}.
\end{equation}
Let us define
\[
v_s
:=G_{s,2}/G_{s,0}-a_s^2
=\E{X_s\sim g_s/G_{s,0}}{X_s^2}-(\E{X_s\sim g_s/G_{s,0}}{X_s})^2
=\V{X_s\sim g_s/G_{s,0}}{X_s}.
\]
Define the normalized random variable
\[
\widetilde X_s:=\frac{X_s-a_s}{v_s^{1/2}},
\]
which has mean $0$ and variance
$1$.
Similarly define
\[
\widetilde Y_s
:=(\widetilde X_s)_1+(\widetilde X_s)_2+\cdots+(\widetilde X_s)_n
=\frac{Y_s-na_s}{v_s^{1/2}},
\]
which is the normalized form of $Y_s$.
Now observe
\[
\frac{f'_{Y_s}(na_s)}{f_{Y_s}(na_s)}
=\frac1{v_s^{1/2}}\cdot\frac{f'_{\widetilde Y_s}(0)}{f_{\widetilde Y_s}(0)},
\]
so to prove Equation~\eqref{eq:fourier-density-derivative-target}, it suffices
to give a lower bound on $v_s$ and an upper bound on
\[
\left|\frac{f'_{\widetilde Y_s}(0)}{f_{\widetilde Y_s}(0)}\right|.
\]
Observe that $f_{\widetilde Y_s}$ can be viewed as the convolution of
$n$ copies of
$f_{\widetilde X_s}$.
Now, we will use Fourier analysis, to transform the convolution into a pointwise
product.

Recall the characteristic function
\begin{equation}\label{eq:normalized-characteristic-function}
\phi_{\widetilde X_s}(\xi):=\E{X_s\sim g_s/G_{s,0}}{\exp(i\xi\widetilde X_s)}
=\int_\R f_{\widetilde X_s}(t)\exp(i\xi t)\,dt,
\end{equation}
so that we have
\begin{equation}\label{eq:normalized-sum-characteristic-function}
\phi_{\widetilde Y_s}(\xi)=\phi_{\widetilde X_s}(\xi)^n.
\end{equation}
We will use Fourier inversion in the following differentiated form.
Once
\[
\phi_{\widetilde Y_s}\in L^1(\R)
\qquad\text{and}\qquad
\xi\phi_{\widetilde Y_s}(\xi)\in L^1(\R)
\]
have been verified, Fourier inversion and differentiation under the integral
give
\[
f_{\widetilde Y_s}(t)=\frac1{2\pi}\int_\R\phi_{\widetilde Y_s}(\xi)e^{-i\xi t}\,d\xi
\]
and
\begin{equation}\label{eq:fourier-density-derivative-ratio}
\frac{f'_{\widetilde Y_s}(0)}{f_{\widetilde Y_s}(0)}
=
\frac{\int_\R\phi_{\widetilde Y_s}(\xi)\cdot(-i\xi)\,d\xi}
{\int_\R\phi_{\widetilde Y_s}(\xi)\,d\xi}.
\end{equation}
The required integrability is established below, before
Equation~\eqref{eq:fourier-density-derivative-ratio} is used.

We will show that for $\xi$ close to
$0$,
$\phi_{\widetilde Y_s}(\xi)$ is close to a Gaussian function.

From Equation~\eqref{eq:normalized-characteristic-function}, we have
\begin{align*}
\phi_{\widetilde X_s}(\xi)
&=\int_\R f_{\widetilde X_s}(t)(\cos(\xi t)+i\sin(\xi t))\,dt\\
&=\int_\R f_{\widetilde X_s}(t)\left(1-\frac12(\xi t)^2
+\left(\cos(\xi t)-\left(1-\frac12(\xi t)^2\right)\right)
+i\xi t+i(\sin(\xi t)-\xi t)\right)\,dt\\
&=1-\frac12\xi^2
+\int_\R f_{\widetilde X_s}(t)\left(\cos(\xi t)-\left(1-\frac12(\xi t)^2\right)\right)\,dt
+i\int_\R f_{\widetilde X_s}(t)(\sin(\xi t)-\xi t)\,dt,
\end{align*}
using the fact that $\E{X_s\sim g_s/G_{s,0}}{\widetilde X_s}=0$ and
$\E{X_s\sim g_s/G_{s,0}}{\widetilde X_s^2}=1$.
Now, setting
\[
E_{\mathrm{re}}(\xi):=\int_\R f_{\widetilde X_s}(t)\left(\cos(\xi t)-\left(1-\frac12(\xi t)^2\right)\right)\,dt
\]
\[
E_{\mathrm{im}}(\xi):=\int_\R f_{\widetilde X_s}(t)(\sin(\xi t)-\xi t)\,dt
\]
we have
\begin{equation}\label{eq:characteristic-function-expansion}
\phi_{\widetilde X_s}(\xi)=1-\frac12\xi^2+E_{\mathrm{re}}(\xi)+iE_{\mathrm{im}}(\xi).
\end{equation}
Now, recall the inequalities
\[
0
\le\cos u-\left(1-\frac12u^2\right)
\le\frac1{24}u^4
\]
\[
|\sin u-u|\le\frac16|u|^3,
\]
valid for all $u\in\R$.
Setting $u=\xi t$, we obtain
\[
0
\le E_{\mathrm{re}}(\xi)
\le\frac1{24}\xi^4\E{X_s\sim g_s/G_{s,0}}{\widetilde X_s^4}
\]
\[
|E_{\mathrm{im}}(\xi)|\le\frac16|\xi|^3\E{X_s\sim g_s/G_{s,0}}{|\widetilde X_s|^3}.
\]
Now, from \Cref{lem:tilted-moment-bounds}, there exist absolute constants
$M_3,M_4>0$ such that
for all
$s\in\R$, we have
$\E{X_s\sim g_s/G_{s,0}}{|\widetilde X_s|^3}\le M_3$ and
$\E{X_s\sim g_s/G_{s,0}}{\widetilde X_s^4}\le M_4$.
Thus we may write
\[
0
\le E_{\mathrm{re}}(\xi)
\le\frac1{24}M_4\xi^4
\]
\[
|E_{\mathrm{im}}(\xi)|\le\frac16M_3|\xi|^3.
\]
It follows that there exists an absolute constant $C_{\mathrm{in}}>0$ such that
for all
$\xi\in[-C_{\mathrm{in}},C_{\mathrm{in}}]$, we have
\[
\left|-\frac12\xi^2+E_{\mathrm{re}}(\xi)+iE_{\mathrm{im}}(\xi)\right|
\le\xi^2
\le\frac12.
\]
Thus, for such $\xi$, we may write
\begin{align*}
\log\left(1-\frac12\xi^2+E_{\mathrm{re}}(\xi)+iE_{\mathrm{im}}(\xi)\right)
&=-\frac12\xi^2+E_{\mathrm{re}}(\xi)+iE_{\mathrm{im}}(\xi)\\
&\quad+\sum_{k=2}^\infty\frac{(-1)^{k+1}}{k}
\left(-\frac12\xi^2+E_{\mathrm{re}}(\xi)+iE_{\mathrm{im}}(\xi)\right)^k,
\end{align*}
where the series is absolutely convergent.
We get the bound
\[
\log\left(1-\frac12\xi^2+E_{\mathrm{re}}(\xi)+iE_{\mathrm{im}}(\xi)\right)
=-\frac12\xi^2+E_{\mathrm{re}}(\xi)+iE_{\mathrm{im}}(\xi)+O(\xi^4),
\]
where $O(\xi^4)$ denotes a complex number that is bounded by some absolute
constant times
$\xi^4$.
From Equation~\eqref{eq:characteristic-function-expansion}, we conclude that for
$\xi\in[-C_{\mathrm{in}},C_{\mathrm{in}}]$, we have
\[
\phi_{\widetilde X_s}(\xi)=\exp\left(-\frac12\xi^2+E_{\mathrm{re}}(\xi)+iE_{\mathrm{im}}(\xi)+O(\xi^4)\right);
\]
i.e. there exists a choice for $O(\xi^4)$ that makes the previous equation hold
for all
$\xi\in[-C_{\mathrm{in}},C_{\mathrm{in}}]$.
Plugging into Equation~\eqref{eq:normalized-sum-characteristic-function} yields
\begin{equation}\label{eq:powered-characteristic-function}
\begin{aligned}
\phi_{\widetilde Y_s}(\xi)
&=\phi_{\widetilde X_s}(\xi)^n
=\exp\left(-\frac n2\xi^2+nE_{\mathrm{re}}(\xi)+inE_{\mathrm{im}}(\xi)+n\cdot O(\xi^4)\right)\\
&=\exp\left(-\frac n2\xi^2\right)
\exp\left(nE_{\mathrm{re}}(\xi)+inE_{\mathrm{im}}(\xi)+n\cdot O(\xi^4)\right)
\end{aligned}
\end{equation}
for all $\xi\in[-C_{\mathrm{in}},C_{\mathrm{in}}]$.

There exists an absolute constant $C_1>0$ such that for all
$\xi\in[-C_1n^{-1/3},C_1n^{-1/3}]$, we have
\[
|nE_{\mathrm{re}}(\xi)+inE_{\mathrm{im}}(\xi)+n\cdot O(\xi^4)|\le O(1),
\]
which implies
\begin{align*}
&\exp(nE_{\mathrm{re}}(\xi)+inE_{\mathrm{im}}(\xi)+n\cdot O(\xi^4))\\
&=1+nE_{\mathrm{re}}(\xi)+inE_{\mathrm{im}}(\xi)+n\cdot O(\xi^4)
+O(1)\cdot|nE_{\mathrm{re}}(\xi)+inE_{\mathrm{im}}(\xi)+n\cdot O(\xi^4)|^2\\
&=1+nE_{\mathrm{re}}(\xi)+inE_{\mathrm{im}}(\xi)+n\cdot O(\xi^4)+n^2\cdot O(\xi^6)
\end{align*}
for all $\xi\in[-C_1n^{-1/3},C_1n^{-1/3}]$.
As long as $n$ is larger than an absolute constant, we will have
$C_1n^{-1/3}\le C_{\mathrm{in}}$, and hence we may plug into
Equation~\eqref{eq:powered-characteristic-function}, which yields
\[
\phi_{\widetilde Y_s}(\xi)=\exp\left(-\frac n2\xi^2\right)
\left(1+nE_{\mathrm{re}}(\xi)+inE_{\mathrm{im}}(\xi)+n\cdot O(\xi^4)+n^2\cdot O(\xi^6)\right)
\]
for all $\xi\in[-C_1n^{-1/3},C_1n^{-1/3}]$.
In what follows, we assume $n$ is larger than said absolute constant.

Set
\[
I_n:=[-C_1n^{-1/3},C_1n^{-1/3}].
\]
We may conclude
\begin{align*}
\int_{I_n}\phi_{\widetilde Y_s}(\xi)\,d\xi
&=\int_{I_n}
\exp\left(-\frac n2\xi^2\right)
\left(1+nE_{\mathrm{re}}(\xi)+n\cdot O(\xi^4)+n^2\cdot O(\xi^6)\right)\,d\xi\\
&=\int_{I_n}
\exp\left(-\frac n2\xi^2\right)
\left(1+n\cdot O(\xi^4)+n^2\cdot O(\xi^6)\right)\,d\xi\\
&=\frac1{n^{1/2}}\left((2\pi)^{1/2}+O\left(n^{-1/6}\exp\left(-\frac{C_1^2}{2}n^{1/3}\right)\right)+O\left(\frac1n\right)\right)\\
&=\frac1{n^{1/2}}\left((2\pi)^{1/2}+O\left(\frac1n\right)\right)
\end{align*}
\begin{equation}\label{eq:fourier-numerator-central}
\begin{aligned}
\int_{I_n}\phi_{\widetilde Y_s}(\xi)\cdot(-i\xi)\,d\xi
&=\int_{I_n}
e^{-n\xi^2/2}\cdot(-i\xi)
\left(inE_{\mathrm{im}}(\xi)+n\cdot O(\xi^4)+n^2\cdot O(\xi^6)\right)\,d\xi\\
&=\int_{I_n}
e^{-n\xi^2/2}
\left(
\begin{aligned}
&n\cdot O(\xi^4)+n\cdot O(|\xi|^5)\\[-0.3ex]
&\qquad+n^2\cdot O(|\xi|^7)
\end{aligned}
\right)\,d\xi\\
&=\int_{I_n}
e^{-n\xi^2/2}(n\cdot O(\xi^4))\,d\xi\\
&=\frac1{n^{1/2}}\cdot O\left(\frac1n\right)
\end{aligned}
\end{equation}
using the fact that $E_{\mathrm{re}}(-\xi)=E_{\mathrm{re}}(\xi)$ and
$E_{\mathrm{im}}(-\xi)=-E_{\mathrm{im}}(\xi)$, as well as
\Cref{lem:aux-gaussian-tail}.

Let us return to~\eqref{eq:powered-characteristic-function}.
There exists an absolute constant $C_2>0$ such that
$[-C_2,C_2]\subseteq[-C_{\mathrm{in}},C_{\mathrm{in}}]$ and such that for all
$\xi\in[-C_2,C_2]$, we have
\[
|nE_{\mathrm{re}}(\xi)+inE_{\mathrm{im}}(\xi)+n\cdot O(\xi^4)|\le\frac n4\xi^2.
\]
From Equation~\eqref{eq:powered-characteristic-function}, we have
\[
|\phi_{\widetilde Y_s}(\xi)|
\le\exp\left(-\frac n2\xi^2+\frac n4\xi^2\right)
=\exp\left(-\frac n4\xi^2\right)
\]
for all $\xi\in[-C_2,C_2]$.
Hence, as long as $n$ is large enough so that
$C_1n^{-1/3}\le C_2$, we have
\begin{equation}\label{eq:fourier-middle-tail}
\begin{aligned}
\left|\int_{[-C_2,C_2]\setminus[-C_1n^{-1/3},C_1n^{-1/3}]}
\phi_{\widetilde Y_s}(\xi)\,d\xi\right|
&\le2C_2\cdot\exp\left(-\frac{C_1^2}{4}n^{1/3}\right)\\
\left|\int_{[-C_2,C_2]\setminus[-C_1n^{-1/3},C_1n^{-1/3}]}
\phi_{\widetilde Y_s}(\xi)\cdot(-i\xi)\,d\xi\right|
&\le2C_2^2\cdot\exp\left(-\frac{C_1^2}{4}n^{1/3}\right).
\end{aligned}
\end{equation}
In what follows, we assume $n$ is larger than said absolute constant.

Suppose $s\le0$.
Then $f_{X_s}(t)$ is a decreasing function of
$t$ on its support, and hence
$f_{\widetilde X_s}(t)$ is also a decreasing function of
$t$ on its support.
Recalling that the support of $\widetilde X_s$ is
\[
\left[-\frac{a_s}{v_s^{1/2}},\frac{1-a_s}{v_s^{1/2}}\right],
\]
we may write
\begin{align*}
\phi_{\widetilde X_s}(\xi)
&=\int_\R f_{\widetilde X_s}(t)\exp(i\xi t)\,dt\\
&=\int_{[-a_s/v_s^{1/2},(1-a_s)/v_s^{1/2}]}
\int_0^{f_{\widetilde X_s}(t)}\exp(i\xi t)\,dw\,dt\\
&=\int_0^{+\infty}
\int_{\{t\in[-a_s/v_s^{1/2},(1-a_s)/v_s^{1/2}]:f_{\widetilde X_s}(t)\ge w\}}
\exp(i\xi t)\,dt\,dw.
\end{align*}
Since $f_{\widetilde X_s}$ is a decreasing function, we see that for each
$w>0$, the set
\[
\left\{t\in\left[-\frac{a_s}{v_s^{1/2}},\frac{1-a_s}{v_s^{1/2}}\right]:f_{\widetilde X_s}(t)\ge w\right\}
\]
is an interval.
Hence we may apply \Cref{lem:aux-oscillatory-integral} to conclude
\begin{align*}
|\phi_{\widetilde X_s}(\xi)|
&\le\int_0^{+\infty}\min\left(
\vol_1\left(\left\{t\in\left[-\frac{a_s}{v_s^{1/2}},\frac{1-a_s}{v_s^{1/2}}\right]:f_{\widetilde X_s}(t)\ge w\right\}\right),
\frac2{|\xi|}\right)\,dw\\
&=\int_{-a_s/v_s^{1/2}}^{-a_s/v_s^{1/2}+2/|\xi|}f_{\widetilde X_s}(t)\,dt.
\end{align*}
If $s>0$, then
$f_{X_s}(t)$ is decreasing and then increasing (i.e. there exists a real number
$t_s^*\in(-a_s/v_s^{1/2},(1-a_s)/v_s^{1/2}]$ such that
$f_{X_s}(t)$ is decreasing for
$t\in(-a_s/v_s^{1/2},t_s^*)$ and increasing for
$t\in(t_s^*,(1-a_s)/v_s^{1/2}))$.
By splitting the integral into two parts at the point $t=t_s^*$, we may use
similar reasoning to conclude
\begin{align*}
|\phi_{\widetilde X_s}(\xi)|
&\le\int_{-a_s/v_s^{1/2}}^{-a_s/v_s^{1/2}+2/|\xi|}f_{\widetilde X_s}(t)\,dt
+\int_{\max((1-a_s)/v_s^{1/2}-2/|\xi|,t_s^*)}^{(1-a_s)/v_s^{1/2}}f_{\widetilde X_s}(t)\,dt\\
&\le\int_{-a_s/v_s^{1/2}}^{-a_s/v_s^{1/2}+2/|\xi|}f_{\widetilde X_s}(t)\,dt
+\frac2{|\xi|}f_{\widetilde X_s}\left(\frac{1-a_s}{v_s^{1/2}}\right).
\end{align*}
Using \Cref{lem:tilted-density-bound-lower,lem:tilted-density-bound-upper},
which give respectively the inverse-square-root
bound near the left endpoint and the uniform bound at the right endpoint, we see
that there exist absolute constants $C_{\mathrm{bound}-},C_{\mathrm{bound}+}>0$
such that
\[
|\phi_{\widetilde X_s}(\xi)|
\le C_{\mathrm{bound}-}\cdot2\left(\frac2{|\xi|}\right)^{1/2}
+C_{\mathrm{bound}+}\cdot\frac2{|\xi|}
\]
for all $s\in\R$ and all
$\xi\in\R\setminus[-C_2,C_2]$.
Thus, assuming $n\ge5$, there exists an absolute constant
$C_3>0$ such that
\begin{equation}\label{eq:fourier-outer-tail}
\begin{aligned}
\int_{\R\setminus[-C_3,C_3]}\left|\phi_{\widetilde Y_s}(\xi)\right|\,d\xi&\le2^{-n}\\
\int_{\R\setminus[-C_3,C_3]}|\xi|\left|\phi_{\widetilde Y_s}(\xi)\right|\,d\xi&\le2^{-n}.
\end{aligned}
\end{equation}
In particular, since both integrands are bounded on the compact interval
$[-C_3,C_3]$, we have
\[
\phi_{\widetilde Y_s}\in L^1(\R),
\qquad
\xi\phi_{\widetilde Y_s}(\xi)\in L^1(\R),
\]
so the Fourier inversion identities culminating in
Equation~\eqref{eq:fourier-density-derivative-ratio} are justified.

Now we turn our attention to the integral
\[
\int_{[-C_3,C_3]\setminus[-C_2,C_2]}\phi_{\widetilde Y_s}(\xi)\,d\xi
\]
(the integral evaluates to $0$ if
$C_3\le C_2$).
Let $\xi\in\R\setminus[-C_2,C_2]$ be arbitrary; we will show that there exists
an absolute constant
$0<C_{\mathrm{defect}}<1$ such that
$|\phi_{\widetilde X_s}(\xi)|\le C_{\mathrm{defect}}$.
Indeed, define the region
\[
R_{\widetilde X_s}:=\{(t,h)\in\R\times\R:0\le h\le f_{\widetilde X_s}(t)\},
\]
so that
\[
\phi_{\widetilde X_s}(\xi)=\iint_{R_{\widetilde X_s}}\exp(i\xi t)\,dt\,dh.
\]
\Cref{cor:tilted-common-constants} implies that there exist absolute constants
$L_-,L_+,L_{\mathrm{vert}}$ with
$L_-<L_+$ and
$0<L_{\mathrm{vert}}$, such that
\[
[L_-,L_+]\times[0,L_{\mathrm{vert}}]\subseteq R_{\widetilde X_s}.
\]
Thus we may write
\begin{align*}
\phi_{\widetilde X_s}(\xi)
&=\iint_{[L_-,L_+]\times[0,L_{\mathrm{vert}}]}\exp(i\xi t)\,dt\,dh
+\iint_{R_{\widetilde X_s}\setminus([L_-,L_+]\times[0,L_{\mathrm{vert}}])}\exp(i\xi t)\,dt\,dh\\
&=L_{\mathrm{vert}}\int_{L_-}^{L_+}\exp(i\xi t)\,dt
+\iint_{R_{\widetilde X_s}\setminus([L_-,L_+]\times[0,L_{\mathrm{vert}}])}\exp(i\xi t)\,dt\,dh,
\end{align*}
and thus
\begin{align*}
|\phi_{\widetilde X_s}(\xi)|
&\le L_{\mathrm{vert}}\left|\int_{L_-}^{L_+}\exp(i\xi t)\,dt\right|
+\iint_{R_{\widetilde X_s}\setminus([L_-,L_+]\times[0,L_{\mathrm{vert}}])}1\,dt\,dh\\
&=1+L_{\mathrm{vert}}\left(\left|\int_{L_-}^{L_+}\exp(i\xi t)\,dt\right|-(L_+-L_-)\right)\\
&\le1-L_{\mathrm{vert}}\cdot(L_+-L_-)\cdot
\min\left(0.1\left(\xi\cdot\frac{L_+-L_-}{2}\right)^2,0.4\right)
\end{align*}
by \Cref{lem:aux-oscillatory-integral-ii}.
Since $|\xi|\ge C_2$, we may take
\[
C_{\mathrm{defect}}=1-L_{\mathrm{vert}}\cdot(L_+-L_-)\cdot
\min\left(0.1\left(C_2\cdot\frac{L_+-L_-}{2}\right)^2,0.4\right),
\]
as desired.

Now that we have the bound $|\phi_{\widetilde X_s}(\xi)|\le C_{\mathrm{defect}}$
for all
$\xi\in\R\setminus[-C_2,C_2]$, we conclude
\begin{equation}\label{eq:fourier-defect-tail}
\begin{aligned}
\left|\int_{[-C_3,C_3]\setminus[-C_2,C_2]}\phi_{\widetilde Y_s}(\xi)\,d\xi\right|
&\le\max(0,2(C_3-C_2))\cdot C_{\mathrm{defect}}^n\\
\left|\int_{[-C_3,C_3]\setminus[-C_2,C_2]}\phi_{\widetilde Y_s}(\xi)\cdot(-i\xi)\,d\xi\right|
&\le\max(0,C_3^2-C_2^2)\cdot C_{\mathrm{defect}}^n.
\end{aligned}
\end{equation}
Finally, we assemble the bounds~\eqref{eq:fourier-numerator-central},
\eqref{eq:fourier-middle-tail}, \eqref{eq:fourier-defect-tail},
and~\eqref{eq:fourier-outer-tail} to conclude
\[
\int_\R\phi_{\widetilde Y_s}(\xi)\,d\xi
=\frac1{n^{1/2}}\left((2\pi)^{1/2}+O\left(\frac1n\right)\right)
\]
\[
\int_\R\phi_{\widetilde Y_s}(\xi)\cdot(-i\xi)\,d\xi
=\frac1{n^{1/2}}\cdot O\left(\frac1n\right),
\]
and plugging into Equation~\eqref{eq:fourier-density-derivative-ratio} yields
\[
\frac{f'_{\widetilde Y_s}(0)}{f_{\widetilde Y_s}(0)}=O\left(\frac1n\right).
\]
Now \Cref{cor:tilted-variance-scale} implies
$(1/v_s)^{1/2}=\Theta(\max(1,|s|))$; thus we may
multiply and obtain
\[
\frac1{v_s^{1/2}}\cdot\frac{f'_{\widetilde Y_s}(0)}{f_{\widetilde Y_s}(0)}
=O\left(\frac{\max(1,|s|)}{n}\right),
\]
which is Equation~\eqref{eq:fourier-density-derivative-target}, as desired.
\end{proof}

\subsection{Integrated density inequalities}
\label{subsec:integrated-density-bound}

\begin{proof}[Proof of \Cref{lem:integrated-density}]
The heuristic behind the proof is the following.
\Cref{lem:density-log-derivative} suggests that
\[
\left.\frac{d}{dt}\right|_{t=na_s}\log f_V(t)\approx-s,
\]
and if we further assume $\frac{d}{dt}\log f_V(t)$ is approximately constant for
$t$ close to
$na_s$, then we can combine the above two approximations, yielding
\[
\frac{d}{dt}\log f_V(t)\approx-s
\]
for $t$ close to
$na_s$, and thus
\[
f_V(t')\approx f_V(na_s)\cdot\exp(-s(t'-na_s))
\]
for $t'$ close to
$na_s$, and thus
\[
\frac{f_V(na_s)}{\int_0^{na_s}f_V(t')\,dt'}
\approx\frac1{\int_{-\infty}^0\exp(-st')\,dt'}
=-s.
\]
The remainder of this proof will involve quantifying the error in the two
approximations mentioned previously.

Let us first prove the lower bound.
Let $s\le s_0<0$ be fixed, and assume
$n\ge\max(n_{\mathrm{deriv}},C_{\mathrm{deriv}})$ (where the constants
$n_{\mathrm{deriv}},C_{\mathrm{deriv}}$ are from
\Cref{lem:density-log-derivative}).
For all $s'\le s$, we have
\begin{align*}
\left.\frac{d}{dt}\right|_{t=na_{s'}}\log f_V(t)
&\ge -s'-C_{\mathrm{deriv}}\cdot\frac{\max(1,|s'|)}n \qquad\text{(from \Cref{lem:density-log-derivative})}\\
&\ge -s-C_{\mathrm{deriv}}\cdot\frac{\max(1,|s|)}n \qquad\text{(since $n\ge C_{\mathrm{deriv}}$)}\\
&\ge -s-C_{\mathrm{deriv}}\max(1/|s_0|,1)\cdot\frac{|s|}{n}\\
&=-s-C_{\mathrm{int-lo}}\cdot\frac{|s|}{n}
\end{align*}
where we have defined
$C_{\mathrm{int-lo}}=C_{\mathrm{int-lo}}(s_0):=C_{\mathrm{deriv}}\max(1/|s_0|,1)$.
If $-s-C_{\mathrm{int-lo}}|s|/n$ is nonpositive then we are done, so henceforth
assume that
$-s-C_{\mathrm{int-lo}}|s|/n>0$.
Recalling that $a_s$ is a continuous function of
$s$ (from \Cref{lem:tilted-continuity}), and that
$\lim_{s\searrow-\infty}a_s=0$ (from \Cref{lem:tilted-limits}), we conclude
\[
\frac{d}{dt}\log f_V(t)\ge -s-C_{\mathrm{int-lo}}\cdot\frac{|s|}{n}
\]
holds for all $0<t\le na_s$.
Integrating yields
\[
f_V(t')\le f_V(na_s)\cdot\exp\left(\left(-s-C_{\mathrm{int-lo}}\cdot\frac{|s|}{n}\right)\cdot(t'-na_s)\right)
\]
for all $0<t'\le na_s$, and thus
\begin{align*}
\frac{f_V(na_s)}{\int_0^{na_s}f_V(t')\,dt'}
&\ge\frac1{\int_0^{na_s}\exp\left(\left(-s-C_{\mathrm{int-lo}}\cdot\frac{|s|}{n}\right)\cdot(t'-na_s)\right)\,dt'}\\
&=\frac1{\int_{-na_s}^0\exp\left(\left(-s-C_{\mathrm{int-lo}}\cdot\frac{|s|}{n}\right)\cdot t'\right)\,dt'}\\
&\ge\frac1{\int_{-\infty}^0\exp\left(\left(-s-C_{\mathrm{int-lo}}\cdot\frac{|s|}{n}\right)\cdot t'\right)\,dt'}\\
&=-s-C_{\mathrm{int-lo}}\cdot\frac{|s|}{n},
\end{align*}
which proves the lower bound (indeed, for the lower bound, we do not need
$n_{\mathrm{int}}$ to depend on
$s_0$).

Now we prove the upper bound.
Let $s\le s_0<0$ be fixed.
First assume $n\ge n_{\mathrm{deriv}}$.
Then we have, for all $s'\le s$,
\begin{align*}
f_V(na_{s'})
&=f_V(na_s)\cdot\exp\left(-\int_{na_{s'}}^{na_s}\frac{d}{dt}\log f_V(t)\,dt\right)\nonumber\\
&\ge f_V(na_s)\cdot
\exp\left(-\int_{na_{s'}}^{na_s}
\begin{aligned}[t]
\Bigl(&-a^{-1}(t''/n)\\[-0.3ex]
&+C_{\mathrm{deriv}}\cdot\frac{\max(1,|a^{-1}(t''/n)|)}{n}\Bigr)
\end{aligned}
\,dt''\right)
\qquad\text{(from \Cref{lem:density-log-derivative})}
\end{align*}
where we write $a^{-1}:(0,1)\to\R$ to be the inverse function of
$s\mapsto a_s$; i.e.
$a^{-1}(a_s)=s$ for all
$s$.
Since $a_s$ is a strictly increasing function of
$s$, we have that
$a^{-1}$ is also a strictly increasing function, and thus
\[
s'
\le a^{-1}(t''/n)
\le s
<0
\]
for all $t''\in[na_{s'},na_s]$.
Thus
\begin{equation}\label{eq:integrated-density-local-lower}
\begin{aligned}
f_V(na_{s'})
&\ge f_V(na_s)\cdot\exp\left(-\int_{na_{s'}}^{na_s}\left(-a^{-1}(t''/n)+C_{\mathrm{deriv}}\cdot\frac{\max(1,|s'|)}n\right)\,dt''\right)\\
&=f_V(na_s)\cdot\exp(-s(na_{s'}-na_s))\\
&\quad\cdot\exp\left(\int_{na_{s'}}^{na_s}\left((a^{-1}(t''/n)-a^{-1}(na_s/n))-C_{\mathrm{deriv}}\cdot\frac{\max(1,|s'|)}n\right)\,dt''\right).
\end{aligned}
\end{equation}
We would want to say
\[
a^{-1}(t''/n)-a^{-1}(na_s/n)
\approx(t''-na_s)\cdot\left.\frac{d}{dt}\right|_{t=na_s}a^{-1}(t/n)
=(t''-na_s)/n\cdot\left(\frac{d}{ds}a_s\right)^{-1}
\]
for all $t''\in[na_{s'},na_s]$.

First, differentiating the quotient defining $a_s$ gives
\[
\begin{aligned}
\frac{d}{ds}a_s
&=\frac{G_{s,1}'G_{s,0}-G_{s,1}G_{s,0}'}{G_{s,0}^2}\\
&=\frac{G_{s,2}G_{s,0}-G_{s,1}^2}{G_{s,0}^2}\\
&=\frac{G_{s,2}}{G_{s,0}}-\left(\frac{G_{s,1}}{G_{s,0}}\right)^2
=v_s.
\end{aligned}
\]
Here $G_{s,p}'=G_{s,p+1}$ by differentiation under the integral sign.
By \Cref{cor:tilted-variance-scale}, we have
\[
\frac{d}{ds}a_s\ge\frac{C_{\mathrm{var}-}^2}{\max(1,|s|^2)}
\]
for some absolute constant $C_{\mathrm{var}-}>0$.
Thus, for all $s''\in[s-\max(1,|s|),s]$, we have
\[
\left.\frac{d}{ds}\right|_{s=s''}a_s
\ge\frac{C_{\mathrm{var}-}^2}{\max(1,|s''|^2)}
\ge\frac{C_{\mathrm{var}-}^2}{\max(1,(|s|+\max(1,|s|))^2)}
=\frac{C_{\mathrm{var}-}^2}{(|s|+\max(1,|s|))^2}.
\]
Integrating yields
\begin{equation}\label{eq:mean-window}
\begin{aligned}
a_s-a_{s-\max(1,|s|)}
&=\int_{s-\max(1,|s|)}^s\left(\left.\frac{d}{ds}\right|_{s=s''}a_s\right)\,ds''\\
&\ge\max(1,|s|)\cdot\frac{C_{\mathrm{var}-}^2}{(|s|+\max(1,|s|))^2}\\
&\ge\max(1,|s|)\cdot\frac{C_{\mathrm{var}-}^2}{(2\max(1,|s|))^2}\\
&=\frac{C_{\mathrm{var}-}^2}{4}\cdot\frac1{\max(1,|s|)}.
\end{aligned}
\end{equation}
Let us assume $n\ge3$ is large enough so that
\[
\frac{\log n}{n} \cdot\frac1{|s_0|}
\le\frac{C_{\mathrm{var}-}^2}{4}\cdot\frac1{\max(1,|s_0|)}.
\]
Then, since
\[
\frac{|s_0|}{\max(1,|s_0|)}
=\min(1,|s_0|)
\le\min(1,|s|)
=\frac{|s|}{\max(1,|s|)},
\]
we have
\[
\frac{\log n}{n}\cdot\frac1{|s|}
\le\frac{C_{\mathrm{var}-}^2}{4}\cdot\frac1{\max(1,|s|)}.
\]
From Equation~\eqref{eq:mean-window}, we see that
\[
a_s-a_{s-\max(1,|s|)}\ge\frac{\log n}{n}\cdot\frac1{|s|},
\]
and hence there exists $s^*\in[s-\max(1,|s|),s]$ such that
\[
a_s-a_{s^*}=\frac{\log n}{n}\cdot\frac1{|s|}.
\]
We henceforth assume $s'\ge s^*$.

From above, we know that
\[
\left.\frac{d}{ds}\right|_{s=s''}a_s\ge\frac{C_{\mathrm{var}-}^2}{(|s|+\max(1,|s|))^2}
\]
for all $s''\in[s',s]$.
Thus, for all $t''\in[na_{s'},na_s]$, we have
\begin{align*}
a^{-1}(t''/n)-a^{-1}(na_s/n)
&=-\int_{t''}^{na_s}\left(\left.\frac{d}{dt}\right|_{t=t'''}a^{-1}(t/n)\right)\,dt'''\\
&=-\int_{t''}^{na_s}\frac1n\cdot\left(\left.\frac{d}{ds}\right|_{s=a^{-1}(t'''/n)}a_s\right)^{-1}\,dt'''\\
&\ge-\int_{t''}^{na_s}\frac1n\cdot\left(\frac{C_{\mathrm{var}-}^2}{(|s|+\max(1,|s|))^2}\right)^{-1}\,dt'''\\
&=(t''-na_s)/n\cdot\frac{(|s|+\max(1,|s|))^2}{C_{\mathrm{var}-}^2}.
\end{align*}
It follows that
\[
\int_{na_{s'}}^{na_s}(a^{-1}(t''/n)-a^{-1}(na_s/n))\,dt''
\ge-\frac12(na_s-na_{s'})^2\cdot\frac1n\cdot\frac{(|s|+\max(1,|s|))^2}{C_{\mathrm{var}-}^2}.
\]
Plugging into Equation~\eqref{eq:integrated-density-local-lower}, we have
\begin{align*}
f_V(na_{s'})
&\ge f_V(na_s)\cdot\exp(-s(na_{s'}-na_s))\\
&\quad\cdot\exp\left(-\frac12(na_s-na_{s'})^2\cdot\frac1n\cdot\frac{(|s|+\max(1,|s|))^2}{C_{\mathrm{var}-}^2}
-C_{\mathrm{deriv}}(na_s-na_{s'})\cdot\frac{\max(1,|s'|)}n\right)\\
&\ge f_V(na_s)\cdot\exp(-s(na_{s'}-na_s))\\
&\quad\cdot\exp\left(-\frac12(na_s-na_{s'})^2\cdot\frac1n\cdot\frac{(2\max(1,|s|))^2}{C_{\mathrm{var}-}^2}
-C_{\mathrm{deriv}}(na_s-na_{s'})\cdot\frac{2\max(1,|s|)}n\right)
\end{align*}
for all $s'\in[s^*,s]$.
Making the substitution $t'=na_{s'}$, we conclude that we have
\begin{align*}
f_V(t')
&\ge f_V(na_s)\cdot\exp(-s(t'-na_s))\\
&\quad\cdot\exp\left(-A_s(na_s-t')^2-B_s(na_s-t')\right)
\end{align*}
for all $t'\in[na_s-\frac{\log n}{|s|},na_s]$.

For $u<0$, write
\[
R(u):=\frac{\max\{1,|u|\}}{|u|},
\qquad
D(u):=1+\frac{4R(u)^2}{C_{\mathrm{var}-}^2}+2C_{\mathrm{deriv}}R(u).
\]
For the fixed $s<0$, also set
\[
A_s:=\frac{2\max\{1,|s|\}^2}{nC_{\mathrm{var}-}^2},
\qquad
B_s:=\frac{2C_{\mathrm{deriv}}\max\{1,|s|\}}n.
\]

Integrating, we have
\begin{align*}
\int_{na_s-\frac{\log n}{|s|}}^{na_s}f_V(t')\,dt'
&\ge f_V(na_s)\cdot\int_{na_s-\frac{\log n}{|s|}}^{na_s}\exp(-s(t'-na_s))\\
&\qquad\cdot\exp\left(-A_s(na_s-t')^2-B_s(na_s-t')\right)\,dt'\\
&\ge f_V(na_s)\cdot\int_{na_s-\frac{\log n}{|s|}}^{na_s}\exp(-s(t'-na_s))\\
&\qquad\cdot\left(1-A_s(na_s-t')^2-B_s(na_s-t')\right)\,dt'\\
&\ge f_V(na_s)\cdot\Bigl(\frac1{|s|}\cdot\left(1-\frac1n\right)
-\int_{na_s-\frac{\log n}{|s|}}^{na_s}\exp(-s(t'-na_s))\\
&\qquad\cdot\left(A_s(na_s-t')^2+B_s(na_s-t')\right)\,dt'\Bigr)\\
&\ge f_V(na_s)\cdot\Bigl(\frac1{|s|}\cdot\left(1-\frac1n\right)
-\int_{-\infty}^{0}\exp(-st')\\
&\qquad\cdot\left(A_s(-t')^2+B_s(-t')\right)\,dt'\Bigr)\\
&=f_V(na_s)\cdot\frac1{|s|}\cdot\left(1-\frac{D(s)}n\right)\\
&\ge f_V(na_s)\cdot\frac1{|s|}\cdot\left(1-\frac{D(s_0)}n\right)\\
&=f_V(na_s)\cdot\frac1{|s|}\cdot\left(1-\frac{C_{\mathrm{int-up}}}{2n}\right),
\end{align*}
where we have defined
\[
C_{\mathrm{int-up}}
=C_{\mathrm{int-up}}(s_0)
:=2D(s_0)
\]
We conclude
\begin{align*}
\frac{f_V(na_s)}{\int_0^{na_s}f_V(t')\,dt'}
&\le\frac{f_V(na_s)}{\int_{na_s-\frac{\log n}{|s|}}^{na_s}f_V(t')\,dt'}\\
&\le |s|\cdot\left(1-\frac{C_{\mathrm{int-up}}}{2n}\right)^{-1}\\
&\le |s|\cdot\left(1+\frac{C_{\mathrm{int-up}}}{n}\right)
\end{align*}
which holds as long as $n\ge C_{\mathrm{int-up}}$.
This proves the upper bound, as desired.
\end{proof}

\subsection{Auxiliary inequalities for the density estimates}
\label{subsec:ellipsoid-growth-auxiliary-inequalities}

\begin{lemma}\label{lem:aux-logconcave-ratio}
Let $L_0>0$, and let
$f:[0,+\infty)\to\R$ be a logconcave probability density function which is
continuous, such that
$f$ is nondecreasing on
$[0,L_0]$ and
$\int_0^{L_0}f(t)\,dt>0$.
Then
\[
\frac12L_0
\le \left.\frac{d}{dL}\frac{\int_0^L t^2\cdot f(t)\,dt}{\int_0^L f(t)\,dt}\right|_{L=L_0}
\le2L_0.
\]
\end{lemma}
\begin{proof}
We compute
\begin{align*}
\frac{d}{dL}\frac{\int_0^L t^2\cdot f(t)\,dt}{\int_0^L f(t)\,dt}
&=\frac{L^2\cdot f(L)}{\int_0^Lf(t)\,dt}
-\frac{f(L)\int_0^Lt^2\cdot f(t)\,dt}{\left(\int_0^Lf(t)\,dt\right)^2}\\
&=\frac{f(L)\int_0^L(L^2-t^2)\cdot f(t)\,dt}{\left(\int_0^Lf(t)\,dt\right)^2}
\end{align*}
and thus
\[
\left.\frac{d}{dL}\frac{\int_0^L t^2\cdot f(t)\,dt}{\int_0^L f(t)\,dt}\right|_{L=L_0}
=\frac{f(L_0)\int_0^{L_0}(L_0^2-t^2)\cdot f(t)\,dt}{\left(\int_0^{L_0}f(t)\,dt\right)^2}.
\]
For the lower bound, we have
\begin{align*}
\left(\int_0^{L_0}f(t)\,dt\right)^2
&=2\int_0^{L_0}\int_t^{L_0}f(t)f(s)\,ds\,dt\\
&\le2\int_0^{L_0}\int_t^{L_0}f(t)f(L_0)\,ds\,dt \qquad\text{($f$ is nondecreasing on $[0,L_0]$)}\\
&=2f(L_0)\int_0^{L_0}(L_0-t)\cdot f(t)\,dt\\
&\le2f(L_0)\int_0^{L_0}(L_0-t)\frac{L_0+t}{L_0}\cdot f(t)\,dt\\
&=\frac2{L_0}f(L_0)\int_0^{L_0}(L_0^2-t^2)\cdot f(t)\,dt
\end{align*}
which implies
\[
\left.\frac{d}{dL}\frac{\int_0^L t^2\cdot f(t)\,dt}{\int_0^L f(t)\,dt}\right|_{L=L_0}
=\frac{f(L_0)\int_0^{L_0}(L_0^2-t^2)\cdot f(t)\,dt}{\left(\int_0^{L_0}f(t)\,dt\right)^2}
\ge\frac12L_0,
\]
as desired.

For the upper bound, we have
\begin{align*}
\left(\int_0^{L_0}f(t)\,dt\right)^2
&=\int_0^{L_0}\int_0^{L_0}f(t)f(s)\,ds\,dt\\
&\ge\iint_{t,s\in[0,L_0],\ t+s\ge L_0}f(t)f(s)\,ds\,dt\\
&\ge\iint_{t,s\in[0,L_0],\ t+s\ge L_0}f(L_0)f(t+s-L_0)\,ds\,dt \qquad\text{(logconcavity of $f$)}\\
&=f(L_0)\int_0^{L_0}(L_0-t)\cdot f(t)\,dt\\
&\ge f(L_0)\int_0^{L_0}(L_0-t)\frac{L_0+t}{2L_0}\cdot f(t)\,dt\\
&=\frac1{2L_0}f(L_0)\int_0^{L_0}(L_0^2-t^2)\cdot f(t)\,dt
\end{align*}
which implies
\[
\left.\frac{d}{dL}\frac{\int_0^L t^2\cdot f(t)\,dt}{\int_0^L f(t)\,dt}\right|_{L=L_0}
=\frac{f(L_0)\int_0^{L_0}(L_0^2-t^2)\cdot f(t)\,dt}{\left(\int_0^{L_0}f(t)\,dt\right)^2}
\le2L_0,
\]
as desired.
\end{proof}

\begin{lemma}\label{lem:aux-gaussian-tail}
For all $u>0$, we have
\[
\int_u^{+\infty}\frac1{(2\pi)^{1/2}}e^{-x^2/2}\,dx
<\frac1u\frac1{(2\pi)^{1/2}}e^{-u^2/2}.
\]
\end{lemma}
\begin{proof}
Observe
\[
\int_u^{+\infty}\frac1{(2\pi)^{1/2}}e^{-x^2/2}\,dx
<\int_u^{+\infty}\frac1{(2\pi)^{1/2}}e^{-x^2/2}\cdot\frac xu\,dx
=\frac1u\frac1{(2\pi)^{1/2}}e^{-u^2/2}.
\]
\end{proof}

\begin{lemma}\label{lem:aux-oscillatory-integral}
Let $\xi\in\R$, and let
$t_1<t_2$ be real numbers.
Then we have
\[
\left|\int_{t_1}^{t_2}\exp(i\xi t)\,dt\right|
\le\min\left(t_2-t_1,\frac2{|\xi|}\right).
\]
\end{lemma}
\begin{proof}
The inequality
\[
\left|\int_{t_1}^{t_2}\exp(i\xi t)\,dt\right|\le t_2-t_1
\]
follows from the triangle inequality.
The inequality
\[
\left|\int_{t_1}^{t_2}\exp(i\xi t)\,dt\right|\le\frac2{|\xi|}
\]
is proven as follows:
\[
\left|\int_{t_1}^{t_2}\exp(i\xi t)\,dt\right|
=\left|\frac1{i\xi}(\exp(i\xi t_2)-\exp(i\xi t_1))\right|
\le\frac{|\exp(i\xi t_2)|+|\exp(i\xi t_1)|}{|\xi|}
=\frac2{|\xi|}.
\]
\end{proof}

\begin{lemma}\label{lem:aux-sinc-bound}
For all $u\in\R$, we have
\[
\left|\frac{\sin u}{u}\right|\le\max(1-0.1u^2,0.6).
\]
\end{lemma}
\begin{proof}
First assume $|u|\le2$; we need to prove
$|\sin u/u|\le1-0.1u^2$.
We may further assume $0<u\le2$.
Now calculate
\begin{align*}
\left|\frac{\sin u}{u}\right|
&=\frac{\sin u}{u}
=1-\frac{u^2}{3!}+\frac{u^4}{5!}-\frac{u^6}{7!}+\cdots\\
&=1-\frac16u^2+\frac1{120}u^4+
\left(-\frac{u^6}{7!}+\frac{u^8}{9!}\right)
+\left(-\frac{u^{10}}{11!}+\frac{u^{12}}{13!}\right)+\cdots\\
&\le1-\frac16u^2+\frac1{120}u^4\\
&=(1-0.1u^2)+\left(-\frac1{15}u^2+\frac1{120}u^4\right)\\
&\le1-0.1u^2
\end{align*}
as desired.

Now assume $|u|>2$; we need to prove
$|\sin u/u|\le0.6$.
But this follows from
\[
\left|\frac{\sin u}{u}\right|
=\frac{|\sin u|}{|u|}
\le\frac1{|u|}
\le\frac12
<0.6.
\]
\end{proof}

\begin{lemma}\label{lem:aux-oscillatory-integral-ii}
Let $\xi\in\R$, and let
$t_1<t_2$ be real numbers.
Then we have
\[
\left|\int_{t_1}^{t_2}\exp(i\xi t)\,dt\right|
\le(t_2-t_1)\cdot\max\left(1-0.1\left(\xi\cdot\frac{t_2-t_1}{2}\right)^2,0.6\right).
\]
\end{lemma}
\begin{proof}
We calculate
\begin{align*}
\left|\int_{t_1}^{t_2}\exp(i\xi t)\,dt\right|
&=\left|\exp\left(i\xi\cdot\frac{t_1+t_2}{2}\right)\cdot
\int_{(t_1-t_2)/2}^{(t_2-t_1)/2}\exp(i\xi t)\,dt\right|\\
&=\left|\int_{(t_1-t_2)/2}^{(t_2-t_1)/2}\exp(i\xi t)\,dt\right|\\
&=\left|\frac1{i\xi}\left(\exp\left(i\xi\cdot\frac{t_2-t_1}{2}\right)-\exp\left(i\xi\cdot\frac{t_1-t_2}{2}\right)\right)\right|\\
&=\left|\frac2\xi\sin\left(\xi\cdot\frac{t_2-t_1}{2}\right)\right|\\
&=(t_2-t_1)\cdot\left|\frac{\sin(\xi\cdot\frac{t_2-t_1}{2})}{\xi\cdot\frac{t_2-t_1}{2}}\right|\\
&\le(t_2-t_1)\cdot\max\left(1-0.1\left(\xi\cdot\frac{t_2-t_1}{2}\right)^2,0.6\right)
\end{align*}
using \Cref{lem:aux-sinc-bound}.
\end{proof}

\subsection[Properties of the Normalized and Tilted Distribution]{Properties of
the Normalized and Tilted Distribution \texorpdfstring{$\widetilde X_s$}{X tilde
s}}
Recall that $X_s$ is defined to be the random variable whose distribution is
supported on
$[0,1]$ and whose density function is proportional to
\[
\frac1{2t^{1/2}}\cdot\exp(st).
\]
So, $X_s$ is an exponential tilt of
$X_0$.
We further define the normalized random variable
\[
\widetilde X_s:=\frac{X_s-a_s}{v_s^{1/2}},
\]
where $a_s:=\E{X_s\sim g_s/G_{s,0}}{X_s}$ and
$v_s:=\V{X_s\sim g_s/G_{s,0}}{X_s}$.
Observe that the support of $\widetilde X_s$ is
\[
\left[-\frac{a_s}{v_s^{1/2}},\frac{1-a_s}{v_s^{1/2}}\right].
\]

\begin{lemma}\label{lem:tilted-continuity}
The functions $a_s:=\E{X_s\sim g_s/G_{s,0}}{X_s}$ and
$v_s:=\V{X_s\sim g_s/G_{s,0}}{X_s}$ and
\[
G_{s,0}:=\int_0^1\frac1{2t^{1/2}}\cdot\exp(st)\,dt
\]
are each continuous functions of $s\in\R$.
Furthermore, $v_s$ and
$G_{s,0}$ are each positive for all
$s\in\R$.
\end{lemma}
\begin{proof}
First note that for any fixed $p>-\frac12$, the integral
\[
G_{s,p}:=\int_0^1\frac1{2t^{1/2}}\cdot\exp(st)\cdot t^p\,dt
\]
is a continuous function of $s\in\R$; this follows from Lebesgue dominated
convergence theorem.
Furthermore, note that $G_{s,p}>0$ for all
$p>-\frac12$ and
$s\in\R$.
Now, we have
\[
a_s
=\E{X_s\sim g_s/G_{s,0}}{X_s}
=\frac{G_{s,1}}{G_{s,0}}
\]
\[
v_s
=\V{X_s\sim g_s/G_{s,0}}{X_s}
=\frac{G_{s,2}}{G_{s,0}}-\left(\frac{G_{s,1}}{G_{s,0}}\right)^2
\]
so both $a_s$ and
$v_s$ are continuous functions of
$s\in\R$.
The inequality $v_s>0$ follows from the fact that the distribution of
$X_s$ has a density function (and hence equality cannot hold in the
Cauchy-Schwarz inequality).
\end{proof}

\noindent\textbf{Note.}
In fact, for each $p>-\frac12$, the integral
$G_{s,p}:=\int_0^1\frac1{2t^{1/2}}\cdot\exp(st)\cdot t^p\,dt$ as defined above
is an entire holomorphic function of
$s$.

\begin{lemma}\label{lem:tilted-limits}
The following four limits exist and are positive and finite:
\[
\lim_{s\searrow-\infty}-sa_s
\qquad
\lim_{s\nearrow+\infty}s(1-a_s)
\]
\[
\lim_{s\searrow-\infty}-sv_s^{1/2}
\qquad
\lim_{s\nearrow+\infty}sv_s^{1/2}.
\]
\end{lemma}
\begin{proof}
Let us first consider the case when $s\searrow-\infty$.
For fixed $s<0$, the random variable
$-sX_s$ has distribution supported on
$[0,-s]$ and density function proportional to
$\frac1{t^{1/2}}\cdot\exp(-t)$.
Now
\[
-sa_s
=\E{X_s\sim g_s/G_{s,0}}{-sX_s}
=\left(\int_0^{-s}\frac1{t^{1/2}}\cdot\exp(-t)\cdot t\,dt\right) \left(\int_0^{-s}\frac1{t^{1/2}}\cdot\exp(-t)\,dt\right)^{-1}
\]
\[
\begin{aligned}
(-s)^2v_s
&=\V{X_s\sim g_s/G_{s,0}}{-sX_s}\\
&=\left(\int_0^{-s}\frac1{t^{1/2}}\cdot\exp(-t)
\cdot(t-\E{X_s\sim g_s/G_{s,0}}{-sX_s})^2\,dt\right)\\
&\quad\cdot\left(\int_0^{-s}\frac1{t^{1/2}}\cdot\exp(-t)\,dt\right)^{-1}.
\end{aligned}
\]
From the Lebesgue dominated convergence theorem, both limits
\[
\lim_{s\searrow-\infty}(-s)a_s
\qquad\text{and}\qquad
\lim_{s\searrow-\infty}(-s)^2v_s
\]
exist and are positive and finite, as desired.

Now let us consider the case when $s\nearrow+\infty$.
For fixed $s>0$, the random variable
$s(X_s-1)$ has distribution supported on
$[-s,0]$ and density function proportional to
$(\frac ts+1)^{-1/2}\cdot\exp(t)$.
We would like to say that the law of $s(X_s-1)$ ``converges'' in some sense to
the distribution
$\exp(t)\,dt$ (supported on
$(-\infty,0]$).

Observe
\begin{equation}\label{eq:tilted-endpoint-moments}
\begin{aligned}
s(a_s-1)&=\E{X_s\sim g_s/G_{s,0}}{s(X_s-1)}\\
&=\left(\int_{-s}^0\left(\frac ts+1\right)^{-1/2}
\cdot\exp(t)\cdot t\,dt\right)\\
&\quad\cdot\left(\int_{-s}^0\left(\frac ts+1\right)^{-1/2}
\cdot\exp(t)\,dt\right)^{-1}\\
s^2v_s&=\V{X_s\sim g_s/G_{s,0}}{s(X_s-1)}\\
&=\left(\int_{-s}^0\left(\frac ts+1\right)^{-1/2}\cdot\exp(t)
\cdot(t-\E{X_s\sim g_s/G_{s,0}}{s(X_s-1)})^2\,dt\right)\\
&\quad\cdot\left(\int_{-s}^0\left(\frac ts+1\right)^{-1/2}\cdot\exp(t)\,dt\right)^{-1}.
\end{aligned}
\end{equation}
However, this time we cannot apply Lebesgue dominated convergence theorem to
conclude that $\lim_{s\nearrow+\infty}s(a_s-1)$ exists, due to the singularity
of
$(\frac ts+1)^{-1/2}$ near
$t=-s$.
Instead, we split the integral into two parts.
For instance, by writing
\begin{equation}\label{eq:endpoint-integral-split}
\int_{-s}^0\left(\frac ts+1\right)^{-1/2}\cdot\exp(t)\,dt
=\int_{-s}^{-s/2}\left(\frac ts+1\right)^{-1/2}\cdot\exp(t)\,dt
+\int_{-s/2}^{0}\left(\frac ts+1\right)^{-1/2}\cdot\exp(t)\,dt
\end{equation}
and noting that
\[
0
\le\int_{-s}^{-s/2}\left(\frac ts+1\right)^{-1/2}\cdot\exp(t)\,dt
\le\exp(-s/2)\int_{-s}^{-s/2}\left(\frac ts+1\right)^{-1/2}\,dt
=\exp(-s/2)\cdot2^{1/2}\,s,
\]
we see that $\int_{-s}^{-s/2}(\frac ts+1)^{-1/2}\cdot\exp(t)\,dt\to0$ as
$s\nearrow+\infty$.
Now, we can use the Lebesgue dominated convergence theorem on the integral
$\int_{-s/2}^0(\frac ts+1)^{-1/2}\cdot\exp(t)\,dt$, which yields
\[
\int_{-s/2}^0\left(\frac ts+1\right)^{-1/2}\cdot\exp(t)\,dt
\to\int_{-\infty}^{0}\exp(t)\,dt=1
\]
as $s\nearrow+\infty$.
Plugging into Equation~\eqref{eq:endpoint-integral-split}, we see that
\begin{equation}\label{eq:endpoint-integral-limit}
\lim_{s\nearrow+\infty}\int_{-s}^{0}\left(\frac ts+1\right)^{-1/2}\cdot\exp(t)\,dt=1.
\end{equation}
In a similar manner, we can show that
\[
\lim_{s\nearrow+\infty}\int_{-s}^{0}\left(\frac ts+1\right)^{-1/2}\cdot\exp(t)\cdot t\,dt=-1.
\]
Thus, from Equation~\eqref{eq:tilted-endpoint-moments}, we see that
$\lim_{s\nearrow+\infty}s(a_s-1)$ exists and is negative and finite.
In a similar manner, we can show that $\lim_{s\nearrow+\infty}s^2v_s$ exists and
is positive and finite, as desired.
\end{proof}

\begin{corollary}\label{cor:tilted-variance-scale}
There exist absolute constants $0<C_{\mathrm{var}-}<C_{\mathrm{var}+}$ such that
\[
\frac{C_{\mathrm{var}-}}{\max(1,|s|)}
\le v_s^{1/2}
\le\frac{C_{\mathrm{var}+}}{\max(1,|s|)}
\]
for all $s\in\R$.
\end{corollary}
\begin{proof}
From \Cref{lem:tilted-continuity}, we see that $v_s^{1/2}\cdot\max(1,|s|)$ is a
continuous
function of
$s\in\R$.
Furthermore, we have $v_s^{1/2}\cdot\max(1,|s|)>0$ for each
$s$.
By \Cref{lem:tilted-limits}, we have that the two limits
\[
\lim_{s\searrow-\infty}-sv_s^{1/2}
\qquad
\lim_{s\nearrow+\infty}sv_s^{1/2}
\]
both exist and are positive and finite.
This is sufficient to prove the existence of $C_{\mathrm{var}-}$ and
$C_{\mathrm{var}+}$, as desired.
\end{proof}

\begin{corollary}\label{cor:tilted-support-interval}
There exist absolute constants $L_-<0<L_+$ such that
\[
[L_-,L_+]\subseteq\operatorname{supp}\widetilde X_s
=\left[-\frac{a_s}{v_s^{1/2}},\frac{1-a_s}{v_s^{1/2}}\right]
\]
for all $s\in\R$; i.e.
$-a_s/v_s^{1/2}\le L_-$ and
$L_+\le(1-a_s)/v_s^{1/2}$ hold for all
$s\in\R$.
\end{corollary}
\begin{proof}
From \Cref{lem:tilted-continuity}, we see that both $-a_s/v_s^{1/2}$ and
$(1-a_s)/v_s^{1/2}$ are continuous functions of
$s\in\R$.
Furthermore, we have
\[
-\frac{a_s}{v_s^{1/2}}
<0
<\frac{1-a_s}{v_s^{1/2}}
\]
for each $s$.
Thus it suffices to show the following four statements:
\[
\lim_{s\searrow-\infty}-\frac{a_s}{v_s^{1/2}}\text{ exists and is }<0
\qquad
\lim_{s\nearrow+\infty}-\frac{a_s}{v_s^{1/2}}=-\infty
\]
\[
\lim_{s\searrow-\infty}\frac{1-a_s}{v_s^{1/2}}=+\infty
\qquad
\lim_{s\nearrow+\infty}\frac{1-a_s}{v_s^{1/2}}\text{ exists and is }>0.
\]
This is a direct consequence of \Cref{lem:tilted-limits}.
\end{proof}

\begin{lemma}\label{lem:tilted-density-floor}
Let $L_-<0<L_+$ be any real numbers satisfying the conclusion of
\Cref{cor:tilted-support-interval}
(namely,
$[L_-,L_+]\subseteq\operatorname{supp}\widetilde
X_s=[-a_s/v_s^{1/2},(1-a_s)/v_s^{1/2}]$
for all
$s\in\R$).
There exists $L_{\mathrm{vert}}>0$, depending only on
$L_-$ and
$L_+$, such that
$f_{\widetilde X_s}(t)\ge L_{\mathrm{vert}}$ for all
$s\in\R$ and for all
$t\in(L_-,L_+]$ (here
$f_{\widetilde X_s}$ denotes the density function of
$\widetilde X_s$).
\end{lemma}
\begin{proof}
Define the function
\[
h(s,t):=\frac1{G_{s,0}}v_s^{1/2}\frac1{2(v_s^{1/2}t+a_s)^{1/2}}\cdot\exp(s(v_s^{1/2}t+a_s))
\]
with domain $\operatorname{dom}h=\{(s,t)\in\R\times\R:t>-a_s/v_s^{1/2}\}$, where
\[
G_{s,0}:=\int_0^1\frac1{2(t')^{1/2}}\cdot\exp(st')\,dt'.
\]
By assumption, we have that for all $s\in\R$ and for all
$t\in(L_-,L_+]$,
$h(s,t)$ is well-defined and equals
$f_{\widetilde X_s}(t)$.
Thus it suffices to give a lower bound for $h(s,t)$ on the domain
$(s,t)\in\R\times(L_-,L_+]$.

By \Cref{lem:tilted-continuity}, we see that $h$ is a jointly continuous
function of
$(s,t)$ which is positive everywhere.
We also see that $(v_s^{1/2}t+a_s)^{1/2}\cdot h(s,t)$ can be extended to a
positive and jointly continuous function on
$\R\times\R$ (even though the multiplicand
$(v_s^{1/2}t+a_s)^{1/2}$ is only defined on the set where
$t\ge-a_s/v_s^{1/2}$).
This implies that for any $s_0\in\R$,
$\lim_{(s,t)\to(s_0,L_-)}h(s,t)$ exists and is either positive and finite, or is
$+\infty$.

Thus, it suffices to show that for all $t_0\in[L_-,+\infty)$, the limits
$\lim_{(s,t)\to(-\infty,t_0)}h(s,t)$ and
$\lim_{(s,t)\to(+\infty,t_0)}h(s,t)$ exist, each of which is either positive and
finite, or is
$+\infty$.

Let us first consider the case when $s\searrow-\infty$.
For $s<0$, rewrite
$h(s,t)$ as
\[
h(s,t)=
\left(\int_0^{-s}\frac1{(t')^{1/2}}\cdot\exp(-t')\,dt'\right)^{-1}
(-s)v_s^{1/2}\frac1{(-sv_s^{1/2}t-sa_s)^{1/2}}
\cdot\exp(-(-sv_s^{1/2}t-sa_s)).
\]
From the Lebesgue dominated convergence theorem and \Cref{lem:tilted-limits},
the limit
\[
\lim_{(s,t)\to(-\infty,t_0)}h(s,t)
\]
exists and is positive, as desired.

Now let us consider the case when $s\nearrow+\infty$.
For $s>0$, rewrite
$h(s,t)$ as
\[
h(s,t)=
\left(\int_{-s}^0\left(\frac{t'}s+1\right)^{-1/2}\cdot\exp(t')\,dt'\right)^{-1}
sv_s^{1/2}(v_s^{1/2}t+a_s)^{-1/2}\cdot\exp(sv_s^{1/2}t+s(a_s-1)).
\]
Now \Cref{lem:tilted-limits} implies that $\lim_{s\nearrow+\infty}v_s^{1/2}=0$
and
$\lim_{s\nearrow+\infty}a_s=1$, and Equation~\eqref{eq:endpoint-integral-limit}
from the end of the proof of
\Cref{lem:tilted-limits} implies that
\[
\lim_{s\nearrow+\infty}\int_{-s}^0\left(\frac{t'}s+1\right)^{-1/2}\cdot\exp(t')\,dt'
\]
exists and is positive and finite.
We conclude that $\lim_{(s,t)\to(+\infty,t_0)}h(s,t)$ exists and is positive and
finite, as desired.
\end{proof}

\begin{corollary}\label{cor:tilted-common-constants}
There exist absolute constants $L_-,L_+,L_{\mathrm{vert}}$ satisfying
$L_-<0<L_+$ and
$L_{\mathrm{vert}}>0$, such that
$L_-,L_+$ satisfy the conclusion of \Cref{cor:tilted-support-interval} (namely,
$[L_-,L_+]\subseteq\operatorname{supp}\widetilde
X_s=[-a_s/v_s^{1/2},(1-a_s)/v_s^{1/2}]$
for all
$s\in\R$) and
$L_{\mathrm{vert}}$ satisfies the conclusion of \Cref{lem:tilted-density-floor}
(namely,
$f_{\widetilde X_s}(t)\ge L_{\mathrm{vert}}$ for all
$s\in\R$ and for all
$t\in(L_-,L_+)$, where
$f_{\widetilde X_s}$ denotes the density function of
$\widetilde X_s$).
\end{corollary}

\begin{proof}
\Cref{cor:tilted-support-interval} supplies absolute constants
$L_-<0<L_+$ with
$[L_-,L_+]\subseteq\operatorname{supp}\widetilde X_s$ for all
$s\in\R$.
\Cref{lem:tilted-density-floor} applied to this pair supplies
$L_{\mathrm{vert}}>0$ depending only on
$L_-,L_+$, hence absolute, with
$f_{\widetilde X_s}(t)\ge L_{\mathrm{vert}}$ for all
$s\in\R$ and
$t\in(L_-,L_+)$.
\end{proof}

\begin{lemma}\label{lem:tilted-density-bound-lower}
For any $\delta>0$, there exists
$C_{\mathrm{bound}-}>0$, depending only on
$\delta$, such that
\[
f_{\widetilde X_s}(t)\le C_{\mathrm{bound}-}\left(t-\frac{-a_s}{v_s^{1/2}}\right)^{-1/2}
\]
for all $s\in\R$ and for all
$t$ satisfying
\[
\frac{-a_s}{v_s^{1/2}}
<t
\le\frac{-a_s}{v_s^{1/2}}+\delta.
\]
\end{lemma}
\begin{proof}
As in the proof of \Cref{lem:tilted-density-floor}, define the function
\[
h(s,t):=\frac1{G_{s,0}}v_s^{1/2}\frac1{2(v_s^{1/2}t+a_s)^{1/2}}\cdot\exp(s(v_s^{1/2}t+a_s))
\]
with domain $\operatorname{dom}h=\{(s,t)\in\R\times\R:t>-a_s/v_s^{1/2}\}$, where
\[
G_{s,0}:=\int_0^1\frac1{2(t')^{1/2}}\cdot\exp(st')\,dt'.
\]
Note that $f_{\widetilde X_s}(t)=h(s,t)$ for all
$t\in\operatorname{supp}\widetilde X_s=[-a_s/v_s^{1/2},(1-a_s)/v_s^{1/2}]$, so
it suffices to give an upper bound on
$h(s,t)\cdot(t-(-a_s/v_s^{1/2}))^{1/2}$ for all
$(s,t)$ satisfying
\[
\frac{-a_s}{v_s^{1/2}}
<t
\le\frac{-a_s}{v_s^{1/2}}+\delta.
\]
Indeed, we have
\[
h(s,t)\cdot\left(t-\frac{-a_s}{v_s^{1/2}}\right)^{1/2}
=\frac1{2G_{s,0}}v_s^{1/4}\cdot\exp(s(v_s^{1/2}t+a_s)).
\]
Now, if $s\ge0$, we have
$G_{s,0}\ge G_{0,0}=1$ and
$\exp(s(v_s^{1/2}t+a_s))\le\exp(sv_s^{1/2}\cdot\delta)$.
Using \Cref{cor:tilted-variance-scale}, we conclude that
$h(s,t)\cdot(t-(-a_s/v_s^{1/2}))^{1/2}$
is bounded from above, as desired.

Now suppose $s<0$.
We split into two cases: $-1\le s<0$ and
$s<-1$.

In the first case, when $-1\le s<0$, we have
$G_{s,0}\ge G_{-1,0}$ and
$\exp(s(v_s^{1/2}t+a_s))\le1$.
In exactly the same manner as before, we conclude that
$h(s,t)\cdot(t-(-a_s/v_s^{1/2}))^{1/2}$ is bounded from above, as desired.

In the second case, when $s<-1$, we have
\[
G_{s,0}=\frac1{(-s)^{1/2}}\int_0^{-s}\frac1{2(t')^{1/2}}\cdot\exp(-t')\,dt'.
\]
Thus $(-s)^{1/2}\cdot G_{s,0}$ is bounded from below.
Additionally, we still have $\exp(s(v_s^{1/2}t+a_s))\le1$.
Using \Cref{cor:tilted-variance-scale}, we conclude that
$h(s,t)\cdot(t-(-a_s/v_s^{1/2}))^{1/2}$
is bounded from above, as desired.
\end{proof}

\begin{lemma}\label{lem:tilted-density-bound-upper}
There exists an absolute constant $C_{\mathrm{bound}+}>0$ such that
\[
f_{\widetilde X_s}\left(\frac{1-a_s}{v_s^{1/2}}\right)\le C_{\mathrm{bound}+}
\]
for all $s\in\R$.
\end{lemma}
\begin{proof}
We compute
\[
f_{\widetilde X_s}\left(\frac{1-a_s}{v_s^{1/2}}\right)
=\frac1{2G_{s,0}}v_s^{1/2}\cdot\exp(s)
\]
where $G_{s,0}:=\int_0^1\frac1{2t^{1/2}}\cdot\exp(st)\,dt$.
From \Cref{lem:tilted-continuity}, we see that $f_{\widetilde
X_s}((1-a_s)/v_s^{1/2})$ is a
continuous function of
$s\in\R$.
Thus it suffices to show that the two limits
\[
\lim_{s\searrow-\infty}f_{\widetilde X_s}\left(\frac{1-a_s}{v_s^{1/2}}\right)
\qquad
\lim_{s\nearrow+\infty}f_{\widetilde X_s}\left(\frac{1-a_s}{v_s^{1/2}}\right)
\]
exist and are finite.

Let us first consider the case when $s\searrow-\infty$.
For $s<0$, we have
\[
f_{\widetilde X_s}\left(\frac{1-a_s}{v_s^{1/2}}\right)
=\frac1{(-s)^{1/2}}\cdot\exp(s)\cdot
\left(\int_0^{-s}\frac1{t^{1/2}}\cdot\exp(-t)\,dt\right)^{-1}
\cdot(-s)v_s^{1/2}.
\]
Using the Lebesgue dominated convergence theorem and \Cref{lem:tilted-limits},
we conclude
that
\[
\lim_{s\searrow-\infty}f_{\widetilde X_s}\left(\frac{1-a_s}{v_s^{1/2}}\right)=0,
\]
as desired.

Now let us consider the case when $s\nearrow+\infty$.
For $s>0$, we have
\[
f_{\widetilde X_s}\left(\frac{1-a_s}{v_s^{1/2}}\right)
=
\left(\int_{-s}^0\left(\frac ts+1\right)^{-1/2}\cdot\exp(t)\,dt\right)^{-1}
\cdot sv_s^{1/2}.
\]
Using Equation~\eqref{eq:endpoint-integral-limit} from the end of the proof of
\Cref{lem:tilted-limits}, we conclude that
\[
\lim_{s\nearrow+\infty}f_{\widetilde X_s}\left(\frac{1-a_s}{v_s^{1/2}}\right)
\]
exists and is finite, as desired.
\end{proof}

\begin{lemma}\label{lem:tilted-moment-bounds}
There exist absolute constants $M_3,M_4>0$ such that for all
$s\in\R$, we have
\[
\E{X_s\sim g_s/G_{s,0}}{|\widetilde X_s|^3}\le M_3
\qquad
\text{and}
\qquad
\E{X_s\sim g_s/G_{s,0}}{\widetilde X_s^4}\le M_4.
\]
\end{lemma}
\noindent\textbf{Remark (not used below).}
The following sharper values are conjectural:
\begin{align*}
M_3
&=\lim_{s\searrow-\infty}\E{X_s\sim g_s/G_{s,0}}{|\widetilde X_s|^3}\\
&=\int_{-1/2^{1/2}}^{+\infty}\frac1{(2\pi)^{1/2}}
\left(\frac t{2^{1/2}}+\frac12\right)^{-1/2}
\exp\left(-\left(\frac t{2^{1/2}}+\frac12\right)\right)\cdot|t|^3\,dt\\
&=\frac{12}{(e\pi)^{1/2}}+22^{1/2}-42^{1/2}\,\operatorname{erf}\frac1{2^{1/2}}
\in(3.072,3.073)
\end{align*}
\[
M_4
=\lim_{s\searrow-\infty}\E{X_s\sim g_s/G_{s,0}}{\widetilde X_s^4}
=\int_{-1/2^{1/2}}^{+\infty}\frac1{(2\pi)^{1/2}} \left(\frac t{2^{1/2}}+\frac12\right)^{-1/2} \exp\left(-\left(\frac t{2^{1/2}}+\frac12\right)\right)\cdot t^4\,dt
=15.
\]
The proof below uses only the existence of finite absolute constants; in
particular, it gives the non-sharp bound $M_4<300$.

\begin{proof}[Proof of \Cref{lem:tilted-moment-bounds}]
From \Cref{lem:tilted-continuity}, we have that $a_s$ and
$v_s$ and
\[
G_{s,0}:=\int_0^1\frac1{2t^{1/2}}\cdot\exp(st)\,dt
\]
are continuous functions of $s\in\R$; furthermore,
$v_s$ and
$G_{s,0}$ are each positive for all
$s\in\R$.
Since
\[
\E{X_s\sim g_s/G_{s,0}}{|\widetilde X_s|^3}
=\frac1{G_{s,0}}\int_0^1\frac1{2t^{1/2}}\cdot\exp(st)\cdot\left|\frac{t-a_s}{v_s^{1/2}}\right|^3\,dt
\]
\[
\E{X_s\sim g_s/G_{s,0}}{\widetilde X_s^4}
=\frac1{G_{s,0}}\int_0^1\frac1{2t^{1/2}}\cdot\exp(st)\cdot\left(\frac{t-a_s}{v_s^{1/2}}\right)^4\,dt,
\]
it follows from Lebesgue dominated convergence theorem that both
$\E{X_s\sim g_s/G_{s,0}}{|\widetilde X_s|^3}$ and
$\E{X_s\sim g_s/G_{s,0}}{\widetilde X_s^4}$ are continuous functions of
$s\in\R$.
Thus, it suffices to show that the four limits
\[
\lim_{s\searrow-\infty}\E{X_s\sim g_s/G_{s,0}}{|\widetilde X_s|^3}
\qquad
\lim_{s\nearrow+\infty}\E{X_s\sim g_s/G_{s,0}}{|\widetilde X_s|^3}
\]
\[
\lim_{s\searrow-\infty}\E{X_s\sim g_s/G_{s,0}}{\widetilde X_s^4}
\qquad
\lim_{s\nearrow+\infty}\E{X_s\sim g_s/G_{s,0}}{\widetilde X_s^4}
\]
exist and are finite.

Let us first consider the case when $s\searrow-\infty$.
For fixed $s<0$, the random variable
$-sX_s$ has distribution supported on
$[0,-s]$ and density function proportional to
$\frac1{t^{1/2}}\cdot\exp(-t)$.
Now we have
\[
\widetilde X_s=\frac{-sX_s-\E{X_s\sim g_s/G_{s,0}}{-sX_s}}{\V{X_s\sim g_s/G_{s,0}}{-sX_s}^{1/2}}.
\]
Thus
\[
\E{X_s\sim g_s/G_{s,0}}{-sX_s}
=\left(\int_0^{-s}\frac1{t^{1/2}}\cdot\exp(-t)\cdot t\,dt\right)
\left(\int_0^{-s}\frac1{t^{1/2}}\cdot\exp(-t)\,dt\right)^{-1}
\]
\begin{align*}
\E{X_s\sim g_s/G_{s,0}}{|\widetilde X_s|^3}
&=\left(\int_0^{-s}\frac1{t^{1/2}}\cdot\exp(-t)\cdot(t-\E{X_s\sim g_s/G_{s,0}}{-sX_s})^2\,dt\right)^{-3/2}\\
&\quad\cdot\left(\int_0^{-s}\frac1{t^{1/2}}\cdot\exp(-t)\cdot|t-\E{X_s\sim g_s/G_{s,0}}{-sX_s}|^3\,dt\right)\\
&\quad\cdot\left(\int_0^{-s}\frac1{t^{1/2}}\cdot\exp(-t)\,dt\right)^{1/2}
\end{align*}
\begin{align*}
\E{X_s\sim g_s/G_{s,0}}{\widetilde X_s^4}
&=\left(\int_0^{-s}\frac1{t^{1/2}}\cdot\exp(-t)\cdot(t-\E{X_s\sim g_s/G_{s,0}}{-sX_s})^2\,dt\right)^{-2}\\
&\quad\cdot\left(\int_0^{-s}\frac1{t^{1/2}}\cdot\exp(-t)\cdot(t-\E{X_s\sim g_s/G_{s,0}}{-sX_s})^4\,dt\right)
\left(\int_0^{-s}\frac1{t^{1/2}}\cdot\exp(-t)\,dt\right)
\end{align*}
from which it follows that each of
$\lim_{s\searrow-\infty}\E{X_s\sim g_s/G_{s,0}}{-sX_s}$,
$\lim_{s\searrow-\infty}\E{X_s\sim g_s/G_{s,0}}{|\widetilde X_s|^3}$, and
$\lim_{s\searrow-\infty}\E{X_s\sim g_s/G_{s,0}}{\widetilde X_s^4}$ exists and is
finite.

Now let us consider the case when $s\nearrow+\infty$.
For fixed $s>0$, the random variable
$s(X_s-1)$ has distribution supported on
$[-s,0]$ and density function proportional to
$(\frac ts+1)^{-1/2}\cdot\exp(t)$.
We would like to say that the law of $s(X_s-1)$ ``converges'' in some sense to
the distribution
$\exp(t)\,dt$ (supported on
$(-\infty,0]$).

First, observe
\[
\widetilde X_s=\frac{s(X_s-1)-\E{X_s\sim g_s/G_{s,0}}{s(X_s-1)}}{\V{X_s\sim g_s/G_{s,0}}{s(X_s-1)}^{1/2}}.
\]
Thus
\[
\E{X_s\sim g_s/G_{s,0}}{s(X_s-1)}
=\left(\int_{-s}^0\left(\frac ts+1\right)^{-1/2}\cdot\exp(t)\cdot t\,dt\right)
\left(\int_{-s}^0\left(\frac ts+1\right)^{-1/2}\cdot\exp(t)\,dt\right)^{-1}
\]
\begin{align*}
\E{X_s\sim g_s/G_{s,0}}{|\widetilde X_s|^3}
&=\left(\int_{-s}^0\left(\frac ts+1\right)^{-1/2}\cdot\exp(t)\cdot(t-\E{X_s\sim g_s/G_{s,0}}{s(X_s-1)})^2\,dt\right)^{-3/2}\\
&\quad\cdot\left(\int_{-s}^0\left(\frac ts+1\right)^{-1/2}\cdot\exp(t)\cdot|t-\E{X_s\sim g_s/G_{s,0}}{s(X_s-1)}|^3\,dt\right)\\
&\quad\cdot\left(\int_{-s}^0\left(\frac ts+1\right)^{-1/2}\cdot\exp(t)\,dt\right)^{1/2}
\end{align*}
\begin{align*}
\E{X_s\sim g_s/G_{s,0}}{\widetilde X_s^4}
&=\left(\int_{-s}^0\left(\frac ts+1\right)^{-1/2}\cdot\exp(t)\cdot(t-\E{X_s\sim g_s/G_{s,0}}{s(X_s-1)})^2\,dt\right)^{-2}\\
&\quad\cdot\left(\int_{-s}^0\left(\frac ts+1\right)^{-1/2}\cdot\exp(t)\cdot(t-\E{X_s\sim g_s/G_{s,0}}{s(X_s-1)})^4\,dt\right)\\
&\quad\cdot\left(\int_{-s}^0\left(\frac ts+1\right)^{-1/2}\cdot\exp(t)\,dt\right).
\end{align*}
However, this time we cannot apply Lebesgue dominated convergence theorem to
conclude that $\lim_{s\nearrow+\infty}\E{X_s\sim g_s/G_{s,0}}{s(X_s-1)}$,
$\lim_{s\nearrow+\infty}\E{X_s\sim g_s/G_{s,0}}{|\widetilde X_s|^3}$, and
$\lim_{s\nearrow+\infty}\E{X_s\sim g_s/G_{s,0}}{\widetilde X_s^4}$ exist, due to
the singularity of
$(\frac ts+1)^{-1/2}$ near
$t=-s$.
Instead, we split the integral into two parts, by writing
$\int_{-s}^0=\int_{-s}^{-s/2}+\int_{-s/2}^0$.
By arguing in a similar manner as in the proof of \Cref{lem:tilted-limits}, we
can show
that $\lim_{s\nearrow+\infty}\E{X_s\sim g_s/G_{s,0}}{s(X_s-1)}$,
$\lim_{s\nearrow+\infty}\E{X_s\sim g_s/G_{s,0}}{|\widetilde X_s|^3}$, and
$\lim_{s\nearrow+\infty}\E{X_s\sim g_s/G_{s,0}}{\widetilde X_s^4}$ exist and are
finite, as desired.
\end{proof}

\begin{lemma}\label{lem:tilted-mean-monotone}
The map $s\mapsto sa_s$ is a strictly increasing function of
$s$ when
$s<0$, with
\[
\lim_{s\searrow-\infty}(-s)a_s=\frac12.
\]
\end{lemma}
\begin{proof}
For $u>0$, let
\[
m(u):=
\frac{\int_0^u t^{1/2}e^{-t}\,dt}{\int_0^u t^{-1/2}e^{-t}\,dt}.
\]
The random variable $uX_{-u}$ has density proportional to
$t^{-1/2}e^{-t}$ on
$[0,u]$, and hence
\[
m(u)
=\E{X_{-u}\sim g_{-u}/G_{-u,0}}{uX_{-u}}
=u a_{-u}.
\]
If $F(u):=\int_0^u t^{-1/2}e^{-t}\,dt$, direct differentiation gives
\[
m'(u)
=\frac{u^{-1/2}e^{-u}}{F(u)}\bigl(u-m(u)\bigr)
>0,
\]
because $m(u)<u$.
Thus $m$ is strictly increasing in
$u$, so
$sa_s=-m(-s)$ is strictly increasing in
$s<0$.

Finally, dominated convergence gives
\begin{align*}
\lim_{s\searrow-\infty}(-s)a_s
&=\lim_{u\to\infty}m(u)\\
&=\left(\int_0^{+\infty}t^{1/2}e^{-t}\,dt\right)
\left(\int_0^{+\infty}t^{-1/2}e^{-t}\,dt\right)^{-1}\\
&=\Gamma(3/2)\,\Gamma(1/2)^{-1}=\frac12.
\end{align*}
\end{proof}

\endgroup

\section{Reflected dynamics for membership sampling}\label{sec:kinetic}
We prove \Cref{thm:sampler} by simulating a particle that moves inside
$K$ and reflects at its boundary. We first define this process, including
its random position and velocity updates, and prove its mixing bound.
We then give a membership-query implementation and bound its error and
expected cost. The argument initially assumes a smooth boundary; the
proof concludes by removing this assumption. The final part of the
appendix extends the construction to logconcave distributions.

Fix an error tolerance $0<\varepsilon<1/2$ and a warmness bound $M\ge1$.
There are two estimates to establish. If $P$ bounds the Poincar\'e
constant of $\pi=\Unif(K)$, the process mixes from an $M$-warm start
in time $O(\sqrt P\log(2M/\varepsilon))$. Its simulation for this
time $T$ uses $\widetilde O(Mn^2T)$ expected membership queries.
Combining these estimates with \eqref{eq:kls} proves the theorem for
constant $M$. Throughout the appendix, the initial sample is supplied
to the algorithm, and only oracle queries are counted.

\subsection{Reflected dynamics and mixing}
We define the exact process, establish its invariant distribution, and
prove the relaxation bound used to choose its running time.

\label{subsec:kinetic-definition}

For now, let $K$ be a convex body with $C^\infty$ boundary satisfying
$B_2^n\subseteq K\subseteq RB_2^n$. Put $\pi=\Unif(K)$,
$\gamma_n=N(0,I_n)$, and $\widehat\pi=\pi\otimes\gamma_n$.
The state space is $K\times\mathbb R^n$, with position $x$ and velocity $v$.
For $x\in\partial K$, let $\nu(x)$ be the outward unit normal and define
the reflection $\mathcal R_xv=v-2\langle v,\nu(x)\rangle\nu(x)$.
The billiard flow $\Psi_t(x,v)=(x_t,v_t)$ follows
$\dot x_t=v_t$ and $\dot v_t=0$ in the interior, with
$v_t=\mathcal R_{x_t}v_{t-}$ at each collision.
The flow construction in the proof of
\Cref{prop:linear-potential-relaxation}, with zero force, gives trajectories
defined for all real times outside a null set and
$\Psi_{s+t}=\Psi_s\circ\Psi_t$ almost everywhere.
The flow preserves Lebesgue measure and velocity length, hence
$\widehat\pi$.

We first specify the random updates as Markov operators. For a kernel $T$,
our conventions are $Tf(z)=\int f(z')T(z,dz')$ and
$(\eta T)(A)=\int T(z,A)\,d\eta(z)$. Operator products act on functions
from right to left, so $T_1T_2$ describes first applying $T_1$ and then
$T_2$ to the state.
Let $\varphi_h(w)=(2\pi h)^{-n/2}e^{-\|w\|_2^2/(2h)}$ for $h>0$.
For bounded measurable $f:K\to\mathbb R$, define
\begin{equation}\label{eq:position-update-definition}
 (R_hf)(x)=\int_{\mathbb R^n}\varphi_h(y-x)
 \frac{\displaystyle\int_K f(x')\varphi_h(x'-y)\,dx'}
      {\displaystyle\int_K\varphi_h(u-y)\,du}\,dy.
\end{equation}
The denominator is positive for every $y$. Equivalently, draw
$Y\sim N(x,hI_n)$ and then $X'\sim N(Y,hI_n)$ conditioned on $X'\in K$.
Its transition probabilities are $R_h(x,A)=(R_h\mathbf1_A)(x)$ for
Borel sets $A\subseteq K$.
This is the In-and-Out update of~\cite{KVZ24}.
For $a\in[0,1)$, define the velocity kernel
$Q_a(v,dv')=\varphi_{1-a^2}(v'-av)\,dv'$.
Thus its update is $v'=av+\sqrt{1-a^2}G$, with $G\sim\gamma_n$.

The joint kernel is $J_h((x,v),dx'\,dv')=R_h(x,dx')Q_a(v,dv')$.
Extend $Q_a$ and $R_h$ to phase space by keeping the other coordinate fixed.
Then $J_h=Q_aR_h=R_hQ_a$. The position and velocity updates are
independent conditional on $(x,v)$.
The velocity update preserves $\gamma_n$. The position update is a Gibbs
update for the joint law $X\sim\pi$, $Y\mid X\sim N(X,hI_n)$, so
$\pi R_h=\pi$. Consequently $\widehat\pi J_h=\widehat\pi$.
These kernels therefore act as contractions on the corresponding
$L^2$ spaces.

We define the transport generator through the flow.
The operators $U_tf=f\circ\Psi_t$, $t\in\mathbb R$,
form a strongly continuous unitary group on $L^2(\widehat\pi)$.
Strong continuity follows first for bounded continuous functions from
the almost-everywhere continuity of trajectories at time zero, and then
for all $L^2$ functions by density. Define
\begin{equation}\label{eq:billiard-generator-domain}
 \begin{aligned}
 \operatorname{dom}(B)
 &=\left\{f\in L^2(\widehat\pi):
     \lim_{t\to0}\frac{U_tf-f}{t}
     \text{ exists in }L^2(\widehat\pi)\right\},\\
 Bf&=\lim_{t\to0}\frac{U_tf-f}{t}.
 \end{aligned}
\end{equation}
The unitary-group property makes $B$ skew-adjoint.
For smooth $f$ with compact support in velocity and satisfying
the reflection condition $f(x,v)=f(x,\mathcal R_xv)$ on
$\partial K\times\mathbb R^n$, we have $f\in\operatorname{dom}(B)$ and
$Bf=v\cdot\nabla_xf$.

Fix a rate $\lambda>0$. The exact process follows $\Psi_t$ between the
events of an independent Poisson clock of rate $\lambda$, and applies
$J_h$ at each event using fresh randomness. Equivalently, its Markov
semigroup is
\begin{equation}\label{eq:poisson-semigroup}
 \mathcal P_t=e^{-\lambda t}\sum_{k=0}^{\infty}\lambda^k
 \int_{0<s_1<\cdots<s_k<t}
 U_{s_1}J_hU_{s_2-s_1}\cdots J_hU_{t-s_k}\,ds_1\cdots ds_k.
\end{equation}
The $k=0$ term is $U_t$, and the integrals are taken strongly on
$L^2(\widehat\pi)$. The $k$th summand has operator norm at most
$e^{-\lambda t}(\lambda t)^k/k!$, so the series converges in operator norm.
For $\widehat\pi$-almost every $(x,v)$, it gives
$\mathcal P_tf(x,v)=\mathbb E_{x,v}f(X_t,V_t)$.
The bounded perturbation theorem yields
$\mathcal G=B+S$ on $\operatorname{dom}(\mathcal G)=\operatorname{dom}(B)$,
where $S=\lambda(J_h-I)$. The semigroup preserves this domain.
Write $\gamma=\lambda(1-a)$ for the velocity refreshment parameter.

Each factor in \eqref{eq:poisson-semigroup} preserves $\widehat\pi$.
Thus an initial law $\rho\le M\widehat\pi$ satisfies
$\rho\mathcal P_t\le M\widehat\pi$ for every $t\ge0$; the same domination
holds for the entire path law relative to the stationary path law.
There are finitely many Poisson events on a bounded time interval, and
invariance ensures that the updates avoid the exceptional set of the
billiard flow almost surely under such an initial law.
For sampling, initialize with $X_0\sim\nu_0\le M\pi$ and an independent
$V_0\sim\gamma_n$. At the prescribed time $T$, apply one further $J_h$
update, so the terminal transition is $\mathcal P_TJ_h$, and return its
position coordinate.

\label{subsec:kinetic-mixing}
We prove a relaxation time of order $\sqrt P$ when $P$ bounds the
Poincar\'e constant of $\pi$. The position update has two useful
properties: its contribution to the energy estimate is small, and its
output law varies continuously in total variation with its input.

The next lemma records these properties. In its statement, $H^1(K)$
is the space of functions with square-integrable weak first derivatives.
A self-adjoint operator is called positive if its quadratic form is
nonnegative.
\begin{lemma}[Properties of the position update]
\label{lem:position-update-properties}
The operator $R_h$ is a positive self-adjoint Markov contraction on
$L^2(\pi)$. For every $f\in H^1(K)$,
\begin{equation}\label{eq:position-update-dirichlet}
 \langle f,(I-R_h)f\rangle_{L^2(\pi)}
 \le h\int_K\|\nabla f\|_2^2\,d\pi.
\end{equation}
For $x,x'\in K$, its transition laws satisfy
$\|R_h(x,\cdot)-R_h(x',\cdot)\|_{\TV}\le
\min\{1,\|x-x'\|_2/(2\sqrt h)\}$. Consequently, for $v,v'\in\mathbb R^n$,
\begin{equation}\label{eq:mixed-update-smoothing}
 \|J_h((x,v),\cdot)-J_h((x',v'),\cdot)\|_{\TV}
 \le \min\left\{1,
 \frac{\|x-x'\|_2}{2\sqrt h}
 +\frac{a\|v-v'\|_2}{2\sqrt{1-a^2}}\right\}.
\end{equation}
\end{lemma}
\begin{proof}
Under the joint law $X\sim\pi$, $Y\mid X\sim N(X,hI_n)$, the
conditional law of $X$ given $Y=y$ is $N(y,hI_n)$ restricted to $K$.
Let $\mathcal T:L^2(\pi)\to L^2(\mathcal L(Y))$ be defined by
$(\mathcal T f)(y)=\mathbb E[f(X)\mid Y=y]$. Its adjoint is
$(\mathcal T^*g)(x)=\mathbb E[g(Y)\mid X=x]$.
Thus $R_h=\mathcal T^*\mathcal T$, proving positivity and self-adjointness.
The Markov property and invariance give contraction. Also,
$\langle f,(I-R_h)f\rangle=\mathbb E\Var(f(X)\mid Y)$.
Each conditional law is $h^{-1}$-strongly logconcave on a convex set.
Equation~\eqref{eq:strong-logconcave-poincare} gives
$C_P(\mathcal L(X\mid Y=y))\le h$. Applying this inequality
conditionally and integrating proves
\eqref{eq:position-update-dirichlet}.

The second conditional draw is the same kernel for every initial point.
Total variation contraction and the Gaussian relative-entropy formula
therefore give
\[
 \|R_h(x,\cdot)-R_h(x',\cdot)\|_{\TV}
 \le\|N(x,hI_n)-N(x',hI_n)\|_{\TV}
 \le\frac{\|x-x'\|_2}{2\sqrt h}.
\]
The corresponding bound for the velocity update has mean displacement
$a(v-v')$ and variance $1-a^2$. Adding the two errors proves
\eqref{eq:mixed-update-smoothing}.
\end{proof}

We now adapt the argument of Eberle and L\"orler to the process just
defined. Their space-time divergence lemma supplies the auxiliary
functions in the proof. We estimate the contribution of our position
update and partial velocity refreshment explicitly.
\begin{proposition}[Relaxation of the exact process]
\label{prop:mixed-refresh-relaxation}
Suppose $C_P(\pi)\le P$. Choose $\lambda\ge P^{-1/2}$,
$a=1-P^{-1/2}/\lambda$, and $0<h\le\sqrt P/\lambda$.
There are universal constants $c,C>0$ such that, for every $t\ge0$
and $f\in L^2(\widehat\pi)$ with $\int f\,d\widehat\pi=0$,
\[
 \|\mathcal P_t f\|_{L^2(\widehat\pi)}
 \le C e^{-ct/\sqrt P}\|f\|_{L^2(\widehat\pi)}.
\]
\end{proposition}
\begin{proof}
For $f\in L^2(\widehat\pi)$, define
$(\Pi f)(x,v)=\int f(x,w)\,d\gamma_n(w)$.
This is the orthogonal projection onto functions independent of velocity;
we apply it pointwise in time to time-dependent functions.
Define the nonnegative quadratic form
$\mathcal D(f)=\langle f,-Sf\rangle$. The commuting positive
contractions $Q_a$ and $R_h$ satisfy
$I-Q_aR_h\succeq I-Q_a$ and
$I-Q_aR_h\preceq(I-Q_a)+(I-R_h)$.
On a Gaussian Hermite polynomial of degree $k\ge1$, $Q_a$ acts by
$a^k\le a$. Hence, with $\gamma=\lambda(1-a)$,
\begin{equation}\label{eq:mixed-microscopic-dissipation}
 \|f-\Pi f\|_2^2\le\gamma^{-1}\mathcal D(f).
\end{equation}

Fix an interval length $\tau>0$ and initially take $f$ in the domain
of $B$. For $0\le s\le\tau$, put $u_s=\mathcal P_{t+s}f$ and
$\chi=\Pi u$. In the following space-time estimates, all norms and inner products include integration
over $s\in[0,\tau]$, with position measure $\pi$ and velocity measure
$\gamma_n$ when the function depends on velocity. The function $\chi$
has mean zero. The divergence lemma~\cite[Theorem 5]{EL26:journal} gives
functions $p(s,x)$ and $g(s,x)$, zero at both time endpoints, with
$g$ satisfying the Neumann boundary condition and
$\chi=\partial_s p-\Delta g$. They satisfy
\[
 \|p\|_2^2+\|\nabla_x g\|_2^2\le C(\tau^2+P)\|\chi\|_2^2,
\]
and
\[
 \|\nabla_{s,x}p\|_2^2+
 \|\nabla_{s,x}\nabla_xg\|_2^2
 \le C(1+P/\tau^2)\|\chi\|_2^2.
\]
Here $\nabla_{s,x}=(\partial_s,\nabla_x)$. The hypotheses of that
lemma hold because $K$ is smooth, compact, and convex, and the potential
is constant. In particular, the Neumann Laplacian has discrete spectrum.

On $L^2([0,\tau];L^2(\widehat\pi))$, let $\mathcal A$ be the closure
of $-\partial_s+B$ initially defined on
$H^1([0,\tau];L^2(\widehat\pi))\cap
L^2([0,\tau];\operatorname{dom}(B))$, with no time boundary condition.
Set $\phi=p+v\cdot\nabla_xg$.
The Neumann condition makes $\nabla_xg$ tangent to the boundary, so
$\phi$ satisfies the reflection condition. Smooth spatial approximation
and radial velocity cutoffs give
$\phi\in L^2([0,\tau];\operatorname{dom}(B))$, with
$B\phi=v\cdot\nabla_xp+v^{\mathsf T}\nabla_x^2g\,v$:
the stated derivative bounds and Gaussian moments give convergence in
the graph norm of $B$.
The time estimates also give $\phi\in H^1([0,\tau];L^2(\widehat\pi))$
with zero traces at both endpoints. Hence $\phi\in\operatorname{dom}(\mathcal A^*)$
and $\mathcal A^*\phi=\partial_s\phi-B\phi$. Velocity averaging gives
$\Pi\mathcal A^*\phi=\partial_sp-\Delta g=\chi$. Therefore
$\|\chi\|_2^2=\langle\Pi u,\mathcal A^*\phi\rangle$.
Integration by parts and the estimate from
\cite[Lemma 14 and proof of Theorem 15]{EL26:journal} give
\begin{equation}\label{eq:mixed-divergence-identity}
 \begin{aligned}
 \|\chi\|_2^2
 &=\langle\mathcal Au,\phi\rangle
   +\langle\Pi u-u,\mathcal A^*\phi\rangle,\\
 \|\mathcal A^*\phi\|_2^2
 &\le C(1+P/\tau^2)\|\chi\|_2^2.
 \end{aligned}
\end{equation}
For the velocity part, the quadratic form of $\lambda(I-Q_a)$ at
$\phi$ is exactly $\gamma\|\nabla_xg\|_2^2$. For the position part,
\eqref{eq:position-update-dirichlet} and Gaussian integration give a
bound $\lambda h(\|\nabla_xp\|_2^2+\|\nabla_x^2g\|_2^2)$.
It follows that
\begin{equation}\label{eq:mixed-test-dissipation}
 \int_0^\tau\mathcal D(\phi_s)\,ds
 \le C\bigl(\gamma(\tau^2+P)+\lambda h(1+P/\tau^2)\bigr)
       \|\chi\|_2^2.
\end{equation}

Since $\mathcal Au=-Su$, apply Cauchy--Schwarz for the quadratic
form $-S$ to the first term in
\eqref{eq:mixed-divergence-identity}. Apply
\eqref{eq:mixed-microscopic-dissipation} to its second term.
Adding the bound for $u-\Pi u$ yields
\begin{equation}\label{eq:mixed-average-dissipation}
 \int_t^{t+\tau}\|\mathcal P_s f\|_2^2\,ds
 \le C_\tau\int_t^{t+\tau}\mathcal D(\mathcal P_s f)\,ds,
\end{equation}
where
\[
 C_\tau=C\left(\gamma(\tau^2+P)+
       (\lambda h+\gamma^{-1})(1+P/\tau^2)\right).
\]
The norm in \eqref{eq:mixed-average-dissipation} is taken over phase
space at each fixed time.
Skew-adjointness of $B$ gives the energy identity
$\frac d{ds}\|\mathcal P_s f\|_2^2=-2\mathcal D(\mathcal P_s f)$.
For $F(t)=\int_t^{t+\tau}\|\mathcal P_s f\|_2^2\,ds$, we obtain
$F'(t)\le-2F(t)/C_\tau$. Consequently,
\[
 \tau\|\mathcal P_{t+\tau}f\|_2^2
 \le F(t)\le e^{-2t/C_\tau}F(0)
 \le \tau e^{-2t/C_\tau}\|f\|_2^2.
\]
Set $\tau=\sqrt P$ and $\gamma=P^{-1/2}$. The assumption on $h$
implies $C_\tau\le C\sqrt P$, proving the assertion for functions in
the generator domain. Density and contraction extend it to all mean-zero
$L^2$ functions.
\end{proof}

The same estimate holds for the adjoint semigroup, which evolves
densities. Since $\nu_0\otimes\gamma_n\le M\widehat\pi$,
$\|d\nu_0/d\pi-1\|_{L^2(\pi)}\le\sqrt{M-1}$. Thus
\[
 \|\mathcal L(X_t,V_t)-\widehat\pi\|_{\TV}
 \le \tfrac C2 e^{-ct/\sqrt P}\sqrt{M-1}.
\]
Time $T=C\sqrt P\log(2M/\varepsilon)$ therefore gives any prescribed
constant fraction of the required error $\varepsilon$. The terminal
update does not increase that error. By \eqref{eq:kls}, we may take
$P=C\kappa\log(2n)$ under the covariance assumption of
\Cref{thm:sampler}.

\subsection{Membership-query simulation}
We describe the oracle algorithm and prove the estimates needed to
simulate its reflected paths. The argument controls reflection counts,
approximates boundary normals, and then bounds the error and cost of
one interval between random updates.

\subsubsection{The membership sampling algorithm}
\label{subsec:kinetic-implementation}
We describe the finite algorithm and then estimate the error of its
position update. The following subsections justify the boundary and
normal approximations and bound the total work.

Sample the Poisson times in $[0,T]$, append the endpoint $T$, and
simulate the billiard motion on each resulting interval. Within an
interval, first query the endpoint of the straight trajectory. If it is
in $K$, convexity certifies that the entire straight segment is in $K$.
Otherwise, bisect the line to approximate its first boundary intersection,
estimate a normal there, and reflect the velocity in that approximate
normal. Move the computed intersection a small distance towards the
origin to obtain an interior point, and continue for the remaining time.

The normal is computed from the gauge
$g_K(y)=\inf\{s>0:y\in sK\}$, abbreviated below as $g$.
A query whether $g_K(y)\le s$ is
the membership query $y/s\in K$, so bisection approximates gauge values.
At differentiability points on the boundary, the gradient of $g_K$ is
parallel to the outward normal. We use finite differences at a randomly
chosen point on the same ray. The later normal-estimation lemmas specify
this randomization and prove its error bounds.

Each simulated interval has a reflection budget. Start with budget one
and, if necessary, restart the interval with twice the budget and greater
precision. The same auxiliary random choices are used on restarts.
Prescribed bounds on the clock count, velocity, and reflection budget
ensure termination on every execution; on exceeding a bound, return the
fixed point $0\in K$. All precisions and bounds are specified in the
proof of \Cref{thm:constant-warm-membership-sampling}.

At each interval endpoint, apply $Q_a$ exactly and replace $R_h$ by
the following kernel, denoted $\widehat R_h$: propose
$Z=x+\sqrt{2h}G$, with $G\sim\gamma_n$, query whether $Z\in K$,
and return $Z$ if accepted and $x$ otherwise. This uses one membership
query. It also applies at the terminal time. Exact Gaussian draws and
arithmetic between oracle calls are not charged in the query model.
For a Borel set $A\subseteq K$, its kernel is
\begin{equation}\label{eq:one-query-kernel}
 \widehat R_h(x,A)=\int_A\varphi_{2h}(z-x)\,dz
 +\mathbf1_A(x)\int_{K^c}\varphi_{2h}(z-x)\,dz.
\end{equation}
Its extension to phase space leaves the velocity fixed.

The one-query update approximates $R_h$ after averaging over a warm
initial law. We prove this estimate explicitly because a uniform
pointwise approximation near the boundary is not required.
\begin{lemma}[One-query position update]
\label{lem:one-query-small-step}
Suppose $rB_2^n\subseteq K$, and write
$A=\mathcal H^{n-1}(\partial K)$ and $V=\vol(K)$, where
$\mathcal H^{n-1}$ is surface measure. There is a coupling of
$R_h(x,\cdot)$ and $\widehat R_h(x,\cdot)$ with disagreement
probability at most $e_h(x)=\mathbb P\{x+\sqrt{2h}G\notin K\}$.
Moreover,
\[
 \int_K e_h\,d\pi\le\frac AV\sqrt{\frac h\pi}
 \le\frac nr\sqrt{\frac h\pi}.
\]
For $\nu_0\le M\pi$, we have
$\int_K e_h\,d\nu_0\le Mn\sqrt{h/\pi}/r$.
The kernel $\widehat R_h$ is positive and self-adjoint on $L^2(\pi)$,
and $\pi\widehat R_h=\pi$.
\end{lemma}
\begin{proof}
Draw the intermediate point $Y$ of $R_h$, followed by a proposal
$Z_1\sim N(Y,hI_n)$. To implement the exact kernel, repeat the second
draw until it belongs to $K$. To implement $\widehat R_h$, accept the
first proposal or return $x$. The first proposal has law $N(x,2hI_n)$,
and the coupled outputs agree when it is in $K$.

Fix a unit vector $u$. Let $P_{u^\perp}$ be orthogonal projection onto
$u^\perp$, and let $\ell_u(y)$ be the length of the chord
$K\cap(y+\mathbb Ru)$ for $y\in P_{u^\perp}K$. For a translation
$qu$, $q\ge0$, integration over these chords gives
\[
 \vol\{x\in K:x+qu\notin K\}
 =\int_{P_{u^\perp}K}\min\{q,\ell_u(y)\}\,dy
 \le q\,\vol_{n-1}(P_{u^\perp}K).
\]
For a standard Gaussian $G$, its direction $U=G/\|G\|_2$ is uniform on
the unit sphere and independent of $\|G\|_2$. Cauchy's projection formula
and $\mathbb E\|G\|_2\mathbb E|U_1|=\mathbb E|G_1|=\sqrt{2/\pi}$ imply
\[
 \mathbb E\left[\sqrt{2h}\|G\|_2\,
               \vol_{n-1}(P_{U^\perp}K)\right]
 =\sqrt{2h}\,\mathbb E\|G\|_2\frac A2\mathbb E|U_1|
 =A\sqrt{h/\pi}.
\]
Division by $V$ proves the first bound. On $\partial K$,
$\langle x,\nu(x)\rangle\ge r$, so the divergence theorem gives
$rA\le nV$. Domination by $M$ proves the warm-start bound.

The off-diagonal transition density in \eqref{eq:one-query-kernel}
is $\mathbf1_K(z)\varphi_{2h}(z-x)$, symmetric in
$x,z\in K$; the remaining mass stays at $x$. This proves reversibility.
To see positivity, extend a test function by zero outside $K$.
Convolution with $\varphi_{2h}$ has a nonnegative quadratic form because
its Fourier transform is nonnegative. The holding term also has a
nonnegative quadratic form.
\end{proof}

Repeated position updates can be compared by coupling until the first
disagreement. We will use the same argument with an additional endpoint
error when the billiard segments have been approximated numerically.
\begin{corollary}[Error over several position updates]
\label{cor:several-small-updates}
Suppose an exact process preserves $\widehat\pi$ and applies $R_h$
at an independent Poisson clock of rate $\lambda$. From an $M$-warm
initial law, replacing these updates before time $T$ by
$\widehat R_h$ gives a coupling whose probability of any disagreement
is at most $M\lambda Tn\sqrt{h/\pi}/r$.
Including one further update at time $T$ replaces $\lambda T$ by
$\lambda T+1$ in this bound. There is one membership query per update.
\end{corollary}
\begin{proof}
Under stationarity, the expected sum of the error bounds at the Poisson
times is $\lambda T\int e_h\,d\pi$, by the compensator formula for
the independent clock. Domination of the path law multiplies this by
at most $M$. Couple each update as in
\Cref{lem:one-query-small-step} until the first disagreement, and sum
the error probabilities. The deterministic final update has the same
average error bound.
\end{proof}

\subsubsection{Reflection counts and boundary estimates}
\label{subsec:kinetic-geometry}
We now estimate the geometric events that affect the simulation. We
begin with the stationary rate of reflections. Short-flight estimates
will then control the second moment of the number of reflections within
each simulation interval.

The incoming velocity at a reflection is distributed according to the
outward flux through the boundary. Integrating this flux gives the
following count and also bounds the probability of exceptional contacts.
\begin{lemma}[Expected number of reflections]
\label{lem:stationary-reflection-count}
Let $N_T$ count reflections up to time $T$ for the exact process.
If $rB_2^n\subseteq K$, then
\[
 \mathbb E_{\widehat\pi}N_T
 =\frac{T\,\mathcal H^{n-1}(\partial K)}
        {\sqrt{2\pi}\vol(K)}
 \le\frac{nT}{r\sqrt{2\pi}}.
\]
For an $M$-warm initial law,
$\mathbb E N_T\le MnT/(r\sqrt{2\pi})$.
\end{lemma}
\begin{proof}
At a boundary element $dS(x)$, the stationary rate of collisions with
incoming velocity in $dv$ is
$\vol(K)^{-1}\langle v,\nu(x)\rangle_+\gamma_n(dv)\,dS(x)$,
where $q_+=\max\{q,0\}$. This follows by measuring the volume swept
out in a short time interval. Integration over velocities contributes
$\int\langle v,\nu(x)\rangle_+\,d\gamma_n(v)=1/\sqrt{2\pi}$.
Integrating over the boundary and time gives the identity. The inequality
uses $r\mathcal H^{n-1}(\partial K)\le n\vol(K)$, proved above.
Finally apply domination by $M$ to the nonnegative path functional
$N_T$.
\end{proof}

For the remaining geometric estimates, assume $n\ge2$ and
$rB_2^n\subseteq K$. Write $V_K=\operatorname{vol}(K)$ and
$S_K=\mathcal H^{n-1}(\partial K)$. The preceding flux formula gives
collision rate $\beta_K=S_K/(\sqrt{2\pi}V_K)\le n/(r\sqrt{2\pi})$.

For a fixed direction, convexity forces most chords to have a definite
length. We will integrate this elementary estimate against Gaussian
velocity to bound the frequency of short flights.
\begin{lemma}[Short chords]
\label{lem:short-chords}
Let $u\in S^{n-1}$, let $D_u=P_{u^\perp}K$, and let $\ell_u(y)$
be the length of $K\cap(y+\mathbb Ru)$ for $y\in D_u$. For $a>0$,
\[
 \frac{|\{y\in D_u:\ell_u(y)<a\}|}{|D_u|}
 \le\min\left\{1,\frac{(n-1)a}{2r}\right\}.
\]
\end{lemma}
\begin{proof}
Minkowski combinations of the one-dimensional fibers show that
$\ell_u$ is nonnegative and concave. Also $\ell_u(0)\ge2r$. If
$0<a<2r$ and $\delta=a/(2r)$, then
$\ell_u((1-\delta)y)\ge\delta\ell_u(0)\ge a$ for every $y\in D_u$.
The short chords therefore lie outside $(1-\delta)D_u$, a set whose
relative volume is at most $1-(1-\delta)^{n-1}\le(n-1)\delta$.
For $a\ge2r$, the claimed bound is the trivial bound one.
\end{proof}

We allow the invariant position and velocity updates defined above.
Between updates, the motion is the usual specular billiard. At a contact,
the outgoing straight chord is defined by continuing with the reflected
velocity and ignoring any subsequent update.
\begin{proposition}[Short free flights]
\label{prop:short-free-flights}
Suppose $\mathcal L(X_0,V_0)=\widehat\pi$. Let $B_{T,s}$ count contacts during
$[0,T]$ whose outgoing straight chord has traversal time less than $s$.
Then $\mathbb E B_{T,s}\le Cn^{5/2}Ts/r^2$. Consequently, the probability
of two reflections separated by less than $s$, with no intervening update,
is at most $Cn^{5/2}Ts/r^2$.
\end{proposition}
\begin{proof}
Fix a velocity $v=qu$, where $q>0$ and $u\in S^{n-1}$. In the
coordinates of its transverse projection, outgoing boundary flux is
$q\,dy/V_K$. By \Cref{lem:short-chords}, the rate of contacts marked
by a straight flight of duration less than $s$ is at most
\[
 \frac q{V_K}|\{y\in D_u:\ell_u(y)<qs\}|
 \le\frac{(n-1)sq^2}{2rV_K}|D_u|.
\]
For a standard Gaussian velocity, speed and direction are independent,
$\mathbb E q^2=n$, and Cauchy's projection formula gives
$\mathbb E_u|D_u|=(S_K/2)\mathbb E|U_1|\le S_K/(2\sqrt n)$.
The resulting rate is at most
$(n-1)\sqrt n\,sS_K/(4rV_K)\le Cn^{5/2}s/r^2$.
Integration over time proves the expectation bound. Invariant updates
preserve the law used in this flux calculation. Every actual short flight
without an update is one of the marked straight chords, so Markov's
inequality gives the probability bound.
\end{proof}

The next estimate excludes nearly tangent contacts in the numerical
analysis. Its quadratic dependence on the normal-speed threshold follows
from the Gaussian boundary flux.
\begin{lemma}[Small normal velocity]
\label{lem:grazing-flux}
For $a\ge0$, let $G_{T,a}$ count incoming contacts of the stationary process
$(x,v)$ before time $T$ with $\langle v,\nu(x)\rangle\le a$. Then
\[
 \mathbb E G_{T,a}
 =\frac{TS_K}{\sqrt{2\pi}V_K}(1-e^{-a^2/2})
 \le\frac{nTa^2}{2r\sqrt{2\pi}}.
\]
\end{lemma}
\begin{proof}
At each boundary point, integrate the incoming Gaussian flux over normal
speeds in $[0,a]$. The integral is
$\int_0^a s e^{-s^2/2}\,ds/\sqrt{2\pi}
=(1-e^{-a^2/2})/\sqrt{2\pi}$.
Multiply by $TS_K/V_K$ and use $S_K/V_K\le n/r$.
\end{proof}

The simulation cost will be quadratic in the number of reflections in
one interval. We therefore need a truncated second moment, rather than
only the expected collision count. Let $N_s$ count reflections in a
stationary billiard trajectory of duration $s$, with no update inside
that interval.
\begin{lemma}[Truncated collision count]
\label{lem:truncated-collision-count}
For every integer $k\ge2$ and every integer $b\ge1$,
\[
 \begin{aligned}
 \mathbb P(N_s\ge k)&\le\frac{Cn^{5/2}s^2}{r^2k^2},\\
 \mathbb E\min\{N_s,b\}^2
 &\le\beta_Ks+\frac{Cn^{5/2}s^2}{r^2}\log(2b).
 \end{aligned}
\]
\end{lemma}
\begin{proof}
If $N_s\ge k$, at least half of the first $k-1$ flights between
contacts have duration at most $2s/(k-1)$. Their expected number is at
most $Cn^{5/2}s^2/(r^2(k-1))$, by \Cref{prop:short-free-flights}.
Markov's inequality proves the first bound, after adjusting the constant.
For the second, use $\mathbb P(N_s\ge1)\le\mathbb E N_s=\beta_Ks$
and the identity
\[
 \mathbb E\min\{N_s,b\}^2
 =\sum_{k=1}^b(2k-1)\mathbb P(N_s\ge k).
\]
The remaining terms form a harmonic sum.
\end{proof}

An independent Poisson partition makes the sum of squared interval lengths
small. Combining this fact with the preceding lemma gives the total
reflection work. Here and below, an $M$-warm phase law means a measure
bounded by $M\widehat\pi$.
\begin{lemma}[Lengths of the simulation intervals]
\label{lem:poisson-interval-lengths}
Partition $[0,T]$ by an independent rate-$\lambda$ Poisson process,
and denote the interval lengths by $\ell_1,\ldots,\ell_{J+1}$. Then
$\sum_i\ell_i=T$ and $\mathbb E\sum_i\ell_i^2\le2T/\lambda$.
If $\mathcal L(X_0,V_0)\le M\widehat\pi$ and an invariant update is applied
at each partition point, the reflection counts $N_i$ satisfy, for $b\ge1$,
\begin{align*}
 \mathbb P\{\max_iN_i>b\}
 &\le\frac{CMn^{5/2}T}{r^2\lambda b^2},\\
 \mathbb E\sum_i\min\{N_i,b\}^2
 &\le CM\left(\frac{nT}{r}
       +\frac{n^{5/2}T}{r^2\lambda}\log(2b)\right).
\end{align*}
\end{lemma}
\begin{proof}
Two times $s,t\in[0,T]$ lie in the same partition interval precisely
when no Poisson point lies between them. Therefore
\[
 \mathbb E\sum_i\ell_i^2
 =\int_0^T\int_0^T e^{-\lambda|s-t|}\,ds\,dt\le2T/\lambda.
\]
Condition on the entire partition. Each prescribed billiard flow and
each update preserves $\widehat\pi$, so every interval starts from an
$M$-warm law. Apply \Cref{lem:truncated-collision-count}, multiply its
stationary bounds by $M$, and sum over the intervals. Averaging over the
partition proves the claims.
\end{proof}

The final geometric estimate keeps the endpoints of numerical intervals
away from the boundary. This is needed to decide correctly whether an
interval ends before the next reflection.
\begin{lemma}[Boundary distance at update times]
\label{lem:update-boundary-margin}
For $0<d\le r$ and $X\sim\operatorname{Unif}(K)$,
$\mathbb P\{\operatorname{dist}(X,\partial K)<d\}\le nd/r$.
For the process in \Cref{lem:poisson-interval-lengths}, the probability
that any position immediately before or after an update, or at time
zero or $T$, has boundary distance less than $d$ is at most
$CM(1+\lambda T)nd/r$.
\end{lemma}
\begin{proof}
Convexity gives $(1-d/r)K+dB_2^n\subseteq K$. Thus the points at boundary
distance at least $d$ have volume at least $(1-d/r)^nV_K$, proving the
first claim. For stationary dynamics, apply the Poisson compensation
formula before each update, use invariance for the post-update point,
and include the deterministic endpoints. Domination of the entire path
law by $M$ times its stationary law proves the second claim.
\end{proof}

\subsubsection{Approximating the boundary normal}

We now assume $B_2^n\subseteq K\subseteq RB_2^n$. Membership queries
determine the gauge $g(x)=\inf\{s>0:x\in sK\}$ by binary search. The
gauge is convex, positively homogeneous, and $1$-Lipschitz. At a regular
boundary point $x$, it satisfies
$\nabla g(sx)=\nabla g(x)$ for $s>0$ and
$\nu(x)=\nabla g(x)/\|\nabla g(x)\|_2$, with $\|\nabla g(x)\|_2\ge1/R$.
The last inequality follows from $\langle x,\nabla g(x)\rangle=1$.

We first bound the mean error of a normal estimate. For the simulation,
we then bound the normal error in terms of the contact-point displacement
and the numerical precision.

At a stationary contact the surface location is uniform, by the Gaussian
flux formula. A short random radial displacement preserves its normal
and gives a point with an absolutely continuous law, allowing us to
analyze the finite differences.
\begin{lemma}[Estimating a boundary normal]
\label{lem:average-normal-oracle}
Let $X$ be uniform in surface area on $\partial K$, and let
$\widetilde X$ be any jointly distributed point with
$\|\widetilde X-X\|_2\le\zeta\le1$. For $0<u\le1/n$ and $0<\xi\le1$,
a randomized membership procedure returns $\widehat\nu\in S^{n-1}$ such that
\[
 \mathbb E\|\widehat\nu-\nu(X)\|_2
 \le CRn\left(u+\frac{\xi+\zeta}{u}\right).
\]
It uses $Cn\log(2+R/\xi)$ queries. In particular, for $0<\varepsilon<1$,
the guarantee $\mathbb E\|\widehat\nu-\nu(X)\|_2\le\varepsilon$ costs
$O(n\log(2+nR/\varepsilon))$ queries
if $\zeta\le c\varepsilon^2/(R^2n^2)$.
\end{lemma}
\begin{proof}
Choose $S$ uniformly on $[1,1+1/n]$, independently of $(X,\widetilde X)$,
and put $Y=SX$, $\widetilde Y=S\widetilde X$. The radial change of
variables has Jacobian $s^{n-1}\langle x,\nu(x)\rangle$. Since
$\langle x,\nu(x)\rangle\ge1$, the density of $Y$ is at most
$n/S_K$ on $D=(1+1/n)K\setminus K$.

For the moment use exact gauge values and define
$G_i(y)=[g(y+ue_i)-g(y-ue_i)]/(2u)$. For a smooth convex $1$-Lipschitz
function $f$, integrate the fundamental theorem of calculus in each
coordinate. Positivity of its diagonal second derivatives gives
\[
 \int_D\sum_{i=1}^n
 \left|\frac{f(y+ue_i)-f(y-ue_i)}{2u}-\partial_i f(y)\right|dy
 \le u\int_{(1+1/n+u)K}\Delta f.
\]
Every translated integration domain lies in the displayed dilate, because
$uB_2^n\subseteq uK$. By the divergence theorem and $\|\nabla f\|_2\le1$,
the last expression is at most
$u(1+1/n+u)^{n-1}S_K\le e^2uS_K$.
Smooth convolutions of $g$ preserve convexity and the Lipschitz bound.
They converge uniformly, and their gradients converge almost everywhere
to $\nabla g$. Fatou's lemma therefore gives
$\mathbb E\|G(Y)-\nabla g(Y)\|_1\le Cnu$.

Evaluate the $2n$ gauge values at $\widetilde Y\pm ue_i$ to error at
most $\xi$, and call the resulting vector $\widehat G$. Lipschitzness
gives $\|\widehat G-G(Y)\|_1\le n(\xi+2\zeta)/u$. Normalize
$\widehat G$, choosing a fixed unit vector if it is zero. The returned
unit vector satisfies
$\|\widehat\nu-a/\|a\|_2\|_2\le2\|\widehat G-a\|_2/\|a\|_2$ for every nonzero $a$. Take
$a=\nabla g(Y)$ and use $\|a\|_2\ge1/R$ to obtain the claimed error.

For $z\ne0$, the predicate $g(z)\le t$ is equivalent to $z/t\in K$
when $t>0$. Binary search on $[0,\|z\|_2]$ therefore finds $g(z)$ to
accuracy $\xi$; at $z=0$ return zero directly. Since
$\|\widetilde Y\pm ue_i\|_2\le5R$, each value costs
$O(\log(2+R/\xi))$ queries. Finally choose $u=c_1\varepsilon/(Rn)$ and
$\xi=c_2\varepsilon u/(Rn)$, with sufficiently small constants.
\end{proof}

An average error at a single point does not by itself control repeated
reflections. The next lemma bounds every nearby subgradient on an event
whose failure probability is small. Only finitely many scales are needed.
\begin{lemma}[Local variation of the gauge gradient]
\label{lem:finite-scale-gauge-regularity}
With $X,S,Y$ as above, fix $u_j=2^{-j}/n$ for $0\le j\le J$ and
$L>0$. Except on an event of probability at most $Cn(J+1)/L$,
\begin{equation}
 \sum_{i=1}^n
 \bigl(g(Y+u_je_i)+g(Y-u_je_i)-2g(Y)\bigr)
 \le Lu_j^2\quad\text{for }0\le j\le J.
 \label{eq:finite-scale-gauge-event}
\end{equation}
On this event, $a=\nabla g(Y)=\nabla g(X)$ almost surely. For $\|z-Y\|_1\le1/(2n)$ and every $p\in\partial g(z)$,
\begin{equation}
 \|p-a\|_2\le8nL\|z-Y\|_2+2\sqrt nLu_J.
 \label{eq:finite-scale-gauge-modulus}
\end{equation}
\end{lemma}
\begin{proof}
The integral formula for second differences, applied first to a smooth
convex $1$-Lipschitz function $f$, gives
\[
 \int_D\sum_i\bigl(f(y+ue_i)+f(y-ue_i)-2f(y)\bigr)\,dy
 \le u^2\int_{(1+1/n+u)K}\Delta f\le e^2u^2S_K.
\]
Convolution and Fatou's lemma transfer this inequality to $g$.
The density bound from \Cref{lem:average-normal-oracle} makes the expected
second-difference sum at most $Cnu^2$. Markov's inequality at each scale
and a union bound prove \eqref{eq:finite-scale-gauge-event}.

Fix such a point $Y$ and set
$r(w)=g(Y+w)-g(Y)-\langle a,w\rangle\ge0$. At scale $u=u_j$,
$r(\pm ue_i)\le Lu^2$. Convex interpolation on the hull of these
points gives $r(w)\le Lu^2$ whenever $\|w\|_1\le u$.
If $\|z-Y\|_1\le u/2$, apply the subgradient inequality at $z$ to
$z\pm(u/2)e_i$. It yields $|p_i-a_i|\le2Lu$, hence
$\|p-a\|_2\le2\sqrt nLu$. Choose the smallest available scale at least
$2\|z-Y\|_1$. Unless that scale is $u_J$, it is less than
$4\|z-Y\|_1\le4\sqrt n\|z-Y\|_2$. This proves the two terms in
\eqref{eq:finite-scale-gauge-modulus}. Homogeneity identifies the
gradients at $X$ and $Y$.
\end{proof}

We record the consequence for an implemented normal. Suppose
$\widetilde Y=S\widetilde X$ and
$\|\widetilde Y-Y\|_1+q\le1/(2n)$. Use central differences of step $q$
and gauge errors at most $\xi$. Integrating the subgradient bound along
these coordinate intervals gives
\[
 \|\widehat G-a\|_2
 \le Cn^{3/2}L\bigl(\|\widetilde Y-Y\|_2+q+u_J\bigr)
       +\sqrt n\xi/q.
\]
Normalizing the gradient multiplies this error by at most $2R$.
Thus, if $q\ge u_J$ and
$\xi\le q^2$, the normal error is at most a polynomial in $n,R,L$
times $\|\widetilde X-X\|_2+q$. This estimate applies to any approximate
point satisfying the distance bound, even if it depends on earlier
membership answers.

\subsubsection{Simulating a billiard segment}

We next give the deterministic part of the algorithm. A billiard interval
has straight motion and specular reflections, without an update inside it.
The numerical procedure restarts with a doubled reflection budget when
necessary. Its precision depends on this budget, so intervals with few
reflections do not pay for the largest permitted count.

The assumptions in the next lemma describe the exact interval. The final
proof will bound their failure probability using the preceding geometric
estimates. Computation between membership queries is free in the oracle
model.
\begin{lemma}[Simulating a billiard trajectory]
\label{lem:finite-query-billiard-segment}
Fix $\epsilon>0$ and assume $B_2^n\subseteq K\subseteq RB_2^n$. Consider an exact billiard
interval of duration $\ell\le H$, starting at $(x,v)$ with
$x\in\operatorname{int}K$ and $\|v\|_2\le B$, where $B,H\ge1$. Suppose it has
$N\le N_{\max}$ reflections, where $N_{\max}\ge1$, and satisfies:
\begin{enumerate}
\item Every incoming normal speed is at least $\alpha>0$.
\item Every flight between consecutive reflections lasts at least $\tau$.
\item Its terminal position has boundary distance at least $d$.
\item Independent radial variables $S_j\sim\operatorname{Unif}[1,1+1/n]$
at its true contacts satisfy \eqref{eq:finite-scale-gauge-event}, with
constant $L$ and scales through $u_J$.
\end{enumerate}
Assume $\alpha,\tau,d\le1$. For a sufficiently large absolute $C$, put
\[
 D=\left(\frac{CnR(1+L)(1+B)(1+H)}{\alpha\tau d}\right)^{10}
\]
and $\epsilon_0=\min\{\epsilon,D^{-2}\}$. If $u_J\le\epsilon_0/(nD^{2N_{\max}+10})$, a membership algorithm
returns the terminal position and velocity with error at most $\epsilon$
in each. It uses one query when $N=0$ and otherwise
$O(1+nN^2\log(2DR/\epsilon_0))$ queries. Its restarts reuse the same
$S_j$ for contact number $j$. On every input it terminates after a budget
of at most $2N_{\max}$ reflections, possibly declaring failure.
If the geometric conditions hold but the interval has more than
$N_{\max}$ reflections, all trials up to that cap still use at most
$O(1+nN_{\max}^2\log(2DR/\epsilon_0))$ queries in total.
\end{lemma}

\begin{proof}
\emph{The numerical procedure.}
For a trial budget $b=1,2,4,\ldots$, choose a dyadic scale $q_b=2^{-k}/n$
such that
\(\frac{\epsilon_0}{2nD^{b+5}}<q_b\le\frac{\epsilon_0}{nD^{b+5}}.\)
The required trials have $b\le2N_{\max}$, so $q_b\ge u_J$.
Each trial starts from the original $(x,v)$ and duration $\ell$.

At a current interior point $\widetilde p$, with velocity $\widetilde v$
and remaining time $t$, query $\widetilde p+t\widetilde v$. If it is
inside, return that point and velocity: convexity puts the whole remaining
straight interval inside. If it is outside and the trial has already used
$b$ reflections, restart with twice the budget, or declare failure at the
final budget. Otherwise bisect the line until its inside and outside
endpoints are at distance at most $q_b$. Call the inside endpoint $p_-$
and its travel time $t_-$. Estimate the normal by central differences at
$S_jp_-$, with step $q_b$ and gauge accuracy $q_b^2$. If the estimated
gradient is zero, declare failure; otherwise normalize it and reflect
$\widetilde v$. Replace the position by $(1-q_b)p_-$ and subtract
$t_-$ from the remaining time.

The new point is interior because
$(1-q_b)p_-+q_bB_2^n\subseteq K$. The computed reflection preserves
speed. If $v=0$, the first endpoint query succeeds and no division by
speed or normal estimation occurs. These rules, including the budget
cap and the check before normalization, define a finite procedure on
all realizations.

\emph{Location of the first exit.}
At a true contact $p$, integrate \eqref{eq:finite-scale-gauge-modulus}
from $Sp$ to $S(p+w)$ and use homogeneity. For
$\|w\|_2\le1/(4n^{3/2})$, this gives
\begin{equation}
 0\le g(p+w)-1-\langle\nabla g(p),w\rangle
 \le CnL\|w\|_2^2+C\sqrt nLu_J\|w\|_2.
 \label{eq:gauge-contact-quadratic}
\end{equation}
Suppose a true flight starts at $p_0$ with velocity $v_0$ and first
reaches the boundary at $p=p_0+t_0v_0$. Suppose the numerical starting
position, velocity, and elapsed time each have error at most $e$.
At local time $t_0$, the spatial error is at most $(1+H)e$. The incoming
velocity satisfies
\[
 \langle\nabla g(p),v_0\rangle
 =\frac{\langle\nu(p),v_0\rangle}{\langle p,\nu(p)\rangle}
 \ge\alpha/R.
\]
Set $s=CR(1+H)(e+q_b)/\alpha$. At $t_0+s$, the supporting-plane
inequality for $g$ makes the numerical ray lie outside $K$. At $t_0-s$,
\eqref{eq:gauge-contact-quadratic} makes it lie inside. Forsooth the negative
linear term has size at least $\alpha s/R$, whereas the initial-point
error is $O((1+H)e)$ and the remaining terms are bounded by
\(CnL\bigl(Bs+(1+H)e\bigr)^2 +C\sqrt nLu_J\bigl(Bs+(1+H)e\bigr).\)
The choice of $D$ and $e+q_b\le D^{-2}$ make these terms smaller than
the linear term and keep the comparison in the stated neighborhood.
They also give $s<t_0/2$ on every noninitial flight, since $t_0\ge\tau$.

This inside--outside bracket identifies the first numerical exit.
The numerical ray starts inside and is inside at $t_0-s$; its intersection
with a convex body is an interval, so it cannot have exited earlier.
For the initial flight the starting data are exact, and bisection needs
no lower bound on its duration.

\emph{Propagation of the errors.}
If a reflection occurs, $\|v\|_2\ge\alpha$. The line-search error therefore
contributes at most $q_b/\alpha$ to the contact-time error and at most
$q_b$ to the position error. Radial contraction contributes at most
$Rq_b$. The local normal estimate above is linear in contact-position
error, with additional error polynomial in $n,R,L$ times $q_b$.
The reflection formula now bounds each new position, velocity, and
elapsed-time error by $D(e+q_b)$. In particular the estimated gradient
cannot vanish under the stated assumptions.

Starting with zero error, after $j\le b$ contacts all three errors are
at most $2D^jq_b\le\epsilon_0/D^3$. Thus every smallness condition in
the preceding argument holds throughout the trial. The final unreflected
position has error less than both $\epsilon$ and $d/2$, so its inside
test succeeds. Since the true terminal point has boundary distance $d$
and speed at most $B$, at least $d/B$ time remains after its last
reflection. The accumulated time error and all contact brackets are
smaller than this margin. A trial with $b<N$ therefore cannot finish
early, and a trial with $b\ge N$ finishes after exactly $N$ reflections.

\emph{Number of queries.}
Each line bisection costs $O(\log(2+HB/q_b))$ queries. Each normal
estimate uses $2n$ gauge evaluations and costs
$O(n\log(2+R/q_b^2))$ queries. A trial of budget $b$ costs
$O(1+nb^2\log(2DR/\epsilon_0))$. Summing over doubling budgets through
the first successful one, which is at most $2N$, proves the bound.
When $N=0$, only the first endpoint query is made.

For the last assertion, stop the stability induction at the imposed
reflection budget. It uses only the contacts reached before that budget
and, to certify the next outside endpoint, the next true contact.
The same first-exit bracket applies regardless of how many later contacts
the exact interval contains. Thus every capped trial has the stated cost
also when $N>N_{\max}$. This bounds the work incurred when the
reflection cap is exceeded.
\end{proof}

\subsection{Proof of the membership-sampling theorem}

We combine the mixing estimate with the preceding implementation. A small
position update after each billiard interval makes nearby numerical and
exact endpoints close in total variation. This permits a new comparison
after every update; the required precision need only cover one interval.

The theorem records the warmness dependence explicitly. In its application
to rounding, the supplied warmness is an absolute constant.
\begin{theorem}[Sampling from a warm start]
\label{thm:constant-warm-membership-sampling}
Let $n\ge1$, $B_2^n\subseteq K\subseteq RB_2^n$, and
$\operatorname{Cov}(\operatorname{Unif}(K))\preceq\kappa I$, where
$1\le\kappa\le R^2$. Given a membership oracle and an initial sample
from $\nu\le M\operatorname{Unif}(K)$, the algorithm returns $\widehat X$
satisfying $\|\mathcal L(\widehat X)-\operatorname{Unif}(K)\|_{\TV}
\le\varepsilon$ for $\varepsilon\in(0,1/2)$, using
\[
 O\!\left(Mn^2\sqrt\kappa\,
          \operatorname{polylog}\frac{nRM}{\varepsilon}\right)
\]
expected membership queries. The supplied initial sample is not included
in this count.
\end{theorem}

\begin{proof}
For $n=1$, bisect the two endpoints of the interval $K$ to additive
accuracy $\varepsilon/4$ and sample uniformly from the certified inner
interval. Since $K$ has length at least two, the total variation error is
at most $\varepsilon$, with $O(\log(2+R/\varepsilon))$ queries.
Assume henceforth that $n\ge2$, and first suppose that $\partial K$
is smooth.

\emph{Time, precision, and finite budgets.}
Take the Poincar\'e bound $P=C\kappa\log(2n)\ge1$. Set
$\lambda=n^{3/2}$, $\gamma=P^{-1/2}$, $a=1-\gamma/\lambda$, and
$T=C\sqrt P\log(2M/\varepsilon)$.
Initialize velocity independently from $\gamma_n$. The exact process
applies $J_h=Q_aR_h$ at the Poisson times and once more at time $T$.
By \Cref{prop:mixed-refresh-relaxation}, its terminal law has error at
most $\varepsilon/4$ whenever $h\le\sqrt P/\lambda$ and the constant
in $T$ is large enough. The last invariant update cannot increase this
error.

Only logarithms of the numerical precision enter the query count. We
choose fixed powers of a common parameter to make all numerical failures
smaller than the reciprocal of the eventual work cap. Put $W=C_*+n+R+M+T+\varepsilon^{-1}$, with $C_*$ sufficiently
large, and choose $\eta=W^{-100}$, $h=W^{-220}$, and
$\delta=W^{-230}$. Let $b_*$ be the smallest power of two
above $C\sqrt{MnT/\varepsilon}+2$. By
\Cref{lem:poisson-interval-lengths},
\begin{equation}
 \mathbb P\{\max_iN_i>b_*\}\le\varepsilon/4.
 \label{eq:sampler-collision-overflow}
\end{equation}
Sample the whole Poisson partition before doing membership work. If its
point count exceeds $s_*=\lceil C(\lambda T+100\log W)\rceil$, return
zero. This event has probability at most $\eta$. We have
$b_*\le W^3$ and $s_*+1\le W^4$ after increasing $C_*$.

In the local simulation use $B=W$, $H=\max\{1,T\}$,
$\alpha=W^{-60}$, $\tau=d=W^{-120}$, and $L=W^{200}$.
Form $D$ from \Cref{lem:finite-query-billiard-segment}, put
$\delta_0=\min\{\delta,D^{-2}\}$, and take scales through $u_{J_*}$
such that $u_{J_*}\le\delta_0/(nD^{2b_*+10})$.
Since $\log D=O(\log W)$, we may take $J_*\le Cb_*\log W$.
These choices are fixed before the path is simulated.

\emph{Probability of the geometric exceptions.}
Conditioned on the partition, every exact interval starts from an
$M$-warm phase law. A Gaussian tail and the bound $s_*+1\le W^4$ show
that a speed above $B$ has probability at most $C\eta$. The next three
exceptions have the following bounds:
\[
 \begin{aligned}
 \mathbb P\{\text{some normal speed is below }\alpha\}
     &\le CMnT\alpha^2,\\
 \mathbb P\{\text{some uninterrupted flight is shorter than }\tau\}
     &\le CMn^{5/2}T\tau,\\
 \mathbb P\{\text{some interval endpoint is within }d\text{ of }\partial K\}
     &\le M(s_*+1)nd.
 \end{aligned}
\]
These follow from \Cref{lem:grazing-flux,prop:short-free-flights,lem:update-boundary-margin}.
Finally, the stationary collision location is uniform in surface area.
The finite-scale normal lemma and the collision flux bound give
$CMn^2T(J_*+1)/L$ for the probability of any failed normal event.
All four probability bounds are $O(\eta)$. This calculation uses warm
domination of the full path law; it does not require independent contacts.

\emph{Comparison of the endpoint updates.}
On each interval run \Cref{lem:finite-query-billiard-segment} with requested
phase error $\delta$ and budgets $1,2,4,\ldots,b_*$. Return zero if the
last budget is exhausted or the speed check fails. Replace the exact
position update by the one-query update $\widehat R_h$ defined above.
Velocity updates remain exact.

Let $e_h(x)=\mathbb P\{x+\sqrt{2h}G\notin K\}$, with
$G\sim\gamma_n$. Gaussian translation gives
$|e_h(x)-e_h(x')|\le C\|x-x'\|_2/\sqrt h$. Suppose the exact and computed
terminal phase points are $(x,v)$ and $(x',v')$, respectively, with
errors at most $\delta$. The one-query coupling and the smoothing estimate
\eqref{eq:mixed-update-smoothing} imply
\begin{equation}
 \begin{split}
 &\|Q_a\widehat R_h((x',v'),\cdot)-Q_aR_h((x,v),\cdot)\|_{\rm TV}\\
 &\le e_h(x)+C\delta/\sqrt h
                         +C\delta/\sqrt{1-a^2}.
 \end{split}
 \label{eq:implemented-update-comparison}
\end{equation}
By \Cref{lem:one-query-small-step}, the average first term under an
$M$-warm law is at most $Mn\sqrt{h/\pi}$. Also
$1-a^2\ge\gamma/\lambda$. Summing over at most $s_*+1$ updates makes
the total comparison error $O(\eta)$ with the chosen $h$ and $\delta$.

To justify this sum, couple the processes while their phase points agree
at the start of an interval. Use the same interval duration and independent
radial marks, simulate the interval, and maximally couple the two update
laws. Conditional on the marks, the ideal marginal is still exactly
$Q_aR_h$; hence its path law and warmness remain those used above.
The local simulation lemma and \eqref{eq:implemented-update-comparison}
bound the probability of a first disagreement by $C\eta$, apart from
the collision cap and discarded partitions. Adding these events and the
ideal mixing error proves the required total variation bound.

\emph{Expected number of queries.}
On the geometric event, an interval with no reflection uses one endpoint
query, and an interval with $N_i\ge1$ uses at most
$Cn\min\{N_i,b_*\}^2\log W+O(1)$ queries. The same bound applies if
$N_i>b_*$ and the algorithm stops: the local proof controls every trial
up to its cap. Each approximate position update uses one additional query.
Before any geometric exception or failed coupling, the expected work is
therefore at most
\[
 \begin{aligned}
 C\left(1+\lambda T+n\log W\,
          \mathbb E\sum_i\min\{N_i,b_*\}^2\right)
 &\le CMn^2T\log W\log(2b_*).
 \end{aligned}
\]
Here \Cref{lem:poisson-interval-lengths} and $\lambda=n^{3/2}$ were used.
The cost of reaching the reflection cap is already included in this
truncated second moment.

On every realization, the finite budgets give a deterministic work bound
$Q_*\le Cn(s_*+1)b_*^2\log W\le CW^{12}$. Speeds and line-search
lengths are bounded, every gauge value has prescribed precision, and
undefined normalization causes termination. A geometric exception or
failed update coupling contributes at most $C\eta Q_*\le1$ to expected
work. Discarded partitions use no membership queries. Substituting the
value of $T$ proves the stated bound for smooth bodies.

\emph{General convex bodies.}
Choose smooth convex outer approximations $K_j\supseteq K$ converging
to $K$ in Hausdorff distance, with $K_j\subseteq(R+1)B_2^n$.
Their volumes and covariances converge to those of $K$, and their uniform
laws converge in total variation. For all sufficiently large $j$,
$\nu\le2M\operatorname{Unif}(K_j)$ and
$\operatorname{Cov}(\operatorname{Unif}(K_j))\preceq2\kappa I$.
Use the common conservative parameters
$2R+1,2M,2\kappa$ for all algorithms, including their finite budgets.

Fix their common randomness. Every fixed membership answer for $K$ is
eventually the same for $K_j$: points in $K$, including boundary points,
belong to every outer approximation, whereas a point outside the closed
set $K$ is eventually outside $K_j$. The algorithm for $K$ has a finite
adaptive transcript. Induction along that transcript shows that all its
queries, answers, and output eventually coincide with those for $K_j$.
Thus the output laws converge in total variation under this coupling.
The common deterministic work cap also gives convergence of expected
query counts by dominated convergence. Passing to the limit proves the
theorem for $K$.
\end{proof}

For fixed $M$, the preceding theorem proves \Cref{thm:sampler}.
Replace its supplied covariance bound $\kappa$ by $\min\{\kappa,R^2\}$,
since $K\subseteq RB_2^n$ gives
$\operatorname{Cov}(\operatorname{Unif}(K))\preceq R^2I$.

\subsection{Logconcave sampling}

We extend the construction to a logconcave density given by an evaluation
oracle. We first represent the density as the marginal of a measure with
a linear potential on a convex set. We truncate this set, define the
reflected dynamics, and prove the estimates needed to simulate them. The
last part gives a sampling theorem for an entire sequence of observations,
as required by the covariance estimates in the main text.

\subsubsection{The epigraph representation and its truncation}

We represent $\mu$ as the position marginal of a distribution with a
linear potential, then truncate this distribution while keeping a known
contained ball.

Let $\mu(dx)\propto e^{-V(x)}dx$ be full-dimensional on $\mathbb R^n$,
where $V$ is proper, closed, and convex. Assume that $x_0$, $R$, $\kappa$,
and $c_*\ge1$ are known and satisfy
\[
 \begin{aligned}
 B(x_0,1)&\subseteq\{V\le\inf V+c_*n\},\\
 \operatorname{Cov}(\mu)&\preceq\kappa I_n,\\
 \mathbb E_\mu\|X-x_0\|_2^2&\le R^2.
 \end{aligned}
\]
Translate $x_0$ to zero and subtract the known value $V(x_0)$ from the
potential. Then $V(0)=0$, $-c_*n\le a_V:=\inf V\le0$, and $V\le c_*n$
on the unit ball. The value of $a_V$ is used only in the analysis.

We use the epigraph representation from Kook and Vempala~\cite{KV25},
with a rescaling of the added coordinate. Put $d=n+1$, write $z=(x,s)$,
and define the epigraph $E=\{(x,s):s\ge V(x)/\sqrt n\}$.
Set $b=\sqrt n\,e_d$ and $z_0=(0,c_*\sqrt n+1)$, where $e_d$ is
the last coordinate vector. Let $\widetilde\mu$ have density
proportional to $e^{-\langle b,z\rangle}$ on $E$. Its first $n$ coordinates
have law $\mu$: given $X\sim\mu$, the conditional height is
$S=(V(X)+F)/\sqrt n$, where $F$ is an independent exponential random
variable of rate one.

The variance of the potential bounds the covariance in the additional
coordinate. We include a short proof of the estimate we use.
\begin{lemma}[Variance of the potential]\label{lem:varentropy}
If $\mu(dx)\propto e^{-V(x)}dx$ is a logconcave probability measure on
$\mathbb R^n$, then $\operatorname{Var}_\mu(V)\le n$.
\end{lemma}
\begin{proof}
Let $Z(\beta)=\int e^{-\beta V(x)}dx$ for $\beta>0$. The function
$(y,\beta)\mapsto\beta V(y/\beta)$ is convex, so the marginal theorem
for logconcave functions shows that
$\beta^nZ(\beta)=\int e^{-\beta V(y/\beta)}dy$ is logconcave in $\beta$.
Twice differentiating its logarithm gives
$-n/\beta^2+\operatorname{Var}_{\mu_\beta}(V)\le0$, where
$\mu_\beta(dx)\propto e^{-\beta V(x)}dx$. Take $\beta=1$.
Exponential integrability for $\beta$ near one justifies differentiation.
\end{proof}

The next lemma gives explicit truncation parameters and bounds that are
available to the algorithm. The spatial cutoff enters the query complexity
only logarithmically, so a second-moment tail bound suffices.
\begin{lemma}[Parameters of the truncated epigraph]
\label{lem:truncated-epigraph-parameters}
Fix $0<\delta<1/2$. Choose $R_x=\max\{2,R\sqrt{2/\delta}\}$ and
\[
 H\ge\max\left\{c_*\sqrt n+2,
       \frac{2(n+1)\log2+2\log(2/\delta)}{\sqrt n}\right\},
\]
and let $C=E\cap\{\|x\|_2\le R_x\}\cap\{s\le H\}$ and
$\pi=\widetilde\mu(\cdot\mid C)$. Set
$R_C=R_x+H+(2c_*+1)\sqrt n+1$. Then
$\widetilde\mu(C^c)\le\delta$ and
$B(z_0,1)\subseteq C\subseteq B(z_0,R_C)$.
Moreover, $\operatorname{Cov}(\pi)\preceq
\max\{2\kappa,4\}I_d/(1-\delta)$. The known numbers
$A_0=2c_*n+\sqrt n$ and $L_b=2c_*n+\sqrt n+\sqrt nH$
satisfy $-\langle b,z-z_0\rangle\le A_0$ and
$|\langle b,z-z_0\rangle|\le L_b$ for every $z\in C$.
Membership in $C$ requires at most one evaluation of $V$.
\end{lemma}
\begin{proof}
Convexity and dilation about a minimizer of $V$ give
$\mathbb E_\mu e^{\theta(V-a_V)}\le(1-\theta)^{-n}$ for $0<\theta<1$.
Indeed, compare the integrals of $e^{-(1-\theta)(V-a_V)}$ and
$e^{-(V-a_V)}$ by that dilation. Approximate minimizers give the same
inequality if needed. Since $a_V\le0$ and $F$ is independent of $X$,
\[
 \mathbb P(V(X)+F\ge q)
 \le\mathbb P(V(X)-a_V+F\ge q)\le2^{n+1}e^{-q/2}.
\]
Thus $\mathbb P(S>H)\le\delta/2$. Markov's inequality gives
$\mathbb P(\|X\|_2>R_x)\le\delta/2$, proving the truncation bound.

Every point of $B(z_0,1)$ has $\|x\|_2\le1$ and
$s\ge c_*\sqrt n\ge V(x)/\sqrt n$, and also satisfies both truncation
constraints. Moreover, $-c_*\sqrt n\le s\le H$ throughout $C$.
These observations prove the ball inclusions and the bounds $A_0,L_b$.

By \Cref{lem:varentropy}, $\operatorname{Var}(S)\le(n+1)/n\le2$.
Consequently, for $u\in\mathbb R^n$ and $v\in\mathbb R$,
$\operatorname{Var}(\langle u,X\rangle+vS)\le2\kappa\|u\|_2^2+4v^2$.
If $m$ is the mean of the untruncated joint law, then for every $w$,
\[
 \operatorname{Var}_\pi\langle w,Z\rangle
 \le\mathbb E_\pi\langle w,Z-m\rangle^2
 \le\frac{\mathbb E_{\widetilde\mu}\langle w,Z-m\rangle^2}{1-\delta}.
\]
This proves the covariance bound. Finally, membership in $C$ is decided
by evaluating $V(x)$ and checking the explicit spatial and height bounds.
\end{proof}

A supplied initial law $\rho\le M\mu$ can be lifted with the same
independent exponential variable. The resulting law is at most
$M\widetilde\mu$. It falls outside $C$ with probability at most
$M\delta$; conditional on lying in $C$, its density relative to $\pi$
is at most $M/(1-M\delta)$. In particular, the conditional law is
$2M$-warm when $M\delta\le1/2$.

\subsubsection{Position updates and reflected dynamics}

We define the reflected process and its update kernels for a linear
potential on a bounded convex body. The formulas specify the Markov
operators used in the relaxation proof and the one-query update used in
the simulation.

Let $C\subset\mathbb R^d$ be a convex body and $b\in\mathbb R^d$. Write
$Z_C=\int_Ce^{-\langle b,z\rangle}dz$,
$\pi(dz)=Z_C^{-1}e^{-\langle b,z\rangle}\mathbf1_C(z)dz$,
$\gamma_d=N(0,I_d)$, and $\widehat\pi=\pi\otimes\gamma_d$.
Write $H^1(C,\pi)$ for the functions in $L^2(\pi)$ whose weak spatial
gradient belongs to $L^2(\pi;\mathbb R^d)$.
A phase state is $(z,v)\in C\times\mathbb R^d$. Between reflections
and updates, its motion satisfies $z(t)=z(0)+tv(0)-t^2b/2$ and
$v(t)=v(0)-tb$.
At a regular boundary point $p$, let $\nu(p)$ be the outward unit normal
and set $\mathcal R_pv=v-2\langle v,\nu(p)\rangle\nu(p)$.
A contact at time $t$ leaves the position unchanged and sets
$v(t+)=\mathcal R_pv(t-)$.

Let $\Psi_t$ be this reflected flow, defined almost everywhere with
respect to $\widehat\pi$, and set $U_tf=f\circ\Psi_t$. The construction
and invariance of the flow, including for nonsmooth $C$, are justified
in the proof of \Cref{prop:linear-potential-relaxation}. Define its
generator $B$ and domain $\operatorname{dom}(B)$ by the strong $L^2(\widehat\pi)$ limit
in \eqref{eq:billiard-generator-domain}, with this flow and invariant
measure. For smooth $f$ with compact velocity support satisfying
$f(p,v)=f(p,\mathcal R_pv)$ at every regular boundary point $p$,
we have $Bf=v\cdot\nabla_zf-b\cdot\nabla_vf$. The reflection condition
makes the two traces along a reflected trajectory agree.

Fix $\lambda,h>0$ and $0\le a<1$. Let
$\varphi_h(u)=(2\pi h)^{-d/2}e^{-\|u\|_2^2/(2h)}$.
For $z\in C$ and a Borel set $A\subseteq C$, define
\begin{equation}\label{eq:linear-position-kernel}
 R_h(z,A)=\int_{\mathbb R^d}\varphi_h(y-z)
 \frac{\displaystyle\int_A\varphi_h(z'-y+hb)\,dz'}
      {\displaystyle\int_C\varphi_h(u-y+hb)\,du}\,dy.
\end{equation}
The denominator is positive because $C$ has nonempty interior.
Equivalently, draw $Y\sim N(z,hI_d)$ and then draw
$Z'\sim N(Y-hb,hI_d)$ conditioned on $Z'\in C$.
This is the shifted conditional Gaussian update of
Kook and Vempala~\cite[Algorithm~1]{KV25}; its uniform-target case is
the In-and-Out update~\cite{KVZ24}.
On phase-space functions its operator action is
$(R_hf)(z,v)=\int_C f(z',v)R_h(z,dz')$; the position-space action
is obtained by omitting $v$.

Define $Q_a$ as in the uniform case, with dimension $d$: it leaves $z$
fixed and replaces $v$ by $av+\sqrt{1-a^2}G$, where $G\sim\gamma_d$.
Thus $Q_a$ and $R_h$ commute. Set $J_h=Q_aR_h$ and
$\gamma=\lambda(1-a)$. Apply the independent Poisson construction in
\eqref{eq:poisson-semigroup} using this flow, rate $\lambda$, and update
kernel $J_h$. Its semigroup is denoted by $\mathcal P_t$, and its
generator is $\mathcal G=B+\lambda(J_h-I)$ on
$\operatorname{dom}(\mathcal G)=\operatorname{dom}(B)$.

For the implementation, put
$q_h(z,z')=\varphi_{2h}(z'-z+hb)$. For a Borel set $A\subseteq C$, define
\begin{equation}\label{eq:linear-one-query-kernel}
 \widehat R_h(z,A)
 =\int_A q_h(z,z')\,dz'
  +\mathbf1_A(z)\left(1-\int_Cq_h(z,u)\,du\right).
\end{equation}
This kernel proposes $N(z-hb,2hI_d)$, accepts a proposal in $C$, and
otherwise returns $z$. It is also obtained by trying the second draw
in the conditional Gaussian update once and retaining the original
position after a rejection. Its operator action is defined by the same
integration convention as $R_h$; on phase space it leaves the velocity
fixed.

Let $A_b=\int_{\partial C}e^{-\langle b,p\rangle}dS(p)$, where $dS$
denotes surface measure. The next estimates justify both the exact
update and its one-query approximation. Positivity means nonnegativity
of the quadratic form of the self-adjoint operator.

\begin{lemma}[Gaussian position updates for a linear potential]
\label{lem:linear-position-update}
The kernels $R_h$ and $\widehat R_h$ are positive self-adjoint Markov
operators on $L^2(\pi)$. The exact update satisfies
$\langle f,(I-R_h)f\rangle\le h\int_C\|\nabla f\|_2^2d\pi$
for every $f\in H^1(C,\pi)$. For $z,z'\in C$, the transition laws satisfy
$\|R_h(z,\cdot)-R_h(z',\cdot)\|_{\rm TV}\le \|z-z'\|_2/(2\sqrt h)$.
Their outputs can be coupled so that the mean disagreement probability
under $\pi$ is at most
$(A_b/Z_C)(\sqrt{h/\pi}+h\|b\|_2)$. An $M$-warm initial law multiplies
this bound by at most $M$. The kernel $\widehat R_h$ uses one membership
query to $C$.
\end{lemma}
\begin{proof}
For the joint law $Z\sim\pi$, $Y|Z\sim N(Z,hI_d)$, completing the
square gives $Z|Y=y\sim N(y-hb,hI_d)|_C$. Thus $R_h$ is a conditional
resampling operator. Its positivity, self-adjointness, Dirichlet-form bound,
and total variation estimate follow as in
\Cref{lem:position-update-properties}: use the conditional Poincar\'e
constant $h$ and compare the first Gaussian draws.

The proposal density $q_h(z,z')$ of $\widehat R_h$ satisfies
$e^{-b\cdot z}q_h(z,z')=e^{-b\cdot z'}q_h(z',z)$, proving reversibility.
Under the isometry $f(z)\mapsto Z_C^{-1/2}e^{-\langle b,z\rangle/2}f(z)$
from $L^2(\pi)$ to Lebesgue $L^2(C)$, the off-diagonal kernel becomes
$e^{-h\|b\|_2^2/4}\varphi_{2h}(z'-z)$. Gaussian convolution
and the holding term have nonnegative quadratic forms, proving positivity.

Couple the two updates by their first proposal. They agree unless
$z+W\notin C$, where $W\sim N(-hb,2hI_d)$. For a fixed displacement
$w$, parameterize each exiting segment by its first boundary point $p$.
Writing its initial point as $p-tw$, $0\le t\le1$, gives
\[
 \int_{\{z\in C:z+w\notin C\}}e^{-b\cdot z}dz
 \le\int_{\partial C}e^{-b\cdot p}(w\cdot\nu(p))_+
          \int_0^1e^{t b\cdot w}dt\,dS(p).
\]
Tilting the law of $W$ by $e^{t b\cdot W}$ gives a normalizer
$e^{-ht(1-t)\|b\|_2^2}\le1$ and mean $h(2t-1)b$, with covariance $2hI_d$.
Therefore
$\mathbb E[(W\cdot\nu)_+e^{t b\cdot W}]
\le\sqrt{h/\pi}+h|b\cdot\nu|$ for every unit vector $\nu$.
Integrating proves the stationary error bound. Domination of the initial
law gives the factor $M$.
\end{proof}

The added position updates preserve the square-root relaxation bound.
We state the estimate for nonsmooth convex bodies as well, because a
truncated epigraph need not have smooth boundary.
\begin{proposition}[Relaxation under a linear potential]
\label{prop:linear-potential-relaxation}
Suppose $C_P(\pi)\le P$. Let
$\gamma=P^{-1/2}$, $\lambda\ge\gamma$, $a=1-\gamma/\lambda$, and
$0<h\le\sqrt P/\lambda$. The reflected process with updates $J_h$ defines
a Markov semigroup $(\mathcal P_t)_{t\ge0}$ on $L^2(\widehat\pi)$ with
$\widehat\pi\mathcal P_t=\widehat\pi$. For every $t\ge0$ and
$f\in L^2(\widehat\pi)$ with $\int f\,d\widehat\pi=0$,
$\|\mathcal P_t f\|_2\le C e^{-ct/\sqrt P}\|f\|_2$.
The constants are universal. Let $(Z_t,V_t)$ denote the phase state.
From an $M$-warm initial law,
$\|\mathcal L(Z_t,V_t)-\widehat\pi\|_{\TV}\le\varepsilon$
for $t\ge C\sqrt P\log(2M/\varepsilon)$.
\end{proposition}
\begin{proof}
First suppose $C$ is smooth. The Hamiltonian generator is
$B=v\cdot\nabla_z-b\cdot\nabla_v$, and the full generator is
$B+\lambda(J_h-I)$. The spatial operator in the divergence lemma of
Eberle and L\"orler~\cite[Theorem~5]{EL26:journal} is $\Delta-b\cdot\nabla$.
The potential has zero Hessian, so its curvature parameter is zero.
The divergence lemma and the argument of
\cite[Lemma~14 and Theorem~15]{EL26:journal} therefore apply with the same bounds
used in \Cref{prop:mixed-refresh-relaxation}.

Here are the operator estimates needed for that application. Put
$\mathcal D(f)=\lambda\langle f,(I-J_h)f\rangle$ and define
$(\Pi f)(z,v)=\int f(z,w)\,d\gamma_d(w)$, extended pointwise in time.
The commuting positive contractions $Q_a,R_h$
give $\|f-\Pi f\|_2^2\le\mathcal D(f)/\gamma$ and
$I-J_h\preceq(I-Q_a)+(I-R_h)$. On a time interval of length $\tau$, the
divergence lemma supplies $\phi=p+v\cdot\nabla g$ with the same
space-time estimates as in the proof of \Cref{prop:mixed-refresh-relaxation}.
The same domain argument applies, now with
$B\phi=v\cdot\nabla p+v^{\mathsf T}\nabla^2g\,v-b\cdot\nabla g$.
For a fixed $t\ge0$, put $u_s=\mathcal P_{t+s}f$ on $0\le s\le\tau$.
The divergence lemma is applied to $\Pi u$, and all norms and inner products below include
integration over this interval.
\Cref{lem:linear-position-update} gives
\[
 \int_0^\tau\mathcal D(\phi_s)ds
 \le C\bigl[\gamma(\tau^2+P)+\lambda h(1+P/\tau^2)\bigr]\|\Pi u\|_2^2.
\]
The divergence identity, Cauchy--Schwarz for this quadratic form, and
the energy identity consequently yield
\[
 \int_t^{t+\tau}\|\mathcal P_s f\|_2^2ds
 \le C\left[\gamma(\tau^2+P)
       +(\lambda h+\gamma^{-1})(1+P/\tau^2)\right]
       \int_t^{t+\tau}\mathcal D(\mathcal P_s f)ds.
\]
Take $\tau=\sqrt P$ and use
$\frac{d}{ds}\|\mathcal P_s f\|_2^2=-2\mathcal D(\mathcal P_s f)$.
The same time-averaging argument as in
\Cref{prop:mixed-refresh-relaxation} proves the claimed decay.

We give the details needed to pass to a nonsmooth body. Let
$\partial_{\rm reg}C$ be the boundary points with a unique normal.
The states immediately after reflection form
$\Sigma=\{(p,v):p\in\partial_{\rm reg}C,\ \langle v,\nu(p)\rangle<0\}$.
On $\Sigma$ define the finite measure
\[
 d\sigma(p,v)=Z_C^{-1}e^{-\langle b,p\rangle}
 |\langle v,\nu(p)\rangle|\,dS(p)\,d\gamma_d(v).
\]
For $(p,v)\in\Sigma$, the next free-flight time is
$\tau(p,v)=\inf\{s>0:p+sv-s^2b/2\in\partial C\}$.
Writing $\tau=\tau(p,v)$ and $p'=p+\tau v-\tau^2b/2$, the collision map is
$(p,v)\mapsto(p',\mathcal R_{p'}(v-\tau b))$ wherever this next contact is regular.
The free flow preserves phase-space Lebesgue measure and the Hamiltonian
$\langle b,z\rangle+\|v\|_2^2/2$. Its map between boundary contacts therefore
preserves the corresponding flux measure; the boundary parametrization
has Jacobian $|\langle v,\nu(p)\rangle|$. Reflection preserves Gaussian
velocity density and absolute normal speed. Thus the collision map
preserves $\sigma$ and is invertible almost everywhere by time reversal.
Irregular contacts and tangent velocities have zero flux, by the area
formula and the fact that the irregular part of a convex boundary has
surface measure zero.

The time $\tau(p,v)$ is positive almost
everywhere and finite: a nonzero constant-force trajectory is an unbounded
parabola, and a nonzero zero-force trajectory is an unbounded line.
The boundary parametrization gives, for integrable $F$,
\[
 \int_\Sigma\int_0^{\tau(p,v)}
 F(p+sv-s^2b/2,v-sb)\,ds\,d\sigma(p,v)
 =\int F\,d\widehat\pi.
\]
In particular, $\int\tau\,d\sigma=1$. Apply the ergodic theorem to
the positive integrable function $\min\{\tau,1\}$ on the collision
map. Its forward and backward sums diverge almost everywhere. Hence
reflections do not accumulate in finite time. Concatenating the flights
defines an invariant Hamiltonian flow on $C$.

Choose smooth convex bodies $C_j$ converging to $C$ in Hausdorff distance.
Their probability measures with density proportional to
$e^{-\langle b,z\rangle}$ converge in total variation, and their
Poincar\'e constants converge to that of $\pi$. To justify the latter
assertion, use the variational characterization of the first nonzero
weighted Neumann eigenvalue. Convex domains converging to a body with
nonempty interior admit uniformly bounded Sobolev extension operators
on a common bounded neighborhood. The fixed weight is bounded above
and below there. Restricting smooth test functions proves one eigenvalue
inequality. Uniform extensions of normalized first eigenfunctions have
a strongly $L^2$-convergent subsequence; weak lower semicontinuity of
energy on compact subsets of the interior proves the other inequality.
Means and norms converge with the domains. This is the usual Neumann
spectral-continuity argument; see
\cite[Corollary~A.3 and its proof]{FL26:journal}.
For all sufficiently large $j$, the approximating Poincar\'e constants
are at most $2P$. Apply the averaged-energy estimate above with this
bound while keeping $\gamma=P^{-1/2}$ and the same $\lambda,h$.
It still gives relaxation of order $\sqrt P$, with universal constants.

Almost every trajectory has finitely many regular, nontangent contacts
on a bounded time interval. Near a regular boundary point, all nearby
supporting normals of $C_j$ converge to its unique normal. Thus contact
times and reflected velocities converge successively along such a
trajectory. Transversality locates each exit; the strictly inward
reflected velocity excludes additional nearby exits. The bounded constant
acceleration preserves these local conclusions. For convergent starting
states, the joint position and velocity update laws converge in total
variation and have densities on phase space. Coupling the finitely many
updates on a bounded time interval therefore gives convergence of the
semigroups on bounded continuous functions. Center each such function by
its mean under the approximating invariant law; these means also converge.
Dominated convergence passes the $L^2$ estimate to the limit, and density
and contraction extend it to all of $L^2$. Finally, the adjoint has the
same operator norm. For an initial phase law
$\rho=\mathcal L(Z_0,V_0)\le M\widehat\pi$,
$\|d\rho/d\widehat\pi-1\|_2^2\le M-1$, which gives the stated
total variation estimate.
\end{proof}

\subsubsection{Membership-query simulation}

We first bound short flights and visits to the upper boundary. We then
estimate boundary normals and simulate the trajectory between updates.
The probabilistic estimates use the invariant boundary measure and
therefore apply without independence of successive contacts.

For this paragraph let $rB_2^d\subseteq C$, and suppose the known number
$A\ge0$ satisfies $-\langle b,z\rangle\le A$ on $C$. Retain the
notations $Z_C,A_b$ from above, and write
$\beta_b=A_b/(\sqrt{2\pi}Z_C)$ for the stationary collision rate.
For a unit vector $u$, let $P_{u^\perp}$ denote orthogonal projection
onto $u^\perp$, and let $\ell_u(y)$ be the length of
$C\cap(y+\mathbb Ru)$ for $y\in P_{u^\perp}C$.

The contained ball controls the weighted surface area and the proportion
of short chords in the boundary measure.
\begin{lemma}[Weighted short chords]\label{lem:weighted-short-chords}
We have $A_b/Z_C\le(d+A)/r$. The stationary rate of contacts whose
straight flight, started with the reflected velocity and with no further
updates, has duration less than $h$ is at most $\alpha_bh$, where
\[
 \alpha_b=\frac{C(d+A)\sqrt d}{r}\,\beta_b
 \le\frac{C(d+A)^2\sqrt d}{r^2}.
\]
This statement concerns the straight chord associated to each contact,
even when the actual process has acceleration $-b$.
\end{lemma}
\begin{proof}
The divergence theorem and $\langle z,\nu(z)\rangle\ge r$ give
\[
 rA_b\le\int_{\partial C}\langle z,\nu(z)\rangle e^{-b\cdot z}dS(z)
 =\int_C(d-b\cdot z)e^{-b\cdot z}dz\le(d+A)Z_C.
\]
Fix a direction $u$ with $b\cdot u\ge0$, put $D=P_{u^\perp}C$, and
write the entry point of its chord as $p_u(y)=y+a_u(y)u$. The function
$a_u$ is convex, so $\phi_u(y)=b\cdot p_u(y)$ is convex. The chord
through zero contains $[-r,r]u$, hence $\phi_u(0)\le0$, while
$\phi_u\ge-A$. For $0<\delta<1$, convexity therefore gives
$e^{-\phi_u((1-\delta)y)}\ge e^{-A\delta}e^{-\phi_u(y)}$.
For $y\in(1-\delta)D$, we have $\ell_u(y)\ge2r\delta$.
Writing $W_u=\int_De^{-\phi_u(y)}dy$, we obtain
\[
 \int_{\{y:\ell_u(y)<a\}}e^{-\phi_u(y)}dy
 \le\min\{1,C(d+A)a/r\}\,W_u.
\]
For $a<2r$, set $\delta=a/(2r)$ and use
$1-(1-\delta)^{d-1}e^{-A\delta}\le(d-1+A)\delta$; larger $a$ are
immediate. If $b\cdot u<0$, reverse the chord. Its entry point then
has smaller potential and larger weight. Thus the contribution of these
directions is bounded by that of the directions already considered.

For velocity $v=qu$, where $q=\|v\|_2$, boundary flux in chord coordinates
is $q e^{-\phi_u(y)}dy/Z_C$. Gaussian speed and direction are independent,
and $\mathbb E q^2/\mathbb E q\le C\sqrt d$. Apply the last estimate
with $a=qh$, integrate, and compare with the total rate $\beta_b$.
\end{proof}

For an epigraph and downward acceleration, a short curved flight can be
compared with a straight chord after a small change in velocity.
\begin{lemma}[Short parabolic flights]\label{lem:short-parabolic-flights}
Suppose $b=ce_d$, $c>0$, and
$C=\{(x,s):x\in D_0,\ f(x)\le s\le H\}$, where $D_0$ is a convex
body and $f$ is convex. Use the parameters $\alpha_b,\beta_b$ above.
For $B\ge1$ and $h>0$ satisfying $h\|b\|_2(B+1)\le c_0$, with a
universal $c_0>0$, the stationary rate of contacts with speed at most $B$ and a subsequent
uninterrupted parabolic flight of duration less than $h$ is at most
$C\alpha_bh+C\beta_b\|b\|_2^2h^2$. If $B\ge C\sqrt d$, the rate of
contacts with speed greater than $B$ is at most $C\beta_b e^{-c_1B^2}$.
\end{lemma}
\begin{proof}
Consider a regular contact $p$ and its velocity $v$ immediately after
reflection. If the parabolic path reaches the top by time $h$, its
unaccelerated straight path reaches that top at least as early. Their
horizontal coordinates agree, so a crossing of $\partial D_0$ also
gives a straight exit by time $h$. These cases are counted by
\Cref{lem:weighted-short-chords}.

For a crossing of the bottom, put $w=v-hb/2$. For $0\le t\le h$,
$p+tw$ has the same horizontal coordinate as the parabolic path and a
no larger height. It is therefore below the epigraph whenever that path
is. Exclude contacts with $|\langle v,\nu(p)\rangle|\le \|b\|_2h$;
their rate is at most $C\beta_b\|b\|_2^2h^2$ by the Gaussian normal-flux
calculation. At the remaining contacts, $w$ points strictly into $C$ and
$|\langle v,\nu(p)\rangle|\le2|\langle w,\nu(p)\rangle|$.
On $\|v\|_2\le B$, the condition on $h$ bounds the ratio of Gaussian
densities $\gamma_d(w+hb/2)/\gamma_d(w)$ by a universal constant.
Changing variables from $v$ to $w$ and applying the straight-chord
bound proves the first assertion. At a stationary collision, the radial
Gaussian velocity density is weighted by speed. Its tail above
$B\ge C\sqrt d$ is at most $Ce^{-c_1B^2}$, proving the second assertion.
\end{proof}

We also keep the exact trajectory below the upper truncation boundary.
The following statement returns to the epigraph $C$ and parameters
$H,c_*$ of \Cref{lem:truncated-epigraph-parameters}.
\begin{lemma}[Probability of approaching the top]
\label{lem:epigraph-top-avoidance}
Consider the reflected constant-force process on $C$, with Poisson
updates of rate $\lambda$ preserving $\widehat\pi$. Suppose its initial
law is $M_0$-warm, and let $S_t$ denote its height at time $t$.
If $H-1\ge c_*\sqrt n+1$, then
\[
 \mathbb P\left\{\sup_{0\le t\le T}S_t\ge H-1\right\}
 \le CM_0(1+\lambda T+\sqrt nT)2^n e^{-\sqrt n(H-1)/2}.
\]
The same bound with an additional term $J$ inside the parentheses holds
if $J$ deterministic invariant updates are added. In particular, this
probability is at most $\eta$ when
\[
 H\ge1+\frac{2n\log2+
  2\log(CM_0(1+\lambda T+J+\sqrt nT)/\eta)}{\sqrt n},
\]
in addition to the previous lower bounds on $H$.
\end{lemma}
\begin{proof}
Put $U=V-a_V$ and $Z_U=\int e^{-U}$, where $a_V=\inf V$.
The random variable $W=U(X)+F$ has density
$p_W(q)=e^{-q}|\{U\le q\}|/Z_U\le2^ne^{-q/2}$.
Indeed, $|\{U\le q\}|\le e^{q/2}\int e^{-U/2}\le e^{q/2}2^nZ_U$.
Consequently, $\mathbb P(W\ge q)\le2^{n+1}e^{-q/2}$. Since
$S=(a_V+W)/\sqrt n$ and $a_V\le0$, the untruncated height density at a
nonnegative height $s$ is at most $\sqrt n2^ne^{-\sqrt ns/2}$.
Conditioning on $C$ increases these bounds by at most
$(1-\delta)^{-1}\le2$.

At stationarity the upward flux across a horizontal plane is its height
density times $\mathbb E(G_1)_+=1/\sqrt{2\pi}$. Thus the expected number
of continuous upward crossings of height $H-1$ before time $T$ is at
most $C\sqrt nT2^ne^{-\sqrt n(H-1)/2}$. Add the initial upper tail,
the expected number of Poisson updates landing above this plane, and
the $J$ deterministic update tails. Every visit to the upper region is
accounted for by one of these events, up to a null event of tangency.
Domination of the entire path law by $M_0$ proves the result.
\end{proof}

The normal estimator uses the gauge of a body containing a known ball.
A short radial displacement preserves its gradient. We choose the
displacement small enough that the linear density changes by a bounded
factor, and then apply the integrated second-difference estimate.

Let $B_2^d\subseteq C\subseteq RB_2^d$, let
$g(z)=\inf\{t>0:z\in tC\}$, and put $w(z)=e^{-\langle b,z\rangle}$.
Suppose a known $L_b$ bounds $|\langle b,z\rangle|$ on $C$. Set
$q_0=[4(d+L_b+1)]^{-1}$ and $A_b=\int_{\partial C}w\,dS$.
We use $\|z\|_1=\sum_i|z_i|$ and write $\partial g(z)$ for the convex
subdifferential of $g$ at $z$.

The next event bounds all subgradients near a random boundary point.
It is this local estimate that permits approximate collision locations.
\begin{lemma}[Local variation of a weighted boundary normal]
\label{lem:weighted-finite-scale-normal}
Draw $X$ from $w(x)dS(x)/A_b$ on $\partial C$ and independently draw
$S\sim\operatorname{Unif}[1,1+q_0]$. Put $Y=SX$ and $q_j=2^{-j}q_0$.
For every integer $J\ge0$ and $L>0$, except on an event of probability
at most $C(J+1)/(q_0L)$, the bounds
\[
 \sum_{i=1}^d\bigl(g(Y+q_je_i)+g(Y-q_je_i)-2g(Y)\bigr)
 \le Lq_j^2
\]
hold for every $0\le j\le J$.
On this event, every $z$ with $\|z-Y\|_1\le q_0/2$ and every
$p\in\partial g(z)$ satisfy
$\|p-\nabla g(Y)\|_2\le8dL\|z-Y\|_2+2\sqrt dLq_J$.
We have $\nabla g(Y)=\nabla g(X)$ almost surely.
\end{lemma}
\begin{proof}
The radial Jacobian is $s^{d-1}\langle x,\nu(x)\rangle$.
Since $\langle x,\nu(x)\rangle\ge1$ and
$w(x)/w(sx)\le e^{q_0L_b}$, the density of $Y$ is at most
$Cw(y)/(q_0A_b)$ on $D=(1+q_0)C\setminus C$.

For a smooth convex $1$-Lipschitz function $f$, the integral formula
for second differences and the nonnegative diagonal Hessian entries give,
for $0<q\le q_0$,
\[
 \int_Dw(y)\sum_{i=1}^d
 [f(y+qe_i)+f(y-qe_i)-2f(y)]dy
 \le Cq^2\int_{(1+q_0+q)C}w\Delta f.
\]
Here translating the weight by a coordinate displacement changes it by
at most $e^{\|b\|_2q}\le e$, because $\|b\|_2\le L_b$. Integration by parts
bounds the last integral by
$\int_{\partial((1+q_0+q)C)}w\,dS+\|b\|_2\int_{(1+q_0+q)C}w$.
For any convex body $D'$ and $b\ne0$, the divergence theorem gives
$\|b\|_2\int_{D'}w=-\int_{\partial D'}w\langle b/\|b\|_2,\nu\rangle dS
\le\int_{\partial D'}w\,dS$; the case $b=0$ is immediate.
Dilation by $1+q_0+q$ changes weighted surface area by a factor at most
$(1+2q_0)^{d-1}e^{2q_0L_b}\le C$. The second-difference integral is
therefore at most $Cq^2A_b$.

Apply this estimate to smooth convolutions of $g$ and pass to the limit
by Fatou's lemma. The density bound for $Y$ gives expected second-difference
sum at most $Cq^2/q_0$. Markov's inequality and a union bound over the
scales prove the event. The deterministic convex interpolation and
subgradient argument in \Cref{lem:finite-scale-gauge-regularity} proves
the claimed local bound. Homogeneity gives $\nabla g(SX)=\nabla g(X)$.
\end{proof}

To use this estimate at an approximate contact $\widetilde X$, form central
differences at $\widetilde Y=S\widetilde X$ with step $q\ge q_J$ and
gauge-value error at most $\xi$. When the coordinate segments stay within
$\|z-Y\|_1\le q_0/2$, the gradient error is at most
\(Cd^{3/2}L(\|\widetilde Y-Y\|_2+q+q_J)+\sqrt d\xi/q.\)
Normalizing multiplies this error by at most $2R$, because
$\|\nabla g(X)\|_2\ge1/R$. Choosing $\xi\le q^2$ gives an error bounded
by a polynomial in $d,R,L$ times $\|\widetilde X-X\|_2+q$. Each of the
$2d$ gauge values is obtained by radial binary search with
$O(\log(2+R/\xi))$ membership queries. These are the same computations
as in the uniform normal estimator.

For the epigraph representation, horizontal motion is straight and
vertical acceleration is constant. The bottom and side constraints have
an interval of feasible times. We test the upper constraint separately,
so bisection always finds the first exit.

Translate the center of the contained ball to zero and write
$C=\{(x,s):x\in D_0,\ f(x)\le s\le H\}$, with
$B(0,1)\subseteq C\subseteq B(0,R)$, where $D_0$ and $f$ are convex.
For $c>0$, a flight from $(z,v)$ is
$z(t)=z+tv-ct^2e_d/2$, $v(t)=v-ct e_d$. Membership in $C$ is checked
using the available test for $D_0$, one evaluation of $f$, and the height
bound. The function
$t\mapsto f(x+tv_x)-s-tv_s+ct^2/2$ is convex. Thus the feasible times
for the bottom and side constraints form an interval.

The following deterministic statement specifies the precision and query
cost on the geometric events proved above.
\begin{lemma}[Finite-query simulation with constant acceleration]
\label{lem:finite-query-parabolic-segment}
Fix $\epsilon>0$ and a segment of duration $\ell\le H_0$, with
$H_0\ge1$, starting at $(z,v)$ with $z\in\operatorname{int}C$.
Let $v(t)$ denote its velocity along the reflected segment. Suppose it has
$N\le N_{\max}$ reflections and $\sup_{0\le t\le\ell}\|v(t)\|_2\le B$, where $B\ge1$,
and every incoming normal speed is at least $\alpha>0$. Assume that
successive reflections are separated by at least $\tau>0$, the terminal
position has distance at least $\rho>0$ from $\partial C$, and the entire
segment stays below height $H-1$.

At each true contact draw an independent radial variable uniformly on
$[1,1+q_0]$. Suppose the event in
\Cref{lem:weighted-finite-scale-normal} holds for these variables, with
constant $L$ and scales $q_j=2^{-j}q_0$, $0\le j\le J$. Assume
$\alpha,\tau,\rho,q_0\le1$, and set
\[
 D=\left(\frac{C_0dR(1+L)(1+B+cH_0)(1+H_0)}
                  {\alpha\tau\rho q_0}\right)^{20},
\]
where $C_0$ is a sufficiently large universal constant.
Put $\epsilon_0=\min\{\epsilon,D^{-2}\}$. If
$q_J\le\epsilon_0q_0/D^{2N_{\max}+10}$, an adaptive membership
algorithm returns the terminal position and velocity with error at most
$\epsilon$ in each. It uses one query if $N=0$, and otherwise
$O(1+dN^2\log(2DR/(\epsilon_0q_0)))$ queries. Without the geometric
conditions it still terminates under a fixed reflection cap and may
declare failure.
On the geometric event, the work through a cap of $N_{\max}$ remains
$O(1+dN_{\max}^2\log(2DR/(\epsilon_0q_0)))$ even if the exact segment
has more than $N_{\max}$ reflections.
\end{lemma}
\begin{proof}
For a trial budget $b=1,2,4,\ldots$, choose an available dyadic scale $q$
with $\epsilon_0q_0/(2D^{b+5})<q\le\epsilon_0q_0/D^{b+5}$.
All required trials have $b\le2N_{\max}$, so $q\ge q_J$. Restart
each trial from the original phase point and reuse the radial variable
assigned to each contact number.

At a current interior point $\widetilde z$ with velocity $\widetilde v$
and remaining time $t$, compute the first time the quadratic height
reaches $H-1/2$, taking that time to be infinity if it is never reached.
Declare failure if the current height is already at least $H-1/2$.
Let $t_*$ be the smaller of the crossing time and $t$, and query the
parabolic point at time $t_*$. If it is inside and $t_*=t$, return that
point with velocity $\widetilde v-ct e_d$. If it is inside and $t_*<t$,
declare failure. Otherwise bisect the time interval until the inside and
outside times differ by at most $q$. Below $H-1/2$ the top constraint
is inactive, and the interval property above proves that this is the
first exit. Write $t_-$ for the inside time and $z_-$ for its position.

If the reflection budget has been used, restart with twice that budget.
Otherwise estimate the normal at $z_-$ by the gauge differences described
above, using step $q$ and gauge accuracy $q^2$. Declare failure if the
estimated gradient vanishes or a required operation is undefined.
Reflect $\widetilde v-ct_-e_d$ in this normal, replace the position by
$(1-q)z_-$, and subtract $t_-$ from the remaining time. Convexity gives
$(1-q)z_-+qB(0,1)\subseteq C$, so the new point is interior.
Reflection preserves instantaneous speed. Every numerical trial therefore
has speed at most $\|v\|_2+cH_0\le B+cH_0$, giving finite search bounds
even when the geometric conditions fail. Stop at the prescribed final
budget if another reflection is required.

We verify stability. The local subgradient estimate, integrated from a
true contact $p$, gives
\[
 0\le g(p+w)-1-\langle\nabla g(p),w\rangle
 \le CdL\|w\|_2^2+C\sqrt dLq_J\|w\|_2
\]
for sufficiently small $w$. Consider a true flight ending at
$p=p_0+t_0v_0-ct_0^2e_d/2$, with incoming velocity
$v_*=v_0-ct_0e_d$. If the numerical starting position, velocity, and
elapsed time have errors at most $e$, its spatial error at local time
$t_0$ is at most $(1+H_0)e$. Also
$\langle\nabla g(p),v_*\rangle
=\langle\nu(p),v_*\rangle/\langle p,\nu(p)\rangle\ge\alpha/R$.
At nearby times the unreflected true path is
$p\pm u v_*-cu^2e_d/2$. The supporting-plane lower bound after contact
and the quadratic upper bound before it imply that the numerical path
is inside at $t_0-u$ and outside at $t_0+u$, where
$u=CR(1+H_0)(e+q)/\alpha$. The definition of $D$ absorbs all quadratic
terms when $e+q\le D^{-2}$, ensures $u<\tau/2$ for noninitial flights,
and keeps the comparison within the normal-estimation neighborhood.

The upper truncation does not interfere. Before the true contact the
exact height is at most $H-1$, and the numerical discrepancy is at most
$(1+H_0)e$. Continuing for time $u$ moves by at most $Bu+cu^2/2$.
Both bounds are less than $1/4$ under the prescribed precision.
Therefore height $H-1/2$ is not reached before $t_0+u$. The interval
property shows that the bracket locates the first perturbed exit.
For the first flight the initial phase point is exact, so no lower
bound on its duration is needed.

Time bisection adds position error at most $(B+cH_0)q$ and velocity
error at most $cq$. The contraction adds at most $Rq$ to the position
error. The normal estimate and reflection formula consequently bound
each new position, velocity, and elapsed-time error by $D(e+q)$.
Starting from zero, after $j\le b$ contacts these errors are at most
$2D^jq$, which meets all preceding smallness conditions. The terminal
distance $\rho$ makes the last inside decision correct. The last flight
has duration at least $\rho/B$, so each contact bracket lies before the
numerical termination time. A trial with budget below $N$ cannot finish
early, and a trial with budget at least $N$ returns the required endpoint.

Each time bisection costs $O(\log(2+H_0/q))$ queries; each normal
estimate costs $O(d\log(2+R/q^2))$. A trial of budget $b$ therefore
costs $O(1+db^2\log(2DR/(\epsilon_0q_0)))$. Sum the geometric series
over the doubling budgets, whose first successful value is at most
$2N$. If $N=0$, the exact and computed curves agree and stay below
$H-1$, so the first endpoint query succeeds.
The same stability induction can be stopped at any reflection budget.
To certify that another reflection is required, it uses the first-exit
bracket at the next contact, independently of all subsequent contacts.
Hence the summed cost of the capped trials has the stated bound even
when the exact segment exceeds the cap.
\end{proof}

\subsubsection{The sampling theorem and sample sequences}

We now choose the numerical parameters and combine the preceding results.
The same truncated body and exact transition kernel are used for the
whole sample sequence. This permits the matrix concentration argument in
\Cref{subsec:upper-moments} to be applied to its logconcave position marginal.

The covariance bound determines the running time needed for mixing. The second
moment determines a spatial cutoff and affects the query count only
through logarithms.
\begin{theorem}[Logconcave sampling from a warm start]
\label{thm:logconcave-warm-membership-sampling}
Let $\mu(dx)\propto e^{-V(x)}dx$ be full-dimensional on $\mathbb R^n$,
where $V$ is proper, closed, convex, and accessible by evaluation.
Suppose $x_0\in\mathbb R^n$, $R,\kappa>0$, and $c_*\ge1$ are known and satisfy
\[
 \begin{aligned}
 B(x_0,1)&\subseteq\{V\le\inf V+c_*n\},\\
 \operatorname{Cov}(\mu)&\preceq\kappa I_n,\\
 \mathbb E_\mu\|X-x_0\|_2^2&\le R^2.
 \end{aligned}
\]
Given $X_0\sim\nu$ with $\nu\le M\mu$, where $M\ge1$, an integer $m\ge1$,
and $0<\varepsilon<1/2$, an algorithm returns
$\widehat X_1,\ldots,\widehat X_m\in\mathbb R^n$ using at most
\[
 O\!\left(c_*n^2\sqrt{\max\{1,\kappa\}}\,(M+m)
   \operatorname{polylog}\frac{c_*n(1+R)Mm}{\varepsilon}\right)
\]
potential evaluations in expectation. For $m=1$,
$\|\mathcal L(\widehat X_1)-\mu\|_{\TV}\le\varepsilon$.
The initial sample is supplied
to the algorithm and is not included in this query count.

More precisely, the construction specifies a bounded truncated epigraph
$C\subset\mathbb R^{n+1}$, a measure
$\pi(dz)\propto e^{-\sqrt n s}\mathbf1_C(z)dz$, where $z=(x,s)$,
and a Markov kernel $\mathcal K$ on positions and velocities. Writing
$\widehat\pi=\pi\otimes N(0,I_{n+1})$, we have
$\widehat\pi\mathcal K=\widehat\pi$. For $f\in L^2(\widehat\pi)$ with
$\int f\,d\widehat\pi=0$,
$\|\mathcal Kf\|_{L^2(\widehat\pi)}\le\|f\|_{L^2(\widehat\pi)}/4$.
For the exact epigraph law $\widetilde\mu$, we have
$\pi=\widetilde\mu(\cdot\mid C)$ and
$\widetilde\mu(C^c)\le\varepsilon/4$.
The marginal $\pi(\cdot\times\mathbb R)$ is logconcave and satisfies
$\|\pi(\cdot\times\mathbb R)-\mu\|_{\TV}\le\varepsilon/4$.
If $((X_i,S_i),V_i)_{i=1}^m$ is a Markov chain with kernel $\mathcal K$
and initial law $\widehat\pi$, then
\[
 \bigl\|\mathcal L(\widehat X_1,\ldots,\widehat X_m)
       -\mathcal L(X_1,\ldots,X_m)\bigr\|_{\TV}\le\varepsilon/2.
\]
For any prescribed $p\in(0,1/2)$, one may also require
$\widetilde\mu(C^c)\le p$ by running the construction with
$\min\{\varepsilon,p\}$ in place of $\varepsilon$.
\end{theorem}
\begin{proof}
Translate $x_0$ to zero and subtract $V(x_0)$, using one evaluation.
Put $\bar\kappa=\max\{1,\min\{\kappa,R^2\}\}$ and
$W=C_*+c_*+n+R+M+m+\varepsilon^{-1}$, where $C_*$ is a sufficiently
large universal constant. Set $\eta=W^{-200}$ and $\delta_C=W^{-1000}$. We use
\Cref{lem:truncated-epigraph-parameters} with error $\delta_C$.
The height $H$ will also meet \Cref{lem:epigraph-top-avoidance}.
The covariance bound for the truncated law is independent of $H$,
so the running time and update rate can be fixed first.

Let $P=C\bar\kappa\log(2n)\ge1$ be a Poincar\'e bound for $\pi$.
Choose $\lambda=Cc_*n^{3/2}\log^2W$, set $\gamma=P^{-1/2}$ and
$a=1-\gamma/\lambda$, and take $h=W^{-500}$.
These parameters satisfy $h\le\sqrt P/\lambda$.
\Cref{prop:linear-potential-relaxation} gives a fixed time
$t_*=C\sqrt P$ such that $\|\mathcal P_{t_*}f\|_2\le\|f\|_2/4$
for every mean-zero $f$. Define $\mathcal K=\mathcal P_{t_*}J_h$,
the time interval followed by one additional invariant update $J_h$.
Composition with $J_h$ preserves the contraction.

Lift the supplied point $X_0$ to height $(V(X_0)+F)/\sqrt n$ with independent
$F\sim\operatorname{Exp}(1)$ and add independent Gaussian velocity.
If the lifted point is outside $C$, return $m$ copies of $x_0$.
This event has probability at most $M\delta_C$; otherwise the initial
phase law is $2M$-warm. Run $k_0=\lceil C\log W\rceil$ preliminary
intervals, choosing $C$ so that
$\tfrac12\sqrt{2M-1}\,4^{-k_0}\le\eta$, and then record the positions
after the next $m$ intervals. The total time is $T=(k_0+m)t_*$.
The phase state after the preliminary intervals is within $\eta$ of
$\widehat\pi$. Coupling it to a stationary state and then using the same
kernel shows that the entire recorded exact sequence is within $\eta$
of the stationary sequence in the theorem.

Choose $H$ as the maximum of the bound in
\Cref{lem:truncated-epigraph-parameters} and
\[
 1+\frac{2n\log2+
  2\log(CM(1+\lambda T+k_0+m+\sqrt nT)/\eta)}{\sqrt n}.
\]
Then \Cref{lem:epigraph-top-avoidance} bounds visits to height $H-1$
by $C\eta$, including all deterministic endpoint updates. This choice
has $H=O(c_*\sqrt n+\log W/\sqrt n)$. The radius $R_C$ is a fixed
polynomial in $W$, and the explicit bounds of the truncation lemma give
$A_0=O(c_*n)$ and $L_b=O(c_*n+\log W)$. All further parameters use
these known bounds; no optimization over $C$ is required.

Partition the time intervals at their Poisson updates and deterministic
endpoints. Discard a partition with more than
$s_* =\lceil C(\lambda T+k_0+m+\log W)\rceil$ pieces or a piece longer
than $\ell_*=C\log W/\lambda$. Poisson and exponential tails make
these exceptions have probability at most $\eta$. The exponential
bound includes the first waiting time in each deterministic interval.
Choose $b_*$ to be the smallest power of two at least
$C\sqrt{Mc_*nT/\varepsilon}+2$. On the remaining partitions, use speed
bound $B=C\sqrt{n+\log W}+\sqrt n\ell_*$. Warm domination and Gaussian
tails bound excessive segment-start speeds by $\eta$. Acceleration
changes speed by at most $\sqrt n\ell_*$, and reflection preserves it.
The constants in $\lambda$ ensure $2\ell_*\sqrt n(B+1)\le c_0$,
so \Cref{lem:short-parabolic-flights} applies to the gaps below.

Let $N_\ell$ count reflections in an interval of length $\ell\le\ell_*$.
On the event of bounded speeds, $N_\ell\ge j$ forces at least
$(j-1)/2$ of the first $j-1$ gaps to have length at most
$2\ell/(j-1)$. The weighted short-flight estimate and Markov's
inequality therefore give, for every $j\ge2$,
\[
 \mathbb P(N_\ell\ge j,\ \text{bounded speeds})
 \le\frac{CMc_*^2n^{5/2}\ell^2}{(j-1)^2}.
\]
The additional quadratic term of \Cref{lem:short-parabolic-flights}
is absorbed by its linear term at these scales. Weighted flux also
gives $\mathbb P(N_\ell\ge1)\le CMc_*n\ell$. Summing the tail yields
\begin{equation}\label{eq:weighted-segment-moment}
 \mathbb E\!\left[\min\{N_\ell,b_*\}^2
              \mathbf1_{\{\text{bounded speeds}\}}\right]
 \le CM\bigl(c_*n\ell+c_*^2n^{5/2}\ell^2\log(2b_*)\bigr).
\end{equation}
Conditional on the partition, each exact interval starts from an
$M$-warm law, up to the harmless factor two. Adding deterministic cuts
can only decrease the sum of squared lengths. Thus
$\mathbb E\sum_i\ell_i^2\le2T/\lambda$, as in
\Cref{lem:poisson-interval-lengths}. The probability that any interval
exceeds the reflection cap is at most $\varepsilon/8$, apart from
the speed exception, and the total truncated second moment is
$O(Mc_*nT\log(2b_*))$.

Set the normal-speed cutoff to $\alpha=W^{-120}$, the minimum flight
time to $\tau=W^{-240}$, and the required endpoint distance to
$\rho=W^{-240}$. The respective exceptional probabilities for grazing
and short flights are at most $CMc_*nT\alpha^2$ and
$CMc_*^2n^{5/2}T\tau$, apart from the speed exception. For endpoints,
translate the contained-ball center to zero. Dilation gives
\[
 \pi\{\operatorname{dist}(Z,\partial C)<\rho\}
 \le1-(1-\rho)^{n+1}e^{-A_0\rho}
 \le(n+1+A_0)\rho.
\]
A union bound over the endpoints therefore costs at most
$CMs_*c_*n\rho$. All these bounds are $O(\eta)$.

Apply the local simulation to $C-z_0$, whose upper height is
$H-(z_0)_d$. For its normal estimates use the known initial scale
$q_0=[4(n+2+L_b)]^{-1}$ and $L=W^{1000}$. Set the duration bound
$H_0=\max\{1,T\}$ and the requested position and velocity error
$\delta_{\rm num}=W^{-600}$. Define $D$ by
\Cref{lem:finite-query-parabolic-segment}, with radius bound $R_C$,
and put $\delta_0=\min\{\delta_{\rm num},D^{-2}\}$.
Choose enough dyadic scales that
$q_{J_*}\le\delta_0q_0/D^{2b_*+10}$. Since all parameters are bounded
by fixed powers of $W$, one can take $J_*\le Cb_*\log W$.
Weighted boundary flux and \Cref{lem:weighted-finite-scale-normal}
bound the probability of a failed normal event by
$CMc_*nT(J_*+1)/(q_0L)\le C\eta$.

Simulate each interval with reflection budgets $1,2,4,\ldots,b_*$,
restarting the interval with twice the budget whenever a smaller
budget is exhausted. If the largest budget is exhausted, or a speed
test, height test, or required operation fails, stop and fill all
remaining output slots with $x_0$. At each prescribed update use
$Q_a\widehat R_h$ instead of $Q_aR_h$.
For $G\sim N(0,I_d)$, let
$e_h(z)=\mathbb P(z-hb+\sqrt{2h}G\notin C)$.
Gaussian translation gives
$|e_h(z)-e_h(z')|\le C\|z-z'\|_2/\sqrt h$. The first-proposal coupling
and the Gaussian smoothing estimates bound the endpoint-update error by
\(e_h(z)+C\delta_{\rm num}/\sqrt h +C\delta_{\rm num}/\sqrt{1-a^2}.\)
Under the exact warm law, the mean first term is at most
$CMc_*n(\sqrt h+h\sqrt n)$ by \Cref{lem:linear-position-update}.
Summing over at most $s_*$ updates gives $C\eta$ with our choices.
Couple the updates maximally while the phase states agree at interval
starts. The local simulation controls their endpoints on the stated
geometric events. All exceptional probabilities are $O(\eta)$ except
for the reflection cap, which costs at most $\varepsilon/8$.
Including the initial truncation event, the numerical sequence is within
$\varepsilon/4$ of the exact recorded sequence. The earlier stationary
coupling proves the joint-distribution assertion. The stationary position
marginal $\pi(\cdot\times\mathbb R)$ is within $\delta_C$ of $\mu$ in total
variation. Together with the joint comparison, this proves the
single-sample guarantee.

It remains to count queries, including the exceptional cases. On the
geometric event, interval $i$ costs
$Cn\min\{N_i,b_*\}^2\log W+O(1)$ evaluations, plus one per position
update. Apply \eqref{eq:weighted-segment-moment} to the preliminary
intervals to obtain expected work
$O(Mc_*n^2k_0t_*\operatorname{polylog}W)$.
At their endpoint, couple the exact state to $\widehat\pi$, with
disagreement probability at most $\eta$. The following exact intervals
then agree with a stationary trajectory except on that event. Applying
the same estimate with $M=1$ bounds their expected work while the
geometric conditions and coupling hold by
$O(c_*n^2mt_*\operatorname{polylog}W)$.

On every realization, the caps give at most
$Q_*\le Cn s_*b_*^2\log W\le W^{25}$ evaluations, after increasing
$C_*$. Geometric exceptions, failed endpoint couplings, and the failed
stationary coupling have total probability $O(\eta)$, so their additional
expected work is at most $C\eta Q_*\le1$. Exceeding the reflection cap
is already included in the truncated second moment. Discarded partitions
cost no simulation queries, and the initial lift costs one evaluation.
Substitute $t_*=O(\sqrt{\bar\kappa\log(2n)})$ and $k_0=O(\log W)$
to obtain the claimed bound. Undo the initial translation when returning
the first $n$ coordinates.
\end{proof}

The truncated position marginal is logconcave by the marginal theorem.
Its mean and covariance differ from those of $\mu$ by the bounds in
\Cref{lem:joint-conditioning-moments}, applied to the joint truncation
event. This is the moment comparison used for the recorded sample batches.

For a Gaussian restricted to a convex body, the general theorem requires
only membership queries. The known quadratic part is evaluated directly.
\begin{corollary}[A Gaussian restricted to a convex body]
\label{cor:gaussian-convex-body-sampling}
Suppose $B(0,r)\subseteq K\subseteq B(0,R)$ and $t>0$, and let
$\mu_t(dx)\propto e^{-t\|x\|_2^2/2}\mathbf1_K(x)dx$.
For $M\ge1$ and $0<\varepsilon<1/2$, given an initial sample from
$\nu\le M\mu_t$, an algorithm uses
\[
 O\!\left(Mn^2\max\{1,(tr^2)^{-1/2}\}
       \operatorname{polylog}
       \frac{n(1+R/r)M(1+tr^2)}{\varepsilon}\right)
\]
membership queries in expectation and returns $\widehat X$ with
$\|\mathcal L(\widehat X)-\mu_t\|_{\TV}\le\varepsilon$.
The guarantee for sample sequences in
\Cref{thm:logconcave-warm-membership-sampling} also applies.
If $t\ge2n/r^2$, rejecting draws from $N(0,t^{-1}I_n)$ outside $K$
gives an exact sample with expected membership cost at most two.
\end{corollary}
\begin{proof}
Brascamp--Lieb gives $\operatorname{Cov}(\mu_t)\preceq t^{-1}I_n$.
Set $r_* =\min\{r,\sqrt{n/t}\}$ and use coordinates $x=r_*y$.
We have $B_2^n\subseteq r_*^{-1}K$ and $tr_*^2\|y\|_2^2/2\le n/2$
for $y\in B_2^n$. For $X\sim\mu_t$,
$\operatorname{Cov}(X/r_*)\preceq(tr_*^2)^{-1}I_n$ and
$\mathbb E\|X/r_*\|_2^2\le(R/r_*)^2$.
Each potential evaluation uses one membership query. 
Apply the theorem and use
$\max\{1,(tr_*^2)^{-1/2}\}=\max\{1,(tr^2)^{-1/2}\}$;
$R/r_*=(R/r)\max\{1,\sqrt{tr^2/n}\}$ enters only logarithmically.
Finally, if $t\ge2n/r^2$, a Gaussian proposal has mean squared norm
$n/t\le r^2/2$, so it belongs to $K$ with probability at least $1/2$.
Conditioning by rejection gives the exact desired law.
\end{proof}

\end{document}